%% file: main.tex
\newcommand{\longversion}[1]{}
\newcommand{\shortversion}[1]{#1}

\documentclass[a4paper,UKenglish,cleveref, autoref, thm-restate]{lipics-v2021}

\title{Knapsack on Graphs with Relaxed Neighborhood Constraints} 

 \author{Palash Dey}{Indian Institute of Technology Kharagpur, India \and \url{https://cse.iitkgp.ac.in/~palash/} }{palash.dey@cse.iitkgp.ac.in}{https://orcid.org/0000-0003-0071-9464}{}

\author{Ashlesha Hota}{Indian Institute of Technology Kharagpur, India}{ashleshahota@gmail.com}{https://orcid.org/0009-0009-8805-4583}{}

\author{Sudeshna Kolay}{Indian Institute of Technology Kharagpur, India \and \url{https://cse.iitkgp.ac.in/~skolay/} }{skolay@cse.iitkgp.ac.in}{https://orcid.org/0000-0002-2975-4856}{}

\authorrunning{Anonymous authors} 

\ccsdesc[100]{\textcolor{red}{}} 

\keywords{Knapsack, neighboring constraints, parameterized complexity} 

\category{} 

\relatedversion{} 

\nolinenumbers 
\hideLIPIcs

\EventEditors{}
\EventNoEds{2}
\EventLongTitle{}
\EventShortTitle{}
\EventAcronym{}
\EventYear{}
\EventDate{}
\EventLocation{}
\EventLogo{}
\SeriesVolume{}
\ArticleNo{}

\input{preamble.tex}
\begin{document}

\maketitle

\input{abstract}

\newpage

\input{introduction}
\input{related_work}

\input{problem_definition}
\input{results_npc}

\input{parameterized_algo}

\input{conclusion}
\bibliography{references}
\input{appendix}
\end{document}

%% file: preamble.tex
\usepackage{longtable}
\newcommand{\tabitem}{~~\llap{\textbullet}~~}

\usepackage{longtable}
\usepackage{hyperref}
\usepackage{amssymb,amsmath}
\usepackage{color}
\usepackage{makecell}
\usepackage{mathtools}
\usepackage{bbm}
\usepackage{setspace}

\usepackage[usenames,svgnames,dvipsnames]{xcolor}

\usepackage{tikz}
\usetikzlibrary{arrows,positioning}

\newcommand{\defproblem}[3]{
  \vspace{1mm}
\begin{center}
\noindent\fbox{

  \begin{minipage}{\textwidth}
  \begin{tabular*}{\textwidth}{@{\extracolsep{\fill}}l} \textsc{\underline{#1}} \\ \end{tabular*}\vspace{1ex}
  {\bf{Input:}} #2  \\
  {\bf{Question:}} #3
  \end{minipage}

  }
\end{center}
  \vspace{1mm}
}

\newdimen\prevdp
\def\leftlabel#1{\noalign{\prevdp=\prevdepth
   \kern-\prevdp\nointerlineskip\vbox to0pt{\vss\hbox{\ensuremath{#1}}}\kern\prevdp}}

\usepackage{url}

\usepackage{enumerate}

\usepackage{xspace}
\newcommand{\NP}{\ensuremath{\mathsf{NP}}\xspace}
\newcommand{\NPC}{\ensuremath{\mathsf{NP}}\text{-complete}\xspace}
\newcommand{\NPH}{\ensuremath{\mathsf{NP}}\text{-hard}\xspace}

\newcommand{\el}{\ensuremath{\ell}\xspace}

\newcommand{\WOH}{\ensuremath{\mathsf{W[1]}}-hard\xspace}

\newcommand{\WTH}{\ensuremath{\mathsf{W[2]}}-hard\xspace}
\newcommand{\WO}{\ensuremath{\mathsf{W[1]}}\xspace}
\newcommand{\WT}{\ensuremath{\mathsf{W[2]}}\xspace}
\newcommand{\FPT}{\ensuremath{\mathsf{FPT}}\xspace}

\newcommand{\OPT}{\ensuremath{\mathsf{OPT}}\xspace}

\newcommand{\Pb}{\ensuremath{\mathsf{P}}\xspace}
\newcommand{\tw}{\text{tw}\xspace}

\newcommand{\clique}{{\sc Clique}\xspace}

\newcommand{\suk}{{\sc set-union knapsack problem}\xspace}
\newcommand{\sor}{{\sc Relaxed 1-Neighborhood Knapsack}\xspace}
\newcommand{\sand}{{\sc Relaxed All-Neighborhood Knapsack}\xspace}
\newcommand{\neighboror}{{\sc 1-Neighborhood Knapsack}\xspace}
\newcommand{\neighborand}{{\sc All-Neighborhood Knapsack}\xspace}
\newcommand{\hor}{{\sc 1-Neighborhood Knapsack}\xspace}
\newcommand{\hand}{{\sc All-Neighborhood Knapsack}\xspace}
\newcommand{\setc}{{\sc Set Cover}\xspace}
\let\oldlambda\lambda
\renewcommand{\lambda}{\ensuremath{\oldlambda}\xspace}
\let\oldalpha\alpha
\renewcommand{\alpha}{\ensuremath{\oldalpha}\xspace}
\let\oldDelta\Delta
\renewcommand{\Delta}{\ensuremath{\oldDelta}\xspace}
\newcommand{\mc}{{\sc Maximum Coverage}\xspace}
\newcommand{\clv}{{\sc Cutting $\el$ vertices}\xspace}

\newcommand{\yes}{{\sc yes}\xspace}

\newcommand{\kp}{\text{\sc Knapsack}\xspace}

\renewcommand{\AA}{\ensuremath{\mathcal A}\xspace}
\newcommand{\BB}{\ensuremath{\mathcal B}\xspace}
\newcommand{\CC}{\ensuremath{\mathcal C}\xspace}

\newcommand{\EE}{\ensuremath{\mathcal E}\xspace}
\newcommand{\FF}{\ensuremath{\mathcal F}\xspace}
\newcommand{\GG}{\ensuremath{\mathcal G}\xspace}

\newcommand{\II}{\ensuremath{\mathcal I}\xspace}
\newcommand{\JJ}{\ensuremath{\mathcal J}\xspace}

\newcommand{\OO}{\ensuremath{\mathcal O}\xspace}
\newcommand{\PP}{\ensuremath{\mathcal P}\xspace}

\renewcommand{\SS}{\ensuremath{\mathcal S}\xspace}

\newcommand{\UU}{\ensuremath{\mathcal U}\xspace}
\newcommand{\VV}{\ensuremath{\mathcal V}\xspace}

\newcommand{\XX}{\ensuremath{\mathcal X}\xspace}
\newcommand{\YY}{\ensuremath{\mathcal Y}\xspace}

\newcommand{\NB}{\ensuremath{\mathbb N}\xspace}

\usepackage{nicefrac}

\newcommand{\eps}{\ensuremath{\varepsilon}\xspace}
\renewcommand{\epsilon}{\eps}

\newcommand{\ignore}[1]{}

\newcommand{\pr}{\ensuremath{\prime}}

\renewcommand{\leq}{\leqslant}
\renewcommand{\geq}{\geqslant}

\usepackage{enumitem}
\setlist[enumerate]{labelwidth=!, labelindent=0pt}

\usepackage{rotating}
\usepackage{amssymb}
\usepackage{latexsym}
\usepackage{epsfig}
\usepackage{xcolor}
\usepackage{url}
\usepackage{fancyhdr}

\usepackage{caption}
\usepackage{subcaption}
\usepackage{array}
\usepackage{multirow}
\usepackage{enumerate}
\usepackage{pdflscape}
\usepackage{amsmath,graphicx,algorithm,algpseudocode,amsfonts}
\usepackage{epstopdf}
\usepackage{float}
\floatstyle{ruled}
\newfloat{Algorithms}{!thb}{lop}
\floatname{Algorithms}{Algorithm}

\allowdisplaybreaks[1]

\algnewcommand\algorithmicinput{\textbf{Input:}}
\algnewcommand\INPUT{\item[\algorithmicinput]}
\algnewcommand\algorithmicoutput{\textbf{Output:}}
\algnewcommand\OUTPUT{\item[\algorithmicoutput]}
\algnewcommand{\LineComment}[1]{\State \(\triangleright\) #1}

\usepackage{pifont}

\usepackage{mathtools}

\usepackage{algorithm}
\usepackage{algpseudocode}


\usepackage{comment}

\usepackage{cleveref}

\crefname{theorem}{Theorem}{\bf Theorems}
\crefname{observation}{Observation}{\bf Observations}
\crefname{corollary}{Corollary}{\bf Corollary}
\crefname{lemma}{Lemma}{\bf Lemmata}
\crefname{corollary}{Corollary}{\bf Corollaries}
\crefname{proposition}{Proposition}{\bf Propositions}
\crefname{definition}{Definition}{\bf Definitions}
\crefname{claim}{Claim}{\bf Claims}
\crefname{reductionrule}{Reduction rule}{\bf Reduction rules}
\usepackage{color,soul} 

%% file: abstract.tex
\begin{abstract}
In the knapsack problems with neighborhood constraints that were studied before, the input is a graph \GG on a set \VV of items, each item $v\in\VV$ has a weight $w_v$ and profit $p_v$, the size $s$ of the knapsack, and the demand $d$. The goal is to compute if there exists a feasible solution whose total weight is at most $s$ and total profit is at most $d$. Here, feasible solutions are all subsets \SS of the items such that, for every item in $\SS$, at least one of its neighbors in \GG is also in \SS for \hor, and all its neighbors in \GG are also in \SS for \hand~\cite{borradaile2012knapsack}. In this paper, we study a "relaxation" of the above problems. Specifically, we allow all possible subsets of items to be feasible solutions. However, only those items for which we pick at least one or all of its neighbor (out-neighbor for directed graph) contribute to profit whereas every item picked contribute to the weight; we call the corresponding problems \sor and \sand.

We fill some of the important research gaps that were present for \hor and \hand and perform an extensive study of the classical and parameterized complexity of \sor and \sand for both directed and undirected graphs. In particular, we show that both \sor and \sand are strongly \NPC even on undirected graphs. Regarding parameterized complexity, we show both \sor and \hor are \WTH parameterized by the size $s$ of the knapsack size. Interestingly, both \sand and \hand are \WOH parameterized by (i) knapsack size, \(s\) plus profit demand, \(d\) and also (ii) parameterized by solution size, $b$. For \sor and \hor, we present a randomized color-coding-based pseudo-\FPT algorithm with running time \(\OO^\star\left(e^b\cdot b^b\cdot 2^{b^2}\cdot {\sf min}\{s^2,d^2\}\right)\), parameterized by the solution size \(b\), and consequently by the demand \(d\). This indicates that 1-neighborhood constraints may be easier to handle algorithmically than all-neighborhood constraints in knapsack. We then consider the treewidth of the input graph (for directed graphs, treewidth of the undirected graph obtained by ignoring the directions) as our parameter and design an $\OO^\star\left(2^{\tw}\cdot {\sf min}\{s^2,d^2\}\right)$ time pseudo fixed-parameter tractable (\FPT) algorithm parameterized by treewidth, \(\tw\) for \hor and \hand, and an $\OO^\star\left(4^{\tw}\cdot {\sf min}\{s^2,d^2\}\right)$ time algorithm for both the relaxed variants. Note that, assuming $\Pb\ne\NP$, we cannot have an \FPT algorithm parameterized by the treewidth of the input graph for any of the four problems studied here since all of them are \NPC even for trees. Finally, we present an additive $1$ approximation for \sor when both the weight and profit of every vertex is $1$.

\end{abstract}

%% file: introduction.tex
\section{Introduction}
We study two new variants of the classical \kp problem. In our problems, we have a graph structure on the set of items of the \kp problem. Formally, the input to our problems is a directed graph $\GG = (\VV,\EE)$ on the set \VV of vertices (these are the items in the classical \kp problem) where every vertex $v \in \VV$ has a non-negative weight $w_v$ and a non-negative profit $p_v$. We are also given the size $s$ of the knapsack and a target profit or demand $d$, both of which are non-negative integers. The computational task is to find a subset $\SS\subseteq\VV$ whose profit is at least $d$ and weight at most $s$. The weight of \SS is the sum of the weights of the vertices in \SS. It is the profit of \SS where our two problems differ. In \sor, the profit of a subset \SS is the sum of the profits of the vertices $v$ in \SS for which at least one of its out-neighbors (if any exist in $\GG$) is included in \SS. Whereas, in \sand, the profit of \SS is the sum of the profits of the vertices $v$ in \SS for which all its neighbors (if any exist in $\GG$) are included in \SS.

To imagine a concrete potential application of our problems, let us think of a city planner who needs to choose a set of public projects to fund from an extensive list of proposals. Each project has a fixed cost. The natural interdependency among the project proposals makes a project valuable only if one/some/all of the other proposals on which it depends materialize. For example, a public gym is useful only if a residential complex is nearby. A natural objective of the planner is to maximize the total value of the projects funded, subject to some budget constraint.

Motivated by these applications, a line of research work, initiated by Borradaile~\cite{borradaile2012knapsack}, studies the \kp problem with neighborhood constraints and develops many interesting algorithms and other results. They call these problems \neighboror and \neighborand. In this setting, unlike ours, not all subsets of vertices constitute a valid solution. In \neighboror and \neighborand, the valid solutions are only those subsets of vertices which include, for every vertex in the subset, 
respectively at least one out-neighbor and all out-neighbors of that vertex. The goal in these problems is to check if there exists a solution which meets the target demand and weight limit.

The stringent requirement to pick at least one or all neighbors of any vertex which is already picked can sometime limit the number of possible solutions in a given instance. For example, there are only two valid solutions in any path graph for both \neighboror and \neighborand. However, in many real-life situations, we may want to choose some vertex so that another vertex is useful. Indeed, this could be the case in our earlier example of selecting a subset of public projects to fund. We remedy this drawback of the existing model by allowing all possible subsets of vertices to be valid solutions, and use the neighborhood constraint to decide which of the selected vertices contribute to the profit. Notice that this relaxation makes all possible subsets of the vertices valid solutions for our problems, thereby increasing, often exponentially, the number of valid solutions from the existing problems. Consequently, we often need different techniques to find the computational complexity of our problems.

\subsection{Contribution}

Our primary contribution is to propose two new problems, namely \sor and \sand. These problems eliminate the main drawback of \neighboror and \neighborand which severely restricted their use in practical applications. We perform a comprehensive study of these two proposed problems using the framework of parameterized complexity. Along the way, we also fill some of the important research gaps that were present in the literature on \neighboror and \neighborand. While \hor and \hand have been well-studied in classical settings, their parameterized aspects have remained unexplored. We fill this gap in this paper. We summarize our main results in \Cref{tab:contributions}.

\input{contributions_table}


In addition to the results listed in \Cref{tab:contributions}, we also show that the problems remain \NPC even on several restricted graph classes, including directed acyclic graphs (\Cref{sor-npc}), bipartite graphs (\Cref{sor-npc-bi}, \Cref{sand-npc}), trees (\Cref{sor-tree-npc}, \Cref{sand-tree-npc}), and general graphs, in both directed and undirected settings.


%


We now give a high-level overview of the techniques used in our results. For brevity, we refer to \neighboror and \neighborand as hard variants and \sor and \sand as soft variants. We show that the soft variants are strongly \NPC even for undirected graphs by reducing from \setc and \clv, respectively. However, \sor on directed graphs with unit weights and profits on vertices remains strongly \NPC. Interestingly, we are able to design an additive 1 approximation algorithm for \sor on directed graph when both the weight and profit of every vertex are 1. For the parameterized side, we prove \WT-hardness results for \sor and \hor parameterized by the size of the knapsack, and \WO-hardness for \sand and \hand parameterized by the size of the knapsack plus demand, and also parameterized by the size of a minimum solution. On algorithmic front, we observe that \sor and \sand admit pseudo polynomial time algorithm for trees. We design pseudo-\FPT algorithm for the directed \hor, \hand, \sor and \sand problems parameterized by the treewidth of the input graph (ignoring the directions of the edges for directed graphs), using dynamic programming over a nice tree decomposition of the input graph.\longversion{ For \hor and \hand, we guess the subset $\SS$ of vertices included in the solution and maintain a table of undominated (weight, profit) pairs for each node in a nice tree decomposition. The DP tracks which vertices in the bag are selected (i.e., part of $\SS$). For soft variants, we additionally track the subset $\PP \subseteq \SS$ of vertices contributing to profit, as selection no longer guarantees profit. The DP proceeds bottom-up through standard transitions, assuming a guessed vertex $v$ is in the solution to ensure feasibility. The final answer is computed at the root from states including $v$. The algorithm runs in time $\OO^\star\left(2^{\tw} \cdot {\sf min}\{s^2,d^2\}\right)$ for hard variants and $\OO^\star\left(4^{\tw} \cdot {\sf min}\{s^2,d^2\}\right)$ for soft variants.}

We next consider the size of a minimum solution as our parameter and develop a color-coding based algorithm. We color the edges randomly and guess a partition of the color set into subsets, each representing a connected component of the solution with at least two vertices in the underlying graph, ignoring the direction of edges. For each color subset, we enumerate all possible tree structures over those colors and perform dynamic programming to compute all undominated feasible (weight-profit) pairs for colorful trees matching that structure. Each DP state maintains a list of such pairs rooted at a specific vertex. We then merge the solutions from all components to obtain the global set of undominated weight-profit pairs. Finally, we output the best pair with weight at most the given budget. This algorithm can be derandomized using standard techniques involving splitters. The randomized version runs in time $\OO^\star\left(e^b\cdot b^b\cdot 2^{b^2}\cdot {\sf min}\{s^2,d^2\}\right)$, and the deterministic version adds only a polynomial overhead. Importantly, note that the size of any minimum solution is at most twice the demand, as each vertex of the solution that contributes to the profit may require at most one additional neighbor to become profitable. Hence, this algorithm also yields a pseudo-\FPT algorithm when parameterized by the demand.  We observe that undirected \sor with unit weights and profits on vertices is solvable in linear time using the algorithm proposed by Borradaile et al.\cite{borradaile2012knapsack}.

%% file: contributions_table.tex
\setlength\LTleft{-1cm}
\begin{longtable}{c|c}\hline 
    \textbf{Knapsack variant} & \textbf{Results}\\\hline\hline

    \makecell[c]{\textsc{1-Neighborhood}\\\textsc{Knapsack}} & \makecell[l]{
        \tabitem strongly \NPC even for uniform, directed graphs (Borradaile et al. \cite{borradaile2012knapsack}) \\
        \tabitem \textbf{W[2]-Hard parameterized by $s$ for directed graphs (\Cref{hor-woh})} \\
        \tabitem  $\OO\left(n+m\right)$ time algorithm for uniform, undirected graphs (Borradaile et al. \cite{borradaile2012knapsack}) \\
        \tabitem $\OO\left(n\cdot P + n^2\right)$ time algorithm for 
        trees (Goebbels et al. \cite{goebbels2022knapsack}) \\        
        \tabitem \textbf{$\OO^\star\left(2^{\tw}\cdot {\sf min}\{s^2,d^2\}\right)$ parameterized by \tw (\Cref{hor-di-fpt})} \\
       \tabitem {\textbf{$\OO^\star\left(e^b\cdot b^b\cdot 2^{b^2}\cdot {\sf min}\{s^2,d^2\}\right)$ randomized algo. parameterized by $b$}} \\ \textbf{(\Cref{hor-fpt-cc})} \\
        \tabitem \textbf{$\OO^\star\left(e^b \cdot b^{\OO(\log b)} \cdot b^b \cdot 2^{b^2} \cdot {\sf min}\{s^2, d^2\}\right)$ deterministic algo. parameterized} \\ \textbf{by $b$} \textbf{(\Cref{hor-fpt-cc-det})} \\
        \tabitem \textbf{$\OO^\star\left(e^{2d} \cdot (2d)^{\OO(\log d)} \cdot (2d)^{2d} \cdot 2^{(2d)^2} \cdot {\sf min}\{s^2, d^2\}\right)$ parameterized by $d$} \\\textbf{(\Cref{hor-fpt-cc-det-profit})} \\
    } \\ \hline

    \makecell[c]{\textsc{All-Neighborhood}\\\textsc{Knapsack}} & \makecell[l]{
         \tabitem strongly \NPC even for uniform, directed graphs (Borradaile et al. \cite{borradaile2012knapsack}) \\
        \tabitem \textbf{W[1]-Hard parameterized by both $s$ and $d$ for directed graphs (\Cref{sand-whard})} \\
        \tabitem \textbf{W[1]-Hard parameterized by $b$ for directed graphs (\Cref{sand-whard-sol-size})} \\
        \tabitem  $\OO\left(n\cdot P + n^2\right)$ time algorithm for 
        undirected graphs (Goebbels et al. \cite{goebbels2022knapsack}) \\
        \tabitem $\OO\left(n\cdot (P+1)\cdot (P+n)\right)$ time algorithm for 
        trees (Goebbels et al. \cite{goebbels2022knapsack}) \\
        \tabitem \textbf{$\OO^\star\left(2^{\tw}\cdot {\sf min}\{s^2,d^2\}\right)$ parameterized by \tw (\Cref{hand-di-fpt})} \\
    } \\ \hline

    \makecell[c]{\textsc{Relaxed}\\\textsc{1-Neighborhood}\\\textsc{Knapsack}} & \makecell[l]{
        \tabitem \textbf{\NPC even for uniform, directed graph (\Cref{sor-npc-uni-dir})} \\
         \tabitem \textbf{strongly \NPC even for undirected bipartite graphs (\Cref{sor-npc-bi})} \\
        \tabitem \textbf{\NPC for trees (\Cref{sor-tree-npc})} \\
        \tabitem \textbf{W[2]-Hard parameterized by $s$ (\Cref{sor-whard})} \\
        \tabitem \textbf{$\OO\left(n+m\right)$ time algorithm for uniform, undirected graphs (\Cref{undi-sor-uni-linear time})} \\
        \tabitem \textbf{$\OO^\star\left(4^{\tw}\cdot {\sf min}\{s^2,d^2\}\right)$ algo. parameterized by \tw (\Cref{sor-di-fpt})} \\
        \tabitem {\textbf{$\OO^\star\left(e^b\cdot b^b\cdot 2^{b^2}\cdot {\sf min}\{s^2,d^2\}\right)$ randomized algo. parameterized by $b$}} \\ \textbf{(\Cref{sor-fpt-cc})} \\
        \tabitem \textbf{$\OO^\star\left(e^b \cdot b^{\OO(\log b)} \cdot b^b \cdot 2^{b^2} \cdot {\sf min}\{s^2, d^2\}\right)$ deterministic algo. parameterized }\\ \textbf{by} $b$ \textbf{(\Cref{sor-fpt-cc-det})} \\
        \tabitem \textbf{$\OO^\star\left(e^{2d} \cdot (2d)^{\OO(\log d)} \cdot (2d)^{2d} \cdot 2^{(2d)^2} \cdot {\sf min}\{s^2, d^2\}\right)$ parameterized by $d$} \\\textbf{(\Cref{sor-fpt-cc-det-profit})} \\
        \tabitem \textbf{additive 1 approximation algo. for uniform, directed graph (\Cref{di-sor-1-approx})} \\
    } \\ \hline

    \makecell[c]{\textsc{Relaxed}\\\textsc{All-Neighborhood}\\\textsc{Knapsack}} & \makecell[l]{
        \tabitem \textbf{strongly \NPC even for uniform, undirected graphs (\Cref{sand-npc})} \\
        \tabitem \textbf{\NPC for trees (\Cref{sand-tree-npc})} \\
        \tabitem \textbf{W[1]-Hard parameterized by both $s$ and $d$ (\Cref{sand-whard})} \\
        \tabitem \textbf{W[1]-Hard parameterized by $b$ (\Cref{sand-whard-sol-size})} \\
        \tabitem \textbf{$\OO^\star\left(4^{\tw}\cdot {\sf min}\{s^2,d^2\}\right)$ (\Cref{sand-di-fpt})} \\
     } \\ \hline  \multicolumn{2}{c}{}\\
\caption{Summary of results. Our results are written in \textbf{bold}. Existing results are shown in normal font. $b$: size of the smallest solution, \tw: treewidth, $s$: size of knapsack, $d$ is target $d$ of knapsack, \(P = \sum_{i=1}^{n} p_i \leq n \cdot \max_{1 \leq i \leq n} p_i\)}.
\label{tab:contributions}
\end{longtable}

%% file: related_work.tex
\subsection{Related Work}
The classical \kp problem on graphs has been studied widely. Knapsack with neighborhood constraints have a rich literature.  Borradaile et al. studied the hard versions of the problem in \cite{borradaile2012knapsack} under the name 1-neighbour knapsack problem and all-neighbour knapsack problem. They introduce the concept of \textit{viable sets}, which are structured subsets of vertices that maintain feasibility while adhering to a predefined \textit{viable family} (e.g., stars, arborescences). For the \textit{1-neighbour knapsack} problem, they develop a polynomial-time approximation scheme (PTAS) for uniform, directed graphs and show that the general, directed case is hard to approximate within a $\frac{1}{\Omega(\log^{1-\epsilon} n)}$ factor. In the undirected case, they prove that partitioning into stars enables an approximation ratio of $\frac{(1 - \epsilon)}{2}\cdot (1 - \frac{1}{e^{(1-\epsilon)}})$. For the \textit{all-neighbours knapsack} problem, they establish that the uniform, undirected case reduces to the classical \kp, while the uniform, directed case is \NPH but admits a PTAS. Their results provide a structured approach to solving constrained knapsack problems via viable sets.

Knapsack with neighborhood constraints were further explored for special kinds of graphs in \cite{goebbels2021knapsack, goebbels2022knapsack}. The authors analyze the computational complexity of the hard version problems across different graph classes- undirected, directed, trees, complement reducible graphs, and minimum series parallel digraphs. For undirected graphs, \hand admits a polynomial time approximation scheme (PTAS), whereas \hor is APX-hard, implying that no PTAS exists unless \Pb = \NP. The authors design efficient dynamic programming algorithms for solving special graph families such as directed trees, directed co-graphs, and minimal series-parallel digraphs, achieving pseudo-polynomial runtime.

Carraway et al. \cite{carraway1993algorithm} studied the Stochastic Linear Knapsack Problem. In this version, costs are fixed and certain, but the returns are uncertain and follow independent normal distributions. The primary goal of the problem is to maximize the probability that the total return from the selected items meets or exceeds a specified target value. The authors propose a hybrid algorithm that guarantees optimality. This new method combines Dynamic Programming with branch-and-bound techniques, which improves performance compared to traditional Pareto optimization approaches. Mougouei et al. \cite{mougouei2017integer} defined the Binary Knapsack Problem with Dependent Item Values (BKP-DIV) problem, in which the value of an item may depend on the presence or absence of other items in the knapsack. These dependencies can be positive or negative, meaning the value of an item may increase or decrease depending on the selection of other items. This dependency is modeled using fuzzy graphs, which help capture the imprecision in these relationships.

The authors in \cite{lalou2023pseudo} studied the Knapsack Problem with Setups expands on the classical Knapsack Problem (KP) by introducing a more complex variant where items are partitioned into distinct groups, each of which incurs a fixed setup cost if at least one item from that group is selected. The authors extend the dynamic programming solution for the classical knapsack problem to handle the added complexity of group setup costs.

%% file: problem_definition.tex
\section{Preliminaries}

We denote the set $\{1,2,\ldots\}$ of natural numbers with \NB. For any integer \el, we denote the sets $\{1,\ldots,\el\}$ and $\{0,1,\ldots,\el\}$ by $[\el]$ and $[\el]_0$ respectively. We denote a graph \GG = (\VV, \EE), where \VV is a non-empty set of vertices and \EE is the set of edges in an undirected graph. In a directed graph \GG, each edge in the set \EE is associated with a direction. In this study, we allow self loops however, we don not allow parallel edges. In an undirected graph, $N_\GG(v) = \{v \mid \{u,v\} \in \EE\}$ denotes the set of open neighborhood of vertex $v$ and the set of closed neighborhood is denoted by $N_\GG[v] = N_\GG(v) \cup \{v\}$ in the graph \GG. For a vertex $v \in \VV$ in a directed graph, $N^+_\GG(v) = \{u \mid (v,u) \in \EE\}$ is the set of all out neighbors of $v$ and $N^-_\GG(v) = \{u \mid (u,v) \in \EE\}$ is the set of all in neighbors of $v$ in the digraph \GG. We omit the subscript~$\GG$ when the graph is clear from context. For a set of vertices $\SS \subseteq \VV$, the subgraph of \GG induced by \SS is denoted by $\GG[\SS]$. Let \( w, p :\VV \to \mathbb{N}_{\geq 0} \) be two functions that assign non-negative weights and profits to the vertices of a graph. We say that a graph is \emph{uniform} if the weight and profit of every vertex is 1, and \emph{general} if the weights and profits are arbitrary.

For the interest of space, we skip the basics of parameterized complexity. They are available in the Appendix.

We now define our problems formally.\longversion{ We start with the classical \kp problem.}

\begin{definition}[\kp]
Given a collection $\II = \{a_1, a_2, \dots, a_n\}$ of $[n]$ items such that each $a_i$ has non-negative size $s_i$ and non-negative profit $p_i$, $\forall i \in [n]$, a knapsack of capacity $c$, and a target profit $d$, does there exists $\JJ \subseteq \II$, such that $w(\JJ) = \sum_{i \in \JJ} s_i \leq c$ and $p(\JJ) = \sum_{i \in \JJ} p_i \geq d$?
\end{definition}

\begin{definition}[\hor] \label{hard-or}
    Given a graph $\GG = (\VV, \EE)$, where each vertex $v \in \VV$ has a weight $w_v$ and a profit $p_v$, a knapsack of size $s$, and a demand $d$, compute a subset $\SS \subseteq \VV$ if exists, such that \(\sum_{\substack{v \in \SS: \\ N(v) \cap \SS \neq \emptyset \ \text{or} \\ N(v) = \emptyset}} w_v \leq s 
    \quad \text{and} \quad 
    \sum_{\substack{v \in \SS: \\ N(v) \cap \SS \neq \emptyset \ \text{or} \\ N(v) = \emptyset}} p_v \geq d.\) For brevity, we call this problem the hard or variant.
\end{definition}

 
\begin{definition}[\sor]
    Given a graph $\GG = (\VV, \EE)$, where each vertex $u \in \VV$ has a weight $w_u$ and a profit $p_u$, a knapsack of size $s$, and a demand $d$, compute a subset $\SS \subseteq \VV$ if exists, such that \(\sum_{v \in \SS} w_v \leq s 
    \quad \text{and} \quad 
    \sum_{\substack{v \in \SS: \\ N(v) \cap \SS \neq \emptyset \ \text{or} \\ N(v) = \emptyset}} p_v \geq d.\) For brevity, we call this problem the soft or variant.
\end{definition}

\longversion{ The \hand variant enforces a strict selection condition, requiring that a vertex can be included in the solution only if all its neighbors are also selected. Both the weight and profit constraints are evaluated over the subset of vertices that satisfy this neighborhood condition.}

\begin{definition}[\hand]
    Given a graph $\GG = (\VV, \EE)$, where each vertex $v \in \VV$ has a weight $w_v$ and a profit $p_v$, a knapsack of size $s$, and a demand $d$, compute a subset $\SS \subseteq \VV$ if exists, such that \(\sum_{\substack{v \in \SS: \\ N(v) \subseteq \SS}} w_v \leq s 
    \quad \text{and} \quad 
    \sum_{\substack{v \in \SS: \\ N(v) \subseteq \SS}} p_v \geq d.\) For brevity, we call this problem the hard and variant.
\end{definition}

\longversion{On the other hand, in the \sand variant, vertices contribute to the total profit only if all their neighbors are also selected. Otherwise, if such vertices are included in the solution, they contribute only to the total weight but not to the profit. Consequently, in \sand, some vertices may be included in the solution not for their own profit contribution but to ensure that other vertices, which depend on them, can contribute to the total profit.}

\begin{definition}[\sand]
    Given a graph $\GG = (\VV, \EE)$, where each vertex $v \in \VV$ has a weight $w_v$ and a profit $p_v$, a knapsack of size $s$, and a demand $d$, compute a subset $\SS \subseteq \VV$ if exists, such that \(\sum_{v \in \SS} w_v \leq s 
    \quad \text{and} \quad 
    \sum_{\substack{v \in \SS: \\ N(v) \subseteq \SS}} p_v \geq d.\) For brevity, we call this problem the soft and variant.
\end{definition}

In all the definitions above, the neighborhood \( N(v) \) refers to $N^+(v)$ the out-neighbors of vertex \( v \) when the graph \( \GG \) is directed (i.e., a digraph). The optimization version of these problems seeks to maximize the total profit subject to the knapsack size constraint \( s \).

For the sake of completeness, we also include the definitions of classical problems such as \setc and \clv, which are used in the subsequent sections.

\begin{definition}[\setc]
  Given a universe \UU of $[n]$ elements and sets $\{S_i\}_{i = 1}^m$, such that $S_i \subseteq \UU$, $\forall i \in [m]$. Does there exists $\II \subseteq [m]$, $|\II| \leq k$, such that $\bigcup_{i\in\II} S_i = \UU$?
\end{definition}

\begin{definition}[\clv]
  Given a graph \GG= (\VV,\EE), integers $l$ and $k$. Does there exists a partition $\VV = \XX \cup \SS \cup \YY$ such that $|\XX| = l$, $|\SS| \leq k$ and there is no edge between \XX and \YY?
\end{definition}

\begin{proposition}\label{self-loop-un-sor}
In the undirected \sor problem, given an instance \GG= (\VV,\EE) if a vertex \( v \in \VV \) has a self-loop, we can eliminate the self-loop by introducing a dummy vertex \( v^\pr \) with zero weight and zero profit, and replacing the self-loop \( \{v, v\} \) with the edge \( \{v, v^\pr\} \), without changing the feasibility of the instance.
\end{proposition}

\begin{proposition}\label{self-loop-di-sor}
In the directed \sor problem, given an instance \( \GG = (\VV, \EE) \) if a vertex \( v \in \VV \) has a self-loop, we can eliminate the self-loop by simply dropping all outgoing edges from \( v \), including the self-loop, without changing the feasibility of the instance.
\end{proposition}

\begin{proposition}\label{self-loop-sand}
In the directed (or undirected) \sand problem, given an instance \( \GG = (\VV, \EE) \) (\GG can be directed or undirected) if a vertex \( v \in \VV \) has a self-loop, we can simply drop the self loop, without changing the feasibility of the instance.
\end{proposition}

For the rest of the paper, we assume without loss of generality that \( \GG \) does not have self-loops. Unless stated otherwise, we use $n$ and $m$ to denote respectively the number of vertices in graph or the items and the number of edges in graph problems; \tw to denote the treewidth of a graph (ignoring directions for directed graphs); $s$ to indicate the knapsack size; $d$ to represent the target solution value; and $b$ to denote the minimum number of vertices in any feasible solution for \yes-instances of the decision version of the problem. In the optimization version, $b$ corresponds to the minimum number of vertices required in any optimal solution that satisfies the given weight constraint.


%% file: results_npc.tex
\section{Hardness Results}
\label{result:hardness}

We present our hardness results in this section. Notice that all the problems on undirected graphs can be reduced to directed graphs by replacing each edge say, $\{u,v\}$ by two directed edges \((u,v)\) and \((v,u)\) in $\GG^\pr$.  Hence, hardness results proved for undirected graphs implies the same for directed graphs. Also, algorithmic results proved for directed graphs applies ditto for undirected graphs too. 

In the interest
of space, we omit the proofs of a few of our results. They are marked ($\star$). We first show that \sor problem is strongly \NPH.

\begin{theorem}\label{sor-npc-bi}
  \sor is strongly \NPC even for undirected bipartite graphs.  
\end{theorem}
\begin{proof}
\sor $\in$ \NP, because given a certificate i.e. a set of vertices $\VV^\pr \subseteq \VV$, it can be verified in polynomial time if  \(\sum_{v \in \VV^\prime} w_v \leq s 
    \quad \text{and} \quad 
    \sum_{\substack{v \in \VV^\prime: \\ N(v) \cap \VV^\prime \neq \emptyset \ \text{or} \\ N(v) = \emptyset}} p_v \geq d.\)

To show hardness, we reduce the \setc problem to \sor.
Given an instance of \setc as (\UU = [n], $\{S_i\}_{i = 1}^m$, $k$), we construct an instance of \sor as follows: Create a bipartite graph $\GG(\VV = \AA \cup \BB, \EE)$ where the vertex set is partitioned into two parts. For each element \( j \in [n] \) of the universe, create a vertex \( u_j \) in \AA , and for each set \( S_i \) where \( i \in [m] \), create a vertex \( v_i \) in \BB. The edge set \( \EE \) captures the membership relation between elements and sets, and is defined as \( \EE = \{\{u_j, v_i\} \mid j \in S_i, \forall j \in [n], i \in [m]\},\) so that there is an edge between \( u_j \in \AA \) and \( v_i \in \BB \) if and only if the element \( j \) belongs to the set \( S_i \). Assign the weights and profits such that each vertex $u \in \AA$ has weight $w_u = 0$, and profit $p_u = 1$. Similarly, each vertex $v \in \BB$ has weight $w_v = 1$ and profit $p_v = 0$.
Set the knapsack size as $s = k$ and the demand as $d = n$.

Now we claim that the \setc is an \yes instance if and only if the \sor instance is an \yes instance. To see this let us suppose \setc is an \yes instance. This means there exists a set $\II \subseteq [m]$ of size $k$ that covers all the $n$ elements of \UU. We pick the vertices from \BB corresponding to the sets in \II and their adjacent vertices from \AA in our knapsack. Notice that the weight of the items picked is exactly $k$, since each vertex in \BB has weight 1 and we have picked only $k$ vertices from \BB. Each vertex from \AA has weight 0 and profit 1 and has at least one neighbor in \BB. Hence the total sum of weights is $k$. Because \II is a set cover, it must cover all the $n$ elements of \UU i.e. all the vertices from \AA are added into the knapsack and total sum of profits is $n$. Therefore the \sor is also an \yes instance.

Conversely, let us assume \sor is an \yes instance. This means there is a set $\VV^\pr \subseteq \VV$ such that  \(\sum_{v \in \VV^\prime} w_v \leq k 
    \quad \text{and} \quad 
    \sum_{\substack{v \in \VV^\prime: \\ N(v) \cap \VV^\prime \neq \emptyset \ \text{or} \\ N(v) = \emptyset}} p_v = n\)
and each vertex has at least one of its out-neighbors in \BB included in the knapsack otherwise the profit cannot be $n$. Notice that the set $\II = \VV^\pr \cap \BB$ is a valid set cover because $\bigcup_{i: v_i \in \II}{S_i} = \UU$ covers all $n$ elements and has weight $k$.
\end{proof}

The same reduction in \Cref{sor-npc-bi} holds for directed graph and directed \hor when all edges, dictated by membership, are directed from elements in $\AA$ to sets in $\BB$, with the same weight and profit assignments. This yields the following corollaries.

\begin{corollary}\label{sor-npc}
 \sor is strongly \NPC for directed acyclic graphs. 
\end{corollary}

The reduction from \setc to \sor is parameter preserving.

\begin{corollary}\label{sor-whard}
  \sor parameterized by the size of the knapsack $s$, is W[2]-hard.
\end{corollary}

\begin{corollary}\label{hor-woh}
  \hor parameterized by the size of the knapsack $s$, is W[2]-hard for directed graphs.  
\end{corollary}

Interestingly, \sor is strongly \NPC even for uniform, directed graphs. The proof is similar to the proof of uniform, directed \hor proved in \cite{borradaile2012knapsack}. 

\begin{observation}\label{sor-npc-uni-dir}
   \sor is strongly \NPC even for uniform, directed graphs. 
\end{observation}

When the graph \GG is an edgeless graph, both the \sor and \sand problems reduce to the classical \kp problem itself and admit a pseudo-polynomial algorithm. We now show that \sor is \NPC for trees.

\begin{theorem}\label{sor-tree-npc}
 ($\star$) \sor is \NPC for trees.  
\end{theorem}
\longversion{\begin{proof}
    \sor for trees $\in$ \NP. To prove hardness, we reduce from \kp problem. 

    We denote an instance of \kp as $(\II, \{s_i\}_{i \in \II}, \{p_i\}_{i \in \II}, c, \alpha)$. Given, an instance of \kp problem, we construct an instance of \sor as follows: let the graph is $\GG = (\VV, \EE)$ with $\VV = \II \cup \{a_{n+1}\}$ where vertex $v_i$ corresponds to item $a_i$ has weight $w_i$ and profit $p_i$, and the vertex $v_{n+1}$ has weight 0 and profit 0. Define the edge set \( \EE = \{ \{v_i, v_{n+1}\} | 1 \leq i \leq n\}\).

Since the central vertex has weight and profit 0, it can always be taken in any feasible solution. Now, the \kp instance $(\II, \{s_i\}_{i \in \II}, \{p_i\}_{i \in \II}, c, \alpha)$ is an \yes instance if and only if the \sor instance $(\GG = (\VV,\EE), \{w_i\}_{i \in \VV}, \{p_i\}_{i \in \VV}, s = c, d = \alpha)$ is an \yes instance. 
\end{proof}
}

\hand for undirected graphs is solvable in pseudo-polynomial time (see \cite{goebbels2021knapsack}). However, it turns out \sand is strongly \NPC even for uniform, undirected graphs. Notice that any optimal solution for \sand is a set $\SS \cup \XX = \VV^\pr \subseteq \VV$ such that the subgraph $\GG[\VV^\pr]$ contains a set of vertices \XX whose neighborhoods are entirely contained in $\VV^\pr$, and hence they contribute to the total profit. There may exist a set \SS of vertices in $\VV^\pr$ which do not themselves contribute to the profit (because not all of their neighbors are selected in $\XX \cup \SS $), but are included in the solution to ensure that the vertices in \XX do satisfy the profit condition. These vertices in \SS can be seen as separating \XX and $\VV \setminus \VV^\pr$.

\begin{theorem}\label{sand-npc}
($\star$) \sand is strongly \NPC even for uniform, undirected graphs. 
\end{theorem}

\longversion{\begin{proof}
\sand $\in$ \NP, because given a certificate i.e. a set of vertices  $\VV^\pr \subseteq \VV$, it can be verified in polynomial time if  \(\sum_{v \in \VV^\prime} w_v \leq s 
\quad \text{and} \quad 
\sum_{\substack{v \in \VV^\prime: \\ N(v) \subseteq \VV^\prime}} p_v \geq d.\)

To show hardness, we reduce the \clv problem to \sand. 

Given an instance of \clv as $(\GG = (\VV, \EE),k,l)$, we construct an instance of \sand as follows: we consider the same graph \GG, for each vertex $u \in \VV$, set $w_u = 1$, $p_u = 1$, knapsack capacity $s = l+k$ and profit $d=l$. We now claim that \clv instance $(\GG = (\VV, \EE),k,l)$ is an \yes instance if and only if \sand instance $(\GG = (\VV,\EE), \{w_i\}_{i \in \VV} = 1, \{p_i\}_{i \in \VV} = 1, s = l+k, d = l)$ is an \yes instance. 

Suppose \clv is an \yes instance then there exists a partition $\VV = \XX \cup \SS \cup \YY$ in $\GG$ such that $|\XX| = l$, $|\SS| \leq k$, and there is no edge between $\XX$ and $\YY$. Consider the set $\VV^\pr = \XX \cup \SS$ in $\GG$. Note that
\(\sum_{v \in \VV^\pr} w_v = |\VV^\pr| = |\XX| + |\SS| \leq l + k = s.\)For each $v \in \XX$, since there is no edge from $\XX$ to $\YY$, all neighbors of $v$ lie in $\XX \cup \SS = \VV^\pr$. Thus, $N(v) \subseteq \VV^\pr$, so $v$ contributes profit. Hence,
\(
\sum_{\substack{v \in \VV^\pr: \\ N(v) \subseteq \VV^\pr}} p_v \geq |\XX| = l = d.
\) Therefore, $\VV^\pr$ is a valid solution for \sand.

Suppose \sand is a YES-instance, then \clv is a YES-instance: Suppose $\VV^\pr \subseteq \VV$ is a solution for \sand such that
\(
\sum_{v \in \VV^\pr} w_v \leq s = k + l \quad \text{and} \quad \sum_{\substack{v \in \VV^\pr: \\ N(v) \subseteq \VV^\pr}} p_v \geq d = l.
\) Let $\XX \subseteq \VV^\pr$ denote the set of vertices in $\VV^\pr$ such that all their neighbors are also in $\VV^\pr$, i.e., 
\(
\XX = \{ v \in \VV^\pr \mid N(v) \subseteq \VV^\pr \}.
\)
Then, the total profit contributed is exactly $|\XX| \geq l$. Since $\sum_{v \in \VV^\pr} w_v = |\VV^\pr| \leq k + l$, define $\SS = \VV^\pr \setminus \XX$, and note that
\(
|\SS| = |\VV^\pr| - |\XX| \leq (k + l) - l = k.
\)
Now define $\YY = \VV \setminus \VV^\pr$. Since every $v \in \XX$ has all its neighbors in $\VV^\pr$, there is no edge between $\XX$ and $\YY$. Hence, we have a partition $\VV = \XX \cup \SS \cup \YY$ with $|\XX| = l$, $|\SS| \leq k$, and no edge between $\XX$ and $\YY$. Thus, \clv is a YES-instance.
\end{proof}
}

The reduction from \clv to \sand is parameter preserving. Marx showed that \clv is \WOH parameterized by both $k$ and $l$, where $|\XX| = l$ and $|\SS| \leq k$ in \cite{marx2006parameterized}.
\begin{corollary}\label{sand-whard}
  \sand parameterized by the size of the knapsack $s$, plus demand $d$, is \WOH even for uniform, directed graphs.
\end{corollary}

\begin{corollary}\label{sand-whard-sol-size}
  \sand parameterized by solution size, $b$ is \WOH.
\end{corollary}
Similar to \sor, \sand is \NPC for trees.

\begin{theorem}\label{sand-tree-npc}
  \sand is \NPC for trees.  
\end{theorem}

The proof of \Cref{sand-tree-npc} is analogous to the proof of \Cref{sor-tree-npc}.

We now establish W-hardness results for \hand.

\begin{theorem}\label{hand-woh}
($\star$) \hand parameterized by the size of the knapsack $s$, plus demand $d$, is \WOH for directed graphs . 
\end{theorem}

\longversion{\begin{proof}
Given an instance of \clique as $(\GG^\pr = (\VV^\pr, \EE^\pr), k)$, we construct an instance of \hand as follows: We construct a bipartite graph $\GG(\VV = \AA \cup \BB, \EE)$ from a given graph $\GG^\pr = (\VV^\pr, \EE^\pr)$ as follows. For each edge $e \in \EE^\pr$, create a vertex $u_e$ in $\AA$, and for each vertex $v \in \VV^\pr$, create a vertex $v$ in $\BB$ i.e, $\AA$ corresponds to the edge set of $\GG^\pr$ and $\BB$ corresponds to the vertex set of $\GG^\pr$. For every edge $e = \{v_i, v_j\} \in \EE^\pr$, we add edges $(u_e, v_i)$ and $(u_e, v_j)$ in $\GG$, where $u_e \in \AA$ and $v_i, v_j \in \BB$, representing the incidence of edge $e$ with vertices $v_i$ and $v_j$ in the original graph $\GG^\pr$. We assign the weights and profits to the vertices of $\GG$ as follows. Each vertex $u \in \AA$ (corresponding to an edge in $\GG^\pr$) is assigned weight \( w_u = 0 \) and profit \( p_u = 1 \). Each vertex \( v \in \BB \) (corresponding to a vertex in $\GG^\pr$) is assigned weight \( w_v = 1 \) and profit \( p_v = 0 \). Set $s = k$ and $d = \binom{k}{2}$. 
We claim that the \clique problem is a \yes instance if and only if the corresponding \hand instance is also a \yes instance.

First, suppose that the \clique problem is a \yes instance. This means there exists a clique of size \(k\) in $\GG^\pr$. We can construct a solution to the \hand instance by selecting the \(k\) nodes from \BB that correspond to the vertices of the \(k\)-\clique. Additionally, we include \(\binom{k}{2}\) nodes from \AA that correspond to the edges in the \(k\)-clique. In this configuration, the \(k\) nodes from set \BB contribute a total weight of \(k\) and no profit, while the edges selected from \AA contribute a profit of \(\binom{k}{2}\). This configuration satisfies the constraints of the knapsack problem: the total weight \(\sum_{\substack{v \in \SS: \\ N(v) \subseteq \SS}} w_u \leq k\) and the total profit \(\sum_{\substack{v \in \SS: \\ N(v) \subseteq \SS}} p_v = \binom{k}{2}\).

Conversely, if the \sand instance is a \yes instance, we can construct a corresponding \(k\)-\clique in $\GG^\pr$. Since we must select at least \(\binom{k}{2}\) edges from \AA to meet the profit requirement, this implies that the chosen edges correspond to a complete subgraph (clique) among the selected vertices in \BB. Therefore, \clique instance $(\GG^\pr = (\VV^\pr, \EE^\pr), k)$ is an \yes instance if and only if  \sand instance $(\GG = (\VV = (\AA \cup \BB),\EE), \{w_i\}_{i \in \VV}, \{p_i\}_{i \in \VV}, s = k, d = \binom{k}{2})$ is an \yes instance.
\end{proof}
}
In the reduction above, if we set $w_u = 1$ of each vertex in \AA and set the size $s = k + \binom{k}{2}$, we obtain the following corollary:

\begin{corollary}
\hand parameterized by solution size $b$, is \WOH for directed graphs . 
\end{corollary}

%% file: parameterized_algo.tex
\section{Algorithmic Results}

In this section, we discuss our algorithmic results. Note that \sor and \sand are \NPC even for trees (\Cref{sor-tree-npc}, \Cref{sand-tree-npc}), and hence they are \WOH parameterized by treewidth of the input graph. Also, since \kp is \WOH parameterized by solution size \cite{abboud2014losing}, we do not hope to have \FPT algorithms for \sor parameterized by solution size. We design pseudo-\FPT algorithms for directed \hor and \hand parameterized by the treewidth \tw, of the input graph. Furthermore, we show that \sor is solvable in polynomial time for uniform, undirected graphs. We also present pseudo-\FPT algorithms for directed \sor and \sand, again parameterized by the treewidth. Additionally, we propose a color-coding-based pseudo-\FPT algorithm for directed \sor, parameterized by the solution size, and provide a derandomized version later. Finally, we show that \sor also admits a pseudo-\FPT algorithm when parameterized by the total profit.


Goebbels et al. \cite{goebbels2022knapsack} proved that \hor and \hand are strongly \NPC for directed graphs. We design pseudo-\FPT algorithms for \hor and \hand for directed graphs with bounded treewidth. 

\input{fpthard}


In the undirected, unweighted setting of the \sor problem where all vertex weights and profits are $1$ and the knapsack size and demand are both equal to $k$, the structure of any feasible solution exhibits a simple yet crucial property. A subset of vertices forms a feasible solution if it induces a forest in $\GG$, where every connected component has size at least $2$. This structural restriction arises directly from the \sor semantics: a vertex contributes to the profit only if at least one of its neighbors is also selected. A pendant vertex of a component selected alone cannot be profitable, as it has no selected neighbor. Thus, the optimal solution consists only of non-trivial components—trees of size at least $2$—where every vertex has at least one neighbor in the solution and hence contributes to the profit. Notice that this is the exact structure an optimal \hor solution admits. 

\begin{observation}
    In a uniform, undirected graph \sor reduces to \hor.
\end{observation}

Using this reduction, we can apply the linear-time algorithm designed by Borradaile et al.\ in~\cite{borradaile2012knapsack} for solving the uniform, undirected \hor problem. This yields the following result:

\begin{theorem}\label{undi-sor-uni-linear time}
The \sor problem on undirected graphs with unit weights and unit profits for a given knapsack of size $k$ can be solved in linear time.    
\end{theorem}

Unlike the undirected version, uniform, directed \sor is strongly \NPC. We design an additive 1 approximation algorithm for uniform, directed \sor.

\begin{theorem}\label{di-sor-1-approx}
The \sor problem on directed graphs with unit weights and unit profits, and a given knapsack of size $k$, admits a linear-time algorithm that returns a solution with profit at least $\OPT - 1$, where $\OPT$ is the profit of an optimal solution.  
\end{theorem}

\begin{proof}
We are given a directed graph $\GG = (\VV, \EE)$ with unit weights and unit profits, and a knapsack size $k$. The goal is to find a subset $\SS \subseteq V$ of size at most $k$ that satisfies the \sor profit conditions and achieves a total profit of at least $\OPT - 1$.

We first compute the strongly connected components (SCCs) of $\GG$, which can be done in linear time using Kosaraju’s or Tarjan’s algorithm. We sort the components in non-increasing order of their sizes. Let the components be $C_1, C_2, \cdots,C_l$. We initialize an empty solution set $\SS = \emptyset$ and a counter $c = |\SS|$.

Next, we iteratively add vertices from entire components to $\SS$ as long as the total size remains within the knapsack budget $k$ and update $c$ accordingly. Let $C_j$ be the first component that cannot be fully included due to size constraints. Let $t = k - c$ denote the remaining budget.

We now analyze the remaining possibilities to fill the remaining budget $t$. Construct the subgraph $\GG[C_j]$ and perform a DFS from any vertex $v \in C_j$ to build a DFS tree. Add the first $t$ vertices visited in DFS order to the solution set $\SS$.

For each of the first $t - 1$ vertices in the DFS order, either the DFS tree edge or a back edge guarantees that the vertex has an out-neighbor among the previously selected vertices. Hence, these $t - 1$ vertices satisfy the \sor profit condition. Only the last selected vertex (i.e., the $t$-th in the DFS order) may not have any out-neighbor among the selected set if all its outgoing edges point to vertices outside the selected subtree. In such a case, this single vertex may not contribute to the profit. Nevertheless, at least $t - 1$ vertices do contribute, which guarantees a total profit of at least $\OPT - 1$.
\end{proof}

We now solve the problem for general directed graphs. In \Cref{sor-npc-bi,sand-npc}, we showed that \sor and \sand are \NPC, hence we do not expect polynomial time algorithm under standard complexity theoretic assumptions. We design pseudo-\FPT algorithms parameterized by treewidth. The approach used in designing pseudo-\FPT algorithms for \hor and \hand breaks down for \sor and \sand, as we now have the choice to include a vertex without accounting for its profit. This necessitates explicitly keeping track of the vertices that contribute to the total weight, and among them, identifying the subset that actually contributes to the total value. On a high level, we store partial solutions in each state where we guess the subset $\SS$ as the vertices included in the optimal solution say, $\VV^\pr$, and $\PP \subseteq \SS$ as the subset that contributes to the total profit.
\begin{theorem}\label{sor-di-fpt}
($\star$) There is an \FPT algorithm for directed \sor with running time $\OO^\star\left(4^{\tw}\cdot {\sf min}\{s^2,d^2\}\right)$, where $n$ is the number of vertices in the input graph, $\tw$ is the treewidth of the input graph, $s$ is the input size of the knapsack and $d$ is the input target value. 
\end{theorem} 

\longversion{
\begin{proof}
We consider a nice tree decomposition of the underlying graph. Let $(\GG = (V,E),\{w_u\}_{u \in V_G}, \{p_u\}_{u \in V_G}, s,d)$ be an input instance of \sor such that $\tw=tw(\GG)$. Let $\VV^\pr$ be an optimal solution to \sor. For technical purposes, we guess a vertex $v \in \VV^\pr$ --- once the guess is fixed, we are only interested in finding solution subsets $\hat{\VV}$ that contain $v$ and $\VV^\pr$ is one such candidate. We also consider a nice edge tree decomposition $(\mathbb{T} = (V_{\mathbb{T}},E_{\mathbb{T}}),\mathcal{X})$ of $\GG$ that is rooted at a node $r$, and where $v$ has been added to all bags of the decomposition. Therefore, $X_{r} = \{v\}$ and each leaf bag is the singleton set $\{v\}$.

We define a function $\ell: V_{\mathbb{T}} \rightarrow \mathbb{N}$ as follows.
For a vertex $t \in V_\mathbb{T}$, $\ell(t) = \sf{dist}_\mathbb{T}(t,r)$, where $r$ is the root. Note that this implies that $\ell(r) = 0$. Let us assume that the values that $\ell$ take over the nodes of $\mathbb{T}$ are between $0$ and $L$. For a node $t\in V_\mathbb{T}$, we denote the set of vertices in the bags in the subtree rooted at $t$ by $V_t$ and $\GG_t=\GG[V_t]$. Now, we describe a dynamic programming algorithm over $(\mathbb{T},\mathcal{X})$. We have the following states.

    \textbf{States:} We maintain a DP table $D$ where a state has the following components:
    \begin{enumerate}
    \item $t$ represents a node in $V_\mathbb{T}$.
    \item $\SS$ represents a subset of the vertex subset $X_t$
    \item $\PP$ represents a subset of the vertex subset $\SS$ 
    \end{enumerate}
    \textbf{Interpretation of States}
For each node \( t \in \mathbb{T} \), we maintain a list \( D[t, \SS, \PP] \) for every subset \( \SS \subseteq X_t \) and \( \PP \subseteq \SS \). Each entry in this list holds a set of undominated (weight, profit) pairs corresponding to valid \sor solutions $\hat{\SS}$ in the graph \( \GG_t \), where:  

\begin{itemize}  
    \item The selected vertices in \( \GG_t \) satisfy \(\hat{\SS} \cap X_t = \SS \).  
    \item The set of profit-contributing vertices is \( \PP \), meaning all vertices in \( \PP \) are part of \( \SS \) and each has at least one selected out-neighbor in \( \GG_t \) if it has a out-neighbor in \( \GG_t \).  
    \item If no such valid solution exists, we set \( D[t, \SS, \PP] = \emptyset \).  
\end{itemize}  

For each state $D[t,\SS,\PP]$, we initialize $D[t,\SS,\PP]$ to the list $\{(0,0)\}$.

    \textbf{Dynamic Programming on $D$}: We update the table $D$ as follows. We initialize the table of states with nodes $t\in V_\mathbb{T}$ such that $\ell(t)=L$. When all such states are initialized, then we move to update states where the node $t$ has $\ell(t) = L-1$, and so on, till we finally update states with $r$ as the node --- note that $\ell(r) =0$. For a particular $j$, $0\leq j< L$ and a state $D[t,\SS,\PP]$ such that $\ell(t) = j$, we can assume that $D[t',\SS',\PP']$ have been computed for all $t'$ such that $\ell(t')>j$ and all subsets $\SS'$ of $X_{t'}$. Now we consider several cases by which $D[t,\SS,\PP]$ is updated based on the nature of $t$ in $\mathbb{T}$:
    
    \begin{enumerate}
    \item \textbf{Leaf node.} Suppose $t$ is a leaf node with $X_{t} = \{v\}$. Then the list stored in $D[t,\SS,\PP]$ depends on the sets $\SS$ and $\PP$.  If $\PP$ is not a subset of $\SS$, we store $D[t,\SS,\PP] = \emptyset$; hence assume otherwise. The only possible cases that can arise are as follows:
    \begin{enumerate}
        \item $D[t,\phi,\phi]$ stores the pair $\{(0, 0)\}$
        \item $D[t,\{v\},\{v\}]$ = $\{(w_v, p_v)\}$, if $w_v \leq s$
        \item $D[t,\{v\},\phi]$ stores the pair $\{(w_v, 0)\}$, if $w_v \leq s$
    \end{enumerate} 
    
    \item \textbf{Introduce node.} Suppose $t$ is an introduce node. Then it has only one child $t'$ where $X_{t'} \subset X_{t}$ and there is exactly one vertex $u \neq v$ that belongs to $X_{t}$ but not $X_{t'}$. Then for all $\SS \subseteq X_t$ and all $\PP \subseteq X_t$ , if $\PP$ is not a subset of $\SS$, we store $D[t,\SS,\PP] = \emptyset$; hence assume otherwise.
    \begin{enumerate}
        \item If $u \not\in \SS$ and therefore $u \not\in \PP$, 
        \begin{enumerate}
            \item if \(\exists w: w\in N^-(u), w \in \PP \) and \(N^+(w) \cap \VV_t = \{u\}\), then $D[t,\SS,\PP] = \emptyset$.
            \item Otherwise, we copy all pairs of $D[t',\SS^\pr = \SS,\PP^\pr=\PP]$ to $D[t',\SS,\PP]$. 
        \end{enumerate}
        
        \item If $u \in \SS$,
        \begin{enumerate}            
            \item if $u$ has out-neighbor $x$ in $\SS$ or \(N^+(u) \cap \VV_t = \emptyset\), then $u \in \PP$, and $\PP = \PP^\pr \cup \{u\}$,
            and for every pair $(w,p)$ in $D[t',\SS^\pr=\SS \setminus \{u\},\PP^\pr=\PP \setminus \{u\}]$ with $w + w_u \leq s$, we add $(w+w_u, p + \sum_{z \in \SS: (z,u) \in \GG, z \not\in \PP^\pr} p_z + p_u)$ and copy to $D[t,\SS,\PP]$.
            
            \item if $u$ does not have a out-neighbor $x$ in $\SS$ and \(N^+(u) \cap \VV_t \neq \emptyset\), then $u \not\in \PP$, and $\PP = \PP^\pr$,
            then for every pair $(w,p)$ in $D[t',\SS^\pr=\SS \setminus \{u\},\PP^\pr=\PP \setminus \{u\}]$ with $w + w_u \leq s$, we add $(w+w_u, p + \sum_{z \in \SS: (z,u) \in \GG, z \not\in \PP^\pr} p_z )$ and copy to $D[t,\SS,\PP]$.

        \end{enumerate}
    \end{enumerate}
    
    \item \textbf{Forget node.} Suppose $t$ is a forget vertex node. Then it has only one child $t'$, and there is a vertex $u \neq v$ such that $X_t\cup\{u\} = X_{t'}$. Then for all $\SS \subseteq X_t$ and all $\PP \subseteq X_t$ , if $\PP$ is not a subset of $\SS$, we store $D[t,\SS,\PP] = \emptyset$; hence assume otherwise. For all $\SS \subseteq X_t$, we copy all feasible undominated pairs stored in $D[t',\SS\cup \{u\},\PP\cup \{u\}]$, $D[t',\SS\cup \{u\},\PP]$, and $D[t',\SS,\PP]$ to $D[t,\SS,\PP]$. If any pair stored in $D[t',\SS\cup \{u\},\PP]$, and $D[t',\SS,\PP]$ is dominated by any pair of $D[t',\SS\cup \{u\},\PP\cup \{u\}]$, we copy only the undominated pairs to $D[t,\SS,\PP]$.

    \item \textbf{Join node.} Suppose $t$ is a join node. Then it has two children $t_1,t_2$ such that $X_t = X_{t_1} = X_{t_2}$. Then for all $\SS \subseteq X_t$, let $(w(\SS), p(\SS))$ be the total weight and value of the vertices in $\SS$. Consider a pair $(w_1,p_1)$ in $D[t_1,\SS_1,\PP_1]$ and a pair $(w_2,p_2)$ in $D[t_2,\SS_2,\PP_2]$, where $\SS_1 \cup \SS_2 = \SS$ and $\PP_1 \cup \PP_2 = \PP$. Then for all pairs $(w_1,p_1) \in D[t_1,\SS_1,\PP_1]$ and $(w_2,p_2) \in D[t_2,\SS_2,\PP_2]$, if $w_1 + w_2 - w(\SS_1 \cup \SS_2) \leq s$, then we add $ \left( w_1 + w_2 - w(\SS_1 \cup \SS_2), p_1+p_2- p(\PP_1 \cap \PP_2) \right)$ and copy to $D[t,\SS,\PP]$.
    \end{enumerate}
    
Finally, in the last step of updating $D[t,\SS,\PP]$, we go through the list saved in $D[t,\SS,\PP]$ and only keep undominated pairs.

The output of the algorithm is a pair $(w,p)$ that is maximum of those stored in $D[r,\{v\},\{v\}]$ and $D[r,\{v\},\emptyset]$ such that $w \leq s$ and $p$ is the maximum value over all pairs in $D[r,\{v\},\{v\}]$ and $D[r,\{v\},\emptyset]$.

\textbf{Proof of correctness}: 
Recall that we are looking for a solution that is a set of vertices $\VV^\pr$ that contains the fixed vertex $v$ that belongs to all bags of the nice tree decomposition. In each state we maintain a list of feasible and undominated (weight, profit) pairs that correspond to the solution $\hat{\SS}$ of $\GG_t$ and  satisfy \( \hat{\SS} \cap X_t = \SS \) and the set of profit contributing vertices is $\PP$.

We now show that the update rules holds for each $X_t$. To prove this formally, we need to consider the cases of what $t$ can be:

\begin{enumerate}
    \item \textbf{Leaf node.} Recall that in our modified nice tree decomposition we have added a vertex $v$ to all the bags. Suppose a leaf node $t$ contains a single vertex $v$, $D[t,\phi,\phi]$ stores the pair $\{(0, 0)\}$, $D[t,\{v\},\{v\}]$ = $\{(w_v, p_v)\}$ if $w_v \leq s$, $D[t,\{v\},\phi]$ stores the pair $\{(w_v, 0)\}$ if $w_v \leq s$, otherwise $\emptyset$. This is true in particular when $j = L$, the base case. From now we can assume that for a node $t$ with $\ell(t) = j < L$ and all subsets $\SS, \PP \subseteq X_t$, $D[t',\SS^\pr,\PP^\pr]$ entries are correct and correspond to a \sor solution in $\GG_t$. when $\ell(t') > j$. 

    \item \textbf{Introduce node.} When $t$ is an introduce node, there is a child $t'$ such that $X_t = X_{t'} \cup \{u\}$. We are introducing a vertex $u$ and the edges associated with it in $\GG_t$. When $\PP$ is not a subset of $\SS$, the state $D[t,\SS,\PP] = \emptyset$ because it violates the definition of $\PP$; hence assume otherwise. Let us prove for each case.
    \begin{enumerate}
        \item When $u$ is not included in $\SS$, it is also not in $\PP$ and the families of sets $\hat{\SS}$ considered in the definition of $D[t,\SS,\PP]$ and of $D[t',\SS,\PP]$ are equal. Notice that there may be a vertex $w \in X_{t}$ that belongs to $\PP$ such that $N^+(w) \cap V_{t'} = \emptyset$, i.e. $w$ did not have any out-neighbor in the graph $\GG_{t'}$. However, with the introduction of the vertex $u$ and the edges associated to it, if $N^+(w) \cap V_{t} = \{u\}$ i.e. $u$ is the only out-neighbor of $w$ in $\GG_t$, since $u$ does not belong to $\SS$, the \sor constraint is violated and hence the state $D[t,\SS,\PP]$ stores $\emptyset$ as it is no longer feasible. Otherwise we copy all pairs of $D[t',\SS^\pr = \SS,\PP^\pr=\PP]$ to $D[t,\SS,\PP]$.

        \item Now consider the case when $u$ is part of $\SS$. Let $\hat{\SS}$ be a feasible solution for \sor attained in the definition of $D[t,\SS,\PP]$. Then it follows that $\hat{\SS} \setminus \{u\}$ is one of the sets considered in the definition of $D[t',\SS \setminus \{v\},\PP \setminus \{v\}]$. Notice that all $z$ in $\SS' \setminus \PP'$ such that $z \in N^-(u)$ get activated or become profitable due to inclusion of $u$. Thus the edge $(z,u)$ in $\GG_t$ is a witness that $z$ is profitable. Now we must compute entries depending on whether $u$ contributes to profit or not. It suffices to just check if $u$ has a neighbor in $X_t$ or not because the nice tree decomposition ensures that the newly introduced vertex $u$ can have neighbors only in $X_t$. There are two subcases:
        \begin{enumerate}
            \item \textbf{Subcase 1: $u\in \PP$}. When $u$ has at least one out-neighbor $x$ selected in $\SS$ or $u$ does not have any out-neighbor in $\GG_{t}$, then $u \in \PP$ and $\PP = \PP^\pr \cup \{u\}$, and for every pair $(w,p)$ in $D[t',\SS^\pr=\SS \setminus \{u\},\PP^\pr=\PP \setminus \{u\}]$ with $w + w_u \leq s$, we add $(w+w_u, p + \sum_{z \in \SS: (z,u) \in \GG, z \not\in \PP^\pr} p_z + p_u)$ and copy to $D[t,\SS,\PP]$.
            
            \item \textbf{Subcase 2: $u \notin \PP$}. When $u$ has a out-neighbor in $\GG_t$ but none of them are selected in $\SS$, then $u$ does not become profitable and thus does not belong to $\PP$. For every pair $(w,p)$ in $D[t',\SS^\pr=\SS \setminus \{u\},\PP^\pr=\PP \setminus \{u\}]$ with $w + w_u \leq s$, we add $(w+w_u, p + \sum_{z \in \SS: (z,u) \in \GG, z \not\in \PP^\pr} p_z )$ and copy to $D[t,\SS,\PP]$.
        \end{enumerate}
    \end{enumerate}
    Since $\ell(t') > \ell(t)$, by induction hypothesis all entries in $D[t', \SS'= \SS,\PP'=\PP]$, $D[t', \SS'= \SS \setminus \{u\},\PP'=\PP]$, $D[t', \SS'= \SS,\PP'=\PP \setminus \{u\}]$ and $D[t', \SS'= \SS \setminus \{u\},\PP'=\PP \setminus \{u\}]$ $\forall$ $\SS', \PP' \subseteq X_{t'}$ are already computed. We update pairs in $D[t,\SS,\PP]$ depending on the cases discussed above.
    
    \item \textbf{Forget Node.}  When $t$ is a forget node, there is a child $t'$ such that $X_t = X_{t'} \setminus \{u\} $. Let $\hat{\SS}$ be a set for which the \sor solution is attained in the definition of $D[t,\SS,\PP]$. If $u \not\in \hat{\SS}$, then $\hat{\SS}$ is one of the sets considered in the definition of $D[t',\SS,\PP]$. And if $u \in \hat{\SS}$, then $\hat{\SS}$ is one of the sets considered in the definition of $D[t',\SS \cup \{u\},\PP]$ and $D[t',\SS \cup \{u\},\PP \cup \{u\}]$. Since $\ell(t') > \ell(t)$, by induction hypothesis all entries in $D[t',\SS' = \SS,\PP' = \PP]$,  $D[t',\SS' = \SS \cup \{u\},\PP' = \PP]$, and $D[t',\SS' = \SS\cup \{u\} ,\PP' = \PP\cup \{u\}]$ $\forall$ $\SS', \PP' \subseteq X_{t'}$ are already computed and feasible. We copy each undominated $(w,p)$ pair stored in $D[t',\SS' = \SS,\PP' = \PP]$,  $D[t',\SS' = \SS \cup \{u\},\PP' = \PP]$, and $D[t',\SS' = \SS\cup \{u\} ,\PP' = \PP\cup \{u\}]$ to $D[t,\SS,\PP]$.

    \item \textbf{Join node.}  When $t$ is a join node, there are two children $t_1$ and $t_2$ of $t$, such that $X_t = X_{t_1} = X_{t_2}$. Let $\hat{\SS}$ be a set for \sor attained in the definition of $D[t,\SS,\PP]$. Let $\hat{\SS}_1 = \hat{\SS} \cap V_{t_1}$ and $\hat{\SS}_2 = \hat{\SS} \cap V_{t_2}$. Observe that $\hat{\SS}_1$ is a solution to \sor in $\GG_{t_1}$ and $\hat{\SS}_1 \cap X_{t_1} = \SS_1$, so this is considered in the definition of $D[t_1, \SS, \PP]$ and similarly, $\hat{\SS}_2 \cap X_{t_2} = \SS_2$. From the definition of nice tree decomposition we know that there is no edge between the vertices of $V_{t_1} \setminus X_t$ and $V_{t_2} \setminus X_t$. Then we merge solutions from the two subgraphs and remove the over-counting. By the induction hypothesis, the computed entries in $D[t_1,\SS_1, \PP_1]$ and $D[t_2,\SS_2, \PP_2]$ where $\SS_1 \cup \SS_2 = \SS$ and $\PP_1 \cup \PP_2 = \PP$ are correct and store the feasible and undominated \sor solutions for the subgraph $G_{t_1}$ in $\SS_1$ and similarly, $\SS_2$ for $G_{t_2}$. Now we add ($w_1 + w_2 - w(\SS_1 \cup \SS_2),  p_1+p_2 -p(\PP_1 \cap \PP_2)$) to $D[t,\SS, \PP]$.

\end{enumerate}

What remains to be shown is that an undominated feasible solution $\VV^\pr$ of \sor in $\GG$ is contained in $D[r,\{v\},\{v\}] \cup D[t,\{v\}, \emptyset]$. Let $w$ be the weight of $\VV^\pr$ and $p$ be the value subject to \sor. Recall that $v \in \VV^\pr$. For each $t$, we consider the subgraph $\GG_t \cap \VV^\pr$. Since the DP state $D[t,\SS,\PP]$ is updated correctly for all subgraphs $\GG_t$, the bottom up dynamic programming approach ensures that all feasible and undominated pairs are correctly propagated. If $\VV^\pr$ is a valid solution, then the corresponding (weight, profit) must be stored in some DP state. Since $v \in \VV^\pr$, the optimal states where $v$ is selected will store the pair $(w,p)$. Therefore, $D[r,\{v\},\{v\}] \cup D[t,\{v\}, \emptyset]$ contains the pair $(w,p)$. 

\textbf{Running time}: There are $n$ choices for the fixed vertex $v$. Upon fixing $v$ and adding it to each bag of $(\mathbb{T}, \mathcal{X})$ we consider the total possible number of states. For every node $t$, we have $2^{|X_t|}$ choices of ${\SS}$ and $2^{|X_t|}$ choices of ${\PP}$ for each choice of $\SS$. For each state, for each $w$, there can be at most one pair with $w$ as the first coordinate; similarly, for each $p$, there can be at most one pair with $p$ as the second coordinate. Thus, the number of undominated pairs in each $D[t,\SS,\PP]$ is at most ${\sf min}\{s,d\}$ time. Since the treewidth of the input graph \GG is at most \tw, it is possible to construct a data structure in time $\tw^{\OO(1)} \cdot n$ that allows performing adjacency queries in time $\OO(\tw)$.  For each node $t$, it takes time $\OO\left(4^{\tw} \cdot \tw^{\OO(1)} \cdot {\sf min}\{s^2,d^2\}\right)$ to compute all the values $D[t,\SS,\PP]$ and remove all undominated pairs. Since we can assume w.l.o.g that the number of nodes of the given tree decompositions is $\OO(\tw \cdot n)$, and there are $n$ choices for the vertex $v$, the running time of the algorithm is $\OO\left(4^{\tw}\cdot n^{\OO(1)} \cdot {\sf min}\{s^2,d^2\}\right)$.
\end{proof}}


\begin{theorem}\label{sand-di-fpt}
  There is an \FPT algorithm for directed \sand with running time $\OO\left(4^{\tw}\cdot n^{\OO(1)} \cdot {\sf min}\{s^2,d^2\}\right)$, where $n$ is the number of vertices in the input graph, $\tw$ is the treewidth of the input graph, $s$ is the input size of the knapsack and $d$ is the input target value.   
\end{theorem} 

\begin{proof}
We consider a nice tree decomposition of the underlying graph. Let $(\GG = (V,E),\{w_u\}_{u \in V_G}, \{p_u\}_{u \in V_G}, s,d)$ be an input instance of \sor such that $\tw=tw(\GG)$. Let $\VV^\pr$ be an optimal solution to \sand. For technical purposes, we guess a vertex $v \in \VV^\pr$ --- once the guess is fixed, we are only interested in finding solution subsets $\hat{\VV}$ that contain $v$ and $\VV^\pr$ is one such candidate. We also consider a nice edge tree decomposition $(\mathbb{T} = (V_{\mathbb{T}},E_{\mathbb{T}}),\mathcal{X})$ of $\GG$ that is rooted at a node $r$, and where $v$ has been added to all bags of the decomposition. Therefore, $X_{r} = \{v\}$ and each leaf bag is the singleton set $\{v\}$.

We define a function $\ell: V_{\mathbb{T}} \rightarrow \mathbb{N}$ as follows.
For a vertex $t \in V_\mathbb{T}$, $\ell(t) = \sf{dist}_\mathbb{T}(t,r)$, where $r$ is the root. Note that this implies that $\ell(r) = 0$. Let us assume that the values that $\ell$ take over the nodes of $\mathbb{T}$ are between $0$ and $L$. For a node $t\in V_\mathbb{T}$, we denote the set of vertices in the bags in the subtree rooted at $t$ by $V_t$ and $\GG_t=\GG[V_t]$. Now, we describe a dynamic programming algorithm over $(\mathbb{T},\mathcal{X})$. We have the following states.

    \textbf{States:} We maintain a DP table $D$ where a state has the following components:
    \begin{enumerate}
    \item $t$ represents a node in $V_\mathbb{T}$.
    \item $\SS$ represents a subset of the vertex subset $X_t$
    \item $\PP$ represents a subset of the vertex subset $\SS$ 
    \end{enumerate}
    \textbf{Interpretation of States}
For each node \( t \in \mathbb{T} \), we maintain a list \( D[t, \SS, \PP] \) for every subset \( \SS \subseteq X_t \) and \( \PP \subseteq \SS \). Each entry in this list holds a set of feasible undominated (weight, profit) pairs corresponding to valid \sand solutions $\hat{\SS}$ in the graph \( \GG_t \), where:  

\begin{itemize}  
    \item The selected vertices in \( \GG_t \) satisfy \(\hat{\SS} \cap X_t = \SS \). 
    \item The set of profit-contributing vertices is \( \PP \), meaning all vertices in \( \PP \) are part of \( \SS \) and each has selected all out-neighbors in \( \GG_t \).  
    \item If no such valid solution exists, we set \( D[t, \SS, \PP] = \emptyset \).  
\end{itemize}  

For each state $D[t,\SS,\PP]$, we initialize $D[t,\SS,\PP]$ to the list $\{(0,0)\}$.

    \textbf{Dynamic Programming on $D$}: We update the table $D$ as follows. We initialize the table of states with nodes $t\in V_\mathbb{T}$ such that $\ell(t)=L$. When all such states are initialized, then we move to update states where the node $t$ has $\ell(t) = L-1$, and so on, till we finally update states with $r$ as the node --- note that $\ell(r) =0$. For a particular $j$, $0\leq j< L$ and a state $D[t,\SS,\PP]$ such that $\ell(t) = j$, we can assume that $D[t',\SS',\PP']$ have been computed for all $t'$ such that $\ell(t')>j$ and all subsets $\SS'$ of $X_{t'}$. Now we consider several cases by which $D[t,\SS,\PP]$ is updated based on the nature of $t$ in $\mathbb{T}$:
    
    \begin{enumerate}
    \item \textbf{Leaf node.} Suppose $t$ is a leaf node with $X_{t} = \{v\}$. Then the list stored in $D[t,\SS,\PP]$ depends on the sets $\SS$ and $\PP$.  If $\PP$ is not a subset of $\SS$, we store $D[t,\SS,\PP] = \emptyset$; hence assume otherwise. The only possible cases that can arise are as follows:
    \begin{enumerate}
        \item $D[t,\phi,\phi]$ stores the pair $\{(0, 0)\}$
        \item $D[t,\{v\},\{v\}]$ = $\{(w_v, p_v)\}$, if $w_v \leq s$ 
        \item $D[t,\{v\},\phi]$ stores the pair $\{(w_v, 0)\}$, if $w_v \leq s$
    \end{enumerate} 

    \item  \textbf{Introduce node.} Suppose $t$ is an introduce node. Then it has only one child $t'$ where $X_{t'} \subset X_{t}$ and there is exactly one vertex $u \neq v$ that belongs to $X_{t}$ but not $X_{t'}$. Then for all $\SS \subseteq X_t$ and all $\PP \subseteq X_t$ , if $\PP$ is not a subset of $\SS$, we store $D[t,\SS,\PP] = \emptyset$; hence assume otherwise. 
    \begin{enumerate}
        \item If $u \not\in \SS$ and therefore $u \not\in \PP$, then 
        \begin{enumerate}
                \item if $\exists w : w \in \PP$ and $w \in N^-(u)$, we store for $D[t,\SS,\PP]$ = $\emptyset$.
                \item Otherwise, we copy all pairs of $D[t',\SS^\pr = \SS,\PP^\pr=\PP]$ to $D[t',\SS,\PP]$.
        \end{enumerate}

        \item If $u \in \SS$,
        \begin{enumerate}
            \item if $u$ does not have a out-neighbor $x$ in $X_t \setminus \SS$ i.e. \(N^+(u) \cap \VV_t = \emptyset\), then $u \in \PP$, and $\PP = \PP^\pr \cup \{u\}$, then for every pair $(w,p)$ in $D[t',\SS^\pr=\SS \setminus \{u\},\PP^\pr=\PP \setminus \{u\}]$ with $w + w_u \leq s$, we add $(w+w_u, p + p_u)$ and copy to $D[t,\SS,\PP]$.
            
            \item if $u$ has a out-neighbor $x$ in $X_t \setminus \SS$, then $u \not\in \PP$, and $\PP = \PP^\pr$, then for every pair $(w,p)$ in $D[t',\SS^\pr=\SS \setminus \{u\},\PP^\pr=\PP \setminus \{u\}]$ with $w + w_u \leq s$, we add $(w+w_u, p + 0)$ and copy to $D[t,\SS,\PP]$.
        \end{enumerate}
        \item Otherwise, we store $D[t,\SS,\PP] = \emptyset$
    \end{enumerate}
    
    \item  \textbf{Forget node.} Suppose $t$ is a forget vertex node. Then it has only one child $t'$, and there is a vertex $u \neq v$ such that $X_t\cup\{u\} = X_{t'}$. Then for all $\SS \subseteq X_t$ and all $\PP \subseteq X_t$ , if $\PP$ is not a subset of $\SS$, we store $D[t,\SS,\PP] = \emptyset$; hence assume otherwise. Then for all $\SS \subseteq X_t$:
    \begin{enumerate}
        \item we copy all feasible undominated pairs stored in $D[t',\SS\cup \{u\},\PP\cup \{u\}]$ to $D[t,\SS,\PP]$ if $u \in \SS^\pr$ and $u \in \PP^\pr$
        \item or we copy all feasible undominated pairs stored in $D[t',\SS\cup \{u\},\PP]$ to $D[t,\SS,\PP]$ if $u \in \SS^\pr$ and $u \not\in \PP^\pr$,
        \item or we copy all feasible undominated pairs stored in $D[t',\SS,\PP]$ to $D[t,\SS,\PP]$ otherwise. 
    \end{enumerate}

    \item \textbf{Join node.} Suppose $t$ is a join node. Then it has two children $t_1,t_2$ such that $X_t = X_{t_1} = X_{t_2}$. Then for all $\SS \subseteq X_t$, let $(w(\SS), p(\SS))$ be the total weight and value of the vertices in $\SS$ Consider a pair $(w_1,p_1)$ in $D[t_1,\SS_1,\PP_1]$ and a pair $(w_2,p_2)$ in $D[t_2,\SS_2,\PP_2]$ where $\SS_1 \cup \SS_2 = \SS$ and $\PP_1 \cap \PP_2 = \PP$. Suppose $w_1 + w_2 - w(\SS_1 \cup \SS_2) \leq s$, then we add $ \left( w_1 + w_2 - w(\SS_1 \cup \SS_2), p_1+p_2- p(\PP_1 \cap \PP_2) \right)$ to $D[t,\SS,\PP]$.
    \end{enumerate}
    
Finally, in the last step of updating $D[t,\SS,\PP]$, we go through the list saved in $D[t,\SS,\PP]$ and only keep undominated pairs. 

The output of the algorithm is a pair $(w,p)$ that is maximum of those stored in $D[r,\{v\},\{v\}]$ and $D[r,\{v\},\emptyset]$ such that $w \leq s$ and $p$ is the maximum value over all pairs in $D[r,\{v\},\{v\}]$ and $D[r,\{v\},\emptyset]$.

We refer to the Appendix for the proof of correctness and runtime analysis.

\longversion{\textbf{Proof of correctness}: 
Recall that we are looking for a solution that is a set of vertices $\VV^\pr$ that contains the fixed vertex $v$ that belongs to all bags of the nice tree decomposition. In each state we maintain a list of feasible and undominated (weight, profit) pairs that correspond to the solution $\hat{\SS}$ of $\GG_t$ and  satisfy \( \hat{\SS} \cap X_t = \SS \) and the set of profit contributing vertices is $\PP$.

We now show that the update rules holds for each $X_t$. To prove this formally, we need to consider the cases of what $t$ can be:

\begin{enumerate}
    \item \textbf{Leaf node.} Recall that in our modified nice tree decomposition we have added a vertex $v$ to all the bags. Suppose a leaf node $t$ contains a single vertex $v$, $D[t,\phi,\phi]$ stores the pair $\{(0, 0)\}$, $D[t,\{v\},\{v\}]$ = $\{(w_v, p_v)\}$ if $w_v \leq s$, $D[t,\{v\},\phi]$ stores the pair $\{(w_v, 0)\}$ if $w_v \leq s$, otherwise $\emptyset$. This is true in particular when $j = L$, the base case. From now we can assume that for a node $t$ with $\ell(t) = j < L$ and all subsets $\SS, \PP \subseteq X_t$, $D[t',\SS^\pr,\PP^\pr]$ entries are correct and correspond to a \sor solution in $\GG_t$. when $\ell(t') > j$.

    \item \textbf{Introduce node.} When $t$ is an introduce node, there is a child $t'$ such that $X_t = X_{t'} \cup \{u\}$. We are introducing a vertex $u$ and the edges associated with it in $\GG_t$. When $\PP$ is not a subset of $\SS$, the state $D[t,\SS,\PP] = \emptyset$ because it violates the definition of $\PP$; hence assume otherwise. Let us prove for each case.
    \begin{enumerate}
        \item When $u$ is not included in $\SS$, it is also not in $\PP$ and the families of sets $\hat{\SS}$ considered in the definition of $D[t,\SS,\PP]$ and of $D[t',\SS,\PP]$ are equal. Notice that there may be a vertex $w \in X_{t}$ that belongs to $\PP$ such that $N^+(w) \cap V_{t'} = \emptyset$, i.e. $w$ did not have any out-neighbor in the graph $\GG_{t'}$. However, with the introduction of the vertex $u$ and the edges associated to it, if $N^+(w) \cap V_{t} = \{u\}$ i.e. $u$ is the only out-neighbor of $w$ in $\GG_t$, since $u$ does not belong to $\SS$, the \sand constraint is violated and hence the state $D[t,\SS,\PP]$ stores $\emptyset$ as it is no longer feasible. Otherwise we copy all pairs of $D[t',\SS^\pr = \SS,\PP^\pr=\PP]$ to $D[t,\SS,\PP]$.

        \item Now consider the case when $u$ is part of $\SS$. Let $\hat{\SS}$ be a feasible solution for \sand attained in the definition of $D[t,\SS,\PP]$. Then it follows that $\hat{\SS} \setminus \{u\}$ is one of the sets considered in the definition of $D[t',\SS \setminus \{v\},\PP \setminus \{v\}]$. Now we must compute entries depending on whether $u$ contributes to profit or not. It suffices to just check if $u$ has a neighbor in $X_t$ or not because the nice tree decomposition ensures that the newly introduced vertex $u$ can have neighbors only in $X_t$. There are two subcases:
        \begin{enumerate}
            \item \textbf{Subcase 1: $u\in \PP$}. When $u$ has all out-neighbors $x$ selected in $\SS$ or $u$ does not have any out-neighbor in $\GG_{t}$, then $u \in \PP$ and $\PP = \PP^\pr \cup \{u\}$, and then for every pair $(w,p)$ in $D[t',\SS^\pr=\SS \setminus \{u\},\PP^\pr=\PP \setminus \{u\}]$ with $w + w_u \leq s$, we add $(w+w_u, p + p_u)$ and copy to $D[t,\SS,\PP]$.
            
            \item \textbf{Subcase 2: $u \notin \PP$}. When $u$ has at least one out-neighbor in $\GG_t$ that is not selected are selected in $\SS$, then $u$ does not become profitable and thus does not belong to $\PP$. For then for every pair $(w,p)$ in $D[t',\SS^\pr=\SS \setminus \{u\},\PP^\pr=\PP \setminus \{u\}]$ with $w + w_u \leq s$, we add $(w+w_u, p + 0)$ and copy to $D[t,\SS,\PP]$.
        \end{enumerate}
    \end{enumerate}
    Since $\ell(t') > \ell(t)$, by induction hypothesis all entries in $D[t', \SS'= \SS,\PP'=\PP]$, $D[t', \SS'= \SS,\PP'=\PP \setminus \{u\}]$ and $D[t', \SS'= \SS \setminus \{u\},\PP'=\PP \setminus \{u\}]$ $\forall$ $\SS', \PP' \subseteq X_{t'}$ are already computed. We update pairs in $D[t,\SS,\PP]$ depending on the cases discussed above.

    \item \textbf{Forget Node.}  When $t$ is a forget node, there is a child $t'$ such that $X_t = X_{t'} \setminus \{u\} $. Let $\hat{\SS}$ be a set for which the \sand solution is attained in the definition of $D[t,\SS,\PP]$. If $u \not\in \hat{\SS}$, then $\hat{\SS}$ is one of the sets considered in the definition of $D[t',\SS,\PP]$. And if $u \in \hat{\SS}$, then $\hat{\SS}$ is one of the sets considered in the definition of $D[t',\SS \cup \{u\},\PP]$ and $D[t',\SS \cup \{u\},\PP \cup \{u\}]$. Since $\ell(t') > \ell(t)$, by induction hypothesis all entries in $D[t',\SS' = \SS,\PP' = \PP]$,  $D[t',\SS' = \SS \cup \{u\},\PP' = \PP]$, and $D[t',\SS' = \SS\cup \{u\} ,\PP' = \PP\cup \{u\}]$ $\forall$ $\SS', \PP' \subseteq X_{t'}$ are already computed and feasible. We copy each undominated $(w,p)$ pair stored in $D[t',\SS' = \SS,\PP' = \PP]$,  $D[t',\SS' = \SS \cup \{u\},\PP' = \PP]$, and $D[t',\SS' = \SS\cup \{u\} ,\PP' = \PP\cup \{u\}]$ to $D[t,\SS,\PP]$ depending on whether $u$ belonged to $\SS$ or $\SS \cup \{u\}$.

    \item \textbf{Join node.}  When $t$ is a join node, there are two children $t_1$ and $t_2$ of $t$, such that $X_t = X_{t_1} = X_{t_2}$. Let $\hat{\SS}$ be a set for \sor attained in the definition of $D[t,\SS,\PP]$. Let $\hat{\SS}_1 = \hat{\SS} \cap V_{t_1}$ and $\hat{\SS}_2 = \hat{\SS} \cap V_{t_2}$. Observe that $\hat{\SS}_1$ is a solution to \sand in $\GG_{t_1}$ and $\hat{\SS}_1 \cap X_{t_1} = \SS_1$, so this is considered in the definition of $D[t_1, \SS, \PP]$ and similarly, $\hat{\SS}_2 \cap X_{t_2} = \SS_2$. From the definition of nice tree decomposition we know that there is no edge between the vertices of $V_{t_1} \setminus X_t$ and $V_{t_2} \setminus X_t$. Then we merge solutions from the two subgraphs and remove the over-counting. By the induction hypothesis, the computed entries in $D[t_1,\SS_1, \PP_1]$ and $D[t_2,\SS_2, \PP_2]$ where $\SS_1 \cup \SS_2 = \SS$ and $\PP_1 \cap \PP_2 = \PP$ are correct and store the feasible and undominated \sand solutions for the subgraph $G_{t_1}$ in $\SS_1$ and similarly, $\SS_2$ for $G_{t_2}$. Now we add ($w_1 + w_2 - w(S_1 \cup S_2),  p_1+p_2 -p(\PP_1 \cap \PP_2)$) to $D[t,\SS, \PP]$.
\end{enumerate}

What remains to be shown is that an undominated feasible solution $\VV^\pr$ of \sand in $\GG$ is contained in $D[r,\{v\},\{v\}] \cup D[t,\{v\}, \emptyset]$. Let $w$ be the weight of $\VV^\pr$ and $p$ be the value subject to \sand. Recall that $v \in \VV^\pr$. For each $t$, we consider the subgraph $\GG_t \cap \VV^\pr$. Since the DP state $D[t,\SS,\PP]$ is updated correctly for all subgraphs $\GG_t$, the bottom up dynamic programming approach ensures that all feasible and undominated pairs are correctly propagated. If $\VV^\pr$ is a valid solution, then the corresponding (weight, profit) must be stored in some DP state. Since $v \in \VV^\pr$, the optimal states where $v$ is selected will store the pair $(w,p)$. Therefore, $D[r,\{v\},\{v\}] \cup D[t,\{v\}, \emptyset]$ contains the pair $(w,p)$. 

\textbf{Running time}: There are $n$ choices for the fixed vertex $v$. Upon fixing $v$ and adding it to each bag of $(\mathbb{T}, \mathcal{X})$ we consider the total possible number of states. For every node $t$, we have $2^{|X_t|}$ choices of ${\SS}$ and $2^{|X_t|}$ choices of ${\PP}$ for each choice of $\SS$. For each state, for each $w$, there can be at most one pair with $w$ as the first coordinate; similarly, for each $p$, there can be at most one pair with $p$ as the second coordinate. Thus, the number of undominated pairs in each $D[t,\SS,\PP]$ is at most ${\sf min}\{s,d\}$ time. Since the treewidth of the input graph \GG is at most \tw, it is possible to construct a data structure in time $\tw^{\OO(1)} \cdot n$ that allows performing adjacency queries in time $\OO(\tw)$.  For each node $t$, it takes time $\OO\left(4^{\tw} \cdot \tw^{\OO(1)} \cdot {\sf min}\{s^2,d^2\}\right)$ to compute all the values $D[t,\SS,\PP]$ and remove all undominated pairs. Since we can assume w.l.o.g that the number of nodes of the given tree decompositions is $\OO(\tw \cdot n)$, and there are $n$ choices for the vertex $v$, the running time of the algorithm is $\OO\left(4^{\tw}\cdot n^{\OO(1)} \cdot {\sf min}\{s^2,d^2\}\right)$.}
\end{proof}

\begin{theorem}\label{sor-fpt-cc}
    There exists a randomized color coding based $\OO^\star\left(e^b\cdot b^b\cdot 2^{b^2}\cdot {\sf min}\{s^2,d^2\}\right)$ pseudo-\FPT algorithm for \sor parameterized by solution size (budget) $b$.
\end{theorem}

\begin{proof}
We color the edges of the given graph uniformly at random, with the goal that each vertex in the optimal solution becomes identifiable through incident edges that are distinctly colored. 

The solution to the directed \sor problem is a subset of vertices \(\SS \subseteq \VV\) such that the subgraph \(\GG[\SS]\) satisfies certain structural constraints. Specifically, any optimal solution must have the following structure:
\begin{itemize}
    \item The subgraph \(\GG[\SS]\) is a forest of disjoint rooted arborescences,
    \item Each vertex in \(\SS\) has at most one outgoing edge in \(\GG[\SS]\),
    \item The underlying undirected subgraph induced by \(\SS\) is connected within each component, and each component contains at least two vertices.
\end{itemize}

These structural properties follow directly from the semantics of the \sor problem: a vertex contributes profit only if it has an outgoing neighbor (i.e., is internal in the induced subgraph) or if it is a sink (no out-neighbors at all). To handle sink vertices without loss of generality, we assume the input graph has no sinks. This can be achieved by adding a dummy vertex of zero weight and zero profit as an out-neighbor to every sink. In the worst case, if the solution contains \(b\) sink vertices, we require at most \(b\) dummy vertices. Hence, we can simulate the original instance on a graph with at most \(2b\) vertices under a budget of \(2b\), while preserving correctness. For simplicity, we proceed under the assumption that no sink vertices are present in \(\GG\). Moreover, we also assume without loss of generality that each vertex has exactly one outgoing edge in the subgraph \(\GG[\SS]\). This assumption is justified because adding multiple out-edges from the same vertex does not increase its profit---a single out-edge suffices to make the vertex internal, and hence, profit-contributing. Since only the existence of an out-neighbor matters (not the count), limiting each vertex to one out-edge suffices to identify both the end points of the edge.

Color the vertices uniformly at random from $\{1,2,\dots,b\}$ and guess $k$ partitions of [b]: $\CC_1, \CC_2, \CC_3, \dots, \CC_k$ such that each partition is a directed component of size at least 2 and no two partitions have the same color class i.e. let the set of colors in $\CC_i$ be denoted by $\chi(\CC_i)$, then $\chi(\CC_i)$ $\cap$ $\chi(\CC_j)$ = $\emptyset$, $\forall i, j \in [k]$ where $i \neq j$. We iterate over all possible choices of $k$. Note that $1 \leq k \leq \frac{b}{2}$. For the rest of the algorithm we assume that we have rightly guessed such partitions with distinct colors.

Consider a partition $\CC_i$. Let $\tau(\CC_i)$ denote the set of all trees in $\CC_i$. We design a dynamic programming based algorithm to find a colorful tree with the exact structure and if this is the case, then the algorithm returns all feasible and undominated pairs of the colorful tree. There may be multiple witness of such a tree however we maintain only undominated pairs. 

\textbf{Dynamic Programming} We apply dynamic programming: for a non empty tree $T(\XX)$ and an edge $e \in \bigcup_{i \in [k]}\CC_i[\EE]$, we define the state $D[T(\XX),e]$ to store a list of all the feasible undominated weight and profit pairs i.e ${(w,p)}$ corresponding to \XX-colorful tree $T(\XX) \in \tau(\CC_i)$ having edge $e$.

$D[T(\XX),e]$ should store all undominated feasible weight-profit pairs for a colorful tree $T(\XX) \in \tau(\CC_i)$ that uses each color in \XX exactly once, contains $e$ as an edge, and has a structure dictated by the tree $T(\XX)$. Since each color appears exactly once in $T(\XX)$, there is a unique choice for edge $e$ corresponding to that color. 

To compute the DP state $D[T(\XX),e]$, which stores all undominated feasible weight-profit pairs for a colorful tree \( T(\XX) \), we observe that each color in \( \XX \) appears exactly once in \( T(\XX) \), meaning the edge \( e \) is uniquely determined by its color. A sanity check condition ensures that if the color of \( e \), \( \chi(e) \), does not belong to \( T(\XX) \), then \( D[T(\XX),e] = \emptyset \), as no valid tree structure can exist. We now describe the dynamic programming to compute feasible solutions for colorful tree structures.

\textbf{Base Case:} If \(|X| = 1\), then the tree \(T(\XX)\) consists of a single edge with color \(c\). For any edge \((u,v) \in E(\GG)\), we define:
\[
D[T(\{c\}), e = (u,v)] =
\begin{cases}
\{(w_u + w_v, p_u)\}, & \text{if } (u,v) \in E(\GG) \text{ has color } c, \\
\emptyset, & \text{otherwise.}
\end{cases}
\]
Here, we include the profit of \(u\) since it has an outgoing edge in the selected subgraph as a witness, and omit the profit of \(v\) because it remains a leaf in the tree.

\textbf{Transition Case:} For \(|\XX| > 1\), we recursively construct the tree \(T(\XX)\) as the union of two disjoint colorful subtrees \(T(\XX_1)\) and \(T(\XX_2)\), connected by an edge \((u,v)\) of color \(c\). Notice that this edge uniquely defines the structure of $T(\XX)$. Assume \(T(\XX_1)\) is rooted at \(u\) and \(T(\XX_2)\) at \(v\). Then for each such decomposition, we compute:
\[
D[T(\XX), e=(u,v)] = \bigcup_{
\substack{(w_1, p_1) \in D[T(\XX_1), (u,x)] \\
(w_2, p_2) \in D[T(\XX_2), (v,y)]}}
\{(w_1 + w_2,\; p_1 + p_2 + p_u)\},
\]
where the union is taken over all valid combinations of edges \((u,x)\) and \((v,y)\) in \(\GG\) that correspond to the root edges of \(T(\XX_1)\) and \(T(\XX_2)\), respectively. Since \(u\) now becomes internal by virtue of having an outgoing edge, its profit \(p_u\) is included in the final pair. After combining all feasible options, we retain only the undominated \((w, p)\) pairs in the resulting state \(D[T(\XX), e=(u,v)]\).

For each valid partition of the color set into components \(\CC_1, \CC_2, \dots, \CC_k\), we compute dynamic programming tables for all corresponding tree structures. Specifically, for each \(\CC_i\), we consider a tree structure \(T(\XX_i)\) such that \(|T(\XX_i)| = |\CC_i|\), and compute the table \(D[T(\XX_i), e]\) for all edges \(e\) realizing the root edge of the tree.

To combine the components into a global solution, we perform a knapsack-style dynamic programming step. For each pair \((w_1, p_1)\) stored in \(D[T(\XX_1), e]\) and each \((w_2, p_2)\) stored in \(D[T(\XX_2), e']\), we consider the sum \((w_1 + w_2, p_1 + p_2)\). If \(w_1 + w_2 \leq s\), we add this pair to the combined list. After this step, we retain only undominated pairs. Since in the worst case we can have at most one pair for each weight \(w \leq s\) and for each profit \(p \leq d\), the total number of pairs we maintain is at most \(\min(s, d)\). Therefore, the time complexity of combining two components is \(\mathcal{O}(\min(s^2, d^2))\). We repeat this merging step across all \(k\) components, combining them sequentially. As a result, the total time complexity of this final merge step is \(\mathcal{O}(\min(s^2, d^2) \cdot (k - 1))\).

We now show that the algorithm runs in time $\OO\left(e^b\cdot b^b\cdot 2^{b^2}\cdot n^{\OO(1)}\cdot {\sf min}\{s^2,d^2\}\right)$, and given a \yes instance, returns a solution with probability at least $e^{-b}$. By repeating the algorithm independently, $e^{b}$ times, we obtain the running time bound. The color coding technique assigns \( b \) colors uniformly at random to all \emph{edges} in the graph. The probability that a fixed solution of size \( b \) gets assigned distinct colors on its edges is at least \( e^{-b} \). To boost this probability to a constant (e.g., at least \( 1/2 \)), we repeat the random coloring \( e^b \) times.

For each coloring, we iterate over all possible ways to partition the \( b \) colors into disjoint subsets, where each subset corresponds to a directed component of size at least 2. The number of such partitions is upper bounded by \( b^b \). For each color subset (i.e., tree), we guess a rooted tree structure over those colors, and the number of possible rooted trees over at most \( b \) colors is bounded by \( 2^{b^2} \). For each such tree structure, we perform dynamic programming over all edges of the graph to compute all undominated feasible \((w,p)\) pairs. The number of such pairs is bounded by \( \min\{s, p\} \), since we only store undominated combinations under the budget constraint \( w \leq s \) and demand requirement \( p \geq d \). Each DP computation involves traversing the tree and combining child subtrees, which can be done in polynomial time per edge. Finally, we merge the \((w, p)\) pairs across all components using a knapsack-style DP, which takes \(\OO(\min\{s^2,d^2\} \cdot (k - 1))\) time for \(k\) components. Since \(k\) can vary over all values up to \(\lfloor b/2 \rfloor\), the number of such combinations is polynomial in \(n\) and is subsumed by the \(n^{\OO(1)}\) factor. Therefore, the overall running time is $\OO\left(e^b\cdot b^b\cdot 2^{b^2}\cdot n^{\OO(1)}\cdot {\sf min}\{s^2,d^2\}\right)$ as claimed.
\end{proof}

Using standard de-randomization techniques, we can convert the randomized color-coding based algorithm into a deterministic one.\longversion{ Instead of selecting a random coloring of the edges from \([m] \to [b]\), we construct a family \(\FF\) of functions \(f : [m] \rightarrow [b]\) such that for every subset of at most \(b\) edges, there exists a function \(f \in \FF\) that assigns distinct colors to the edges in the subset. Such a family is called an \((m, b)\)-perfect hash family and can be constructed deterministically of size \(\OO(b^{\OO(\log b)} \cdot \log m)\).

We then iterate over each \(f \in \FF\), use it to color the edges of the graph, and invoke the dynamic programming algorithm for each coloring. The guarantee provided by the perfect hash family ensures that, if there exists a feasible solution of size at most \(b\), then for some \(f \in \FF\), the edges of that solution are colored distinctly, and the algorithm will find it. This derandomization step adds a multiplicative overhead of \(b^{\OO(\log b)}\) to the overall running time. We obtain the following theorem.}

\begin{theorem}\label{sor-fpt-cc-det}
    There exists a deterministic color coding based $\OO^\star\left(e^b \cdot b^{\OO(\log b)} \cdot b^b \cdot 2^{b^2} \cdot {\sf min}\{s^2,d^2\}\right)$ pseudo-\FPT algorithm for \sor parameterized by solution size (budget) $b$.
\end{theorem}

We now observe that the \sor problem is also fixed-parameter tractable when parameterized by the demand $d$ as parameter. Since each vertex contributes a non-negative profit value and can contribute only if it either has at least one out-neighbor in the selected subset or has no out-neighbors at all, the number of vertices required to achieve profit at least \(d\) is bounded. In the worst case, if each vertex contributes only a unit profit, then selecting at most \(2d\) vertices suffices to obtain total profit at least \(d\), as some vertices may not directly contribute but enable others to do so through outgoing edges. Hence, without loss of generality, we can assume that any optimal solution has size at most \(b \leq 2d\). 

By applying the deterministic algorithm of \Cref{sor-fpt-cc-det} with \(b = 2d\), we obtain an algorithm whose running time is bounded as a function of \(d\), leading to the following result.

\begin{theorem}\label{sor-fpt-cc-det-profit}
    There exists a deterministic color coding based $\OO^\star\left(e^{2d} \cdot (2d)^{\OO(\log d)} \cdot (2d)^{2d} \cdot 2^{(2d)^2} \cdot {\sf min}\{s^2,d^2\}\right)$ pseudo-\FPT algorithm for \sor parameterized by the demand \(d\).
\end{theorem}

We observe that all the color coding based results developed for the directed \sor problem also extend directly to the directed \hor variant. In both problems,  the underlying structural properties of feasible solutions—namely, that the induced subgraph consists of components where each vertex has at most one outgoing edge and each component forms a rooted arborescence of size at least two—remain unchanged. Since the color coding framework operates by encoding such structures through distinct colors and enumerating tree-shaped patterns over these colors, the dynamic programming approach and the associated correctness and runtime guarantees hold identically for \hor. Thus, we obtain the following corollaries for \hor.

\begin{corollary}\label{hor-fpt-cc}
    There exists a randomized color coding based $\OO^\star\left(e^b\cdot b^b\cdot 2^{b^2}\cdot {\sf min}\{s^2,d^2\}\right)$ pseudo-\FPT algorithm for \hor parameterized by solution size (budget) $b$.
\end{corollary}

\begin{corollary}\label{hor-fpt-cc-det}
    There exists a deterministic color coding based $\OO^\star\left(e^b \cdot b^{\OO(\log b)} \cdot b^b \cdot 2^{b^2} \cdot {\sf min}\{s^2,d^2\}\right)$ pseudo-\FPT algorithm for \hor parameterized by solution size (budget) $b$.
\end{corollary}

\begin{corollary}\label{hor-fpt-cc-det-profit}
    There exists a deterministic color coding based $\OO^\star\left(e^{2d} \cdot (2d)^{\OO(\log d)} \cdot (2d)^{2d} \cdot 2^{(2d)^2} \cdot {\sf min}\{s^2,d^2\}\right)$ pseudo-\FPT algorithm for \hor parameterized by the demand \(d\).
\end{corollary}

%% file: fpthard.tex
\begin{theorem}\label{hor-di-fpt}
($\star$)  There is an \FPT algorithm for directed \hor with running time $\OO^\star\left(2^{\tw}\cdot {\sf min}\{s^2,d^2\}\right)$, where $n$ is the number of vertices in the input graph, $\tw$ is the treewidth of the input graph, $s$ is the input size of the knapsack and $d$ is the input target value. 
\end{theorem} 

\longversion{\begin{proof}
We consider a nice tree decomposition of the underlying graph. Let $(\GG = (V,E),\{w_u\}_{u \in V_G}, \{p_u\}_{u \in V_G}, s,d)$ be an input instance of \hor such that $\tw=tw(\GG)$. Let $\VV^\pr$ be an optimal solution to \hor. For technical purposes, we guess a vertex $v \in \VV^\pr$ --- once the guess is fixed, we are only interested in finding solution subsets $\hat{\VV}$ that contain $v$ and $\VV^\pr$ is one such candidate. We also consider a nice edge tree decomposition $(\mathbb{T} = (V_{\mathbb{T}},E_{\mathbb{T}}),\mathcal{X})$ of $\GG$ that is rooted at a node $r$, and where $v$ has been added to all bags of the decomposition. Therefore, $X_{r} = \{v\}$ and each leaf bag is the singleton set $\{v\}$.

We define a function $\ell: V_{\mathbb{T}} \rightarrow \mathbb{N}$ as follows.
For a vertex $t \in V_\mathbb{T}$, $\ell(t) = \sf{dist}_\mathbb{T}(t,r)$, where $r$ is the root. Note that this implies that $\ell(r) = 0$. Let us assume that the values that $\ell$ take over the nodes of $\mathbb{T}$ are between $0$ and $L$. For a node $t\in V_\mathbb{T}$, we denote the set of vertices in the bags in the subtree rooted at $t$ by $V_t$ and $\GG_t=\GG[V_t]$. Now, we describe a dynamic programming algorithm over $(\mathbb{T},\mathcal{X})$. We have the following states.

    \textbf{States:} We maintain a DP table $D$ where a state has the following components:
    \begin{enumerate}
    \item $t$ represents a node in $V_\mathbb{T}$.
    \item $\SS$ represents a subset of the vertex subset $X_t$
    \end{enumerate}
    \textbf{Interpretation of States}
For each node \( t \in \mathbb{T} \), we maintain a list \( D[t, \SS] \) for every subset \( \SS \subseteq X_t \). Each entry in this list holds a set of feasible and undominated (weight, profit) pairs corresponding to valid \hor solutions $\hat{\SS}$ in the graph \( \GG_t \), where:  

\begin{itemize}  
    \item The selected vertices in \( \GG_t \) satisfy \(\hat{\SS} \cap X_t = \SS \).   
    \item If no such valid solution exists, we set \( D[t, \SS] = \emptyset \).  
\end{itemize}  

For each state $D[t,\SS]$, we initialize $D[t,\SS]$ to the list $\{(0,0)\}$.

    \textbf{Dynamic Programming on $D$}: We update the table $D$ as follows. We initialize the table of states with nodes $t\in V_\mathbb{T}$ such that $\ell(t)=L$. When all such states are initialized, then we move to update states where the node $t$ has $\ell(t) = L-1$, and so on, till we finally update states with $r$ as the node --- note that $\ell(r) =0$. For a particular $j$, $0\leq j< L$ and a state $D[t,\SS]$ such that $\ell(t) = j$, we can assume that $D[t',\SS']$ have been computed for all $t'$ such that $\ell(t')>j$ and all subsets $\SS'$ of $X_{t'}$. Now we consider several cases by which $D[t,\SS]$ is updated based on the nature of $t$ in $\mathbb{T}$:
    
    \begin{enumerate}
    \item \textbf{Leaf node.} Suppose $t$ is a leaf node with $X_{t} = \{v\}$. The only possible cases that can arise are as follows:
    \begin{enumerate}
        \item $D[t,\phi]$ stores the pair $\{(0, 0)\}$
        \item $D[t,\{v\}]$ = $\{(w_v, p_v)\}$, if $w_v \leq s$
    \end{enumerate} 
    
    \item \textbf{Introduce node.} Suppose $t$ is an introduce node. Then it has only one child $t'$ where $X_{t'} \subset X_{t}$ and there is exactly one vertex $u \neq v$ that belongs to $X_{t}$ but not $X_{t'}$. Then for all $\SS \subseteq X_t$:
    \begin{enumerate}
        \item If $u \not\in \SS$, 
        \begin{enumerate}
            \item if \(\exists w: w\in N^-(u), w \in \SS \) and \(N^+(w) \cap \VV_t = \{u\}\), then $D[t,\SS] = \emptyset$, hence assume otherwise.
            \item Otherwise, we copy all pairs of $D[t',\SS^\pr = \SS \setminus \{u\}]$ to $D[t',\SS,]$. 
        \end{enumerate}
        
        \item If $u \in \SS$,
        \begin{enumerate}            
            \item if $u$ has out-neighbor $x$ in $\SS$ or \(N^+(u) \cap \VV_t = \emptyset\), then for every pair $(w,p)$ in $D[t',\SS^\pr=\SS \setminus \{u\}]$ with $w + w_u \leq s$, we add $(w+w_u, p + p_u)$ and copy to $D[t,\SS]$.
            
            \item Otherwise, $D[t,\SS]$ = $\emptyset$

        \end{enumerate}
    \end{enumerate}
    
    \item \textbf{Forget node.} Suppose $t$ is a forget vertex node. Then it has only one child $t'$, and there is a vertex $u \neq v$ such that $X_t\cup\{u\} = X_{t'}$. Then for all $\SS \subseteq X_t$ we copy all feasible undominated pairs stored in $D[t',\SS\cup \{u\}]$ and $D[t',\SS]$ to $D[t,\SS]$. We remove any dominated pair and maintain only the undominated pairs in $D[t,\SS]$.

    \item \textbf{Join node.} Suppose $t$ is a join node. Then it has two children $t_1,t_2$ such that $X_t = X_{t_1} = X_{t_2}$. Then for all $\SS \subseteq X_t$, let $(w(\SS), p(\SS))$ be the total weight and value of the vertices in $\SS$. Consider a pair $(w_1,p_1)$ in $D[t_1,\SS_1]$ and a pair $(w_2,p_2)$ in $D[t_2,\SS_2]$, where $\SS_1 \cup \SS_2 = \SS$. Then for all pairs $(w_1,p_1) \in D[t_1,\SS_1]$ and $(w_2,p_2) \in D[t_2,\SS_2]$, if $w_1 + w_2 - w(\SS_1 \cup \SS_2) \leq s$, then we add $ \left( w_1 + w_2 - w(\SS_1 \cup \SS_2), p_1+p_2- p(\SS_1 \cup \SS_2) \right)$ and copy to $D[t,\SS]$.
    \end{enumerate}
    
Finally, in the last step of updating $D[t,\SS]$, we go through the list saved in $D[t,\SS]$ and only keep undominated pairs.

The output of the algorithm is a pair $(w,p)$ that is maximum of those stored in $D[r,\{v\}]$ such that $w \leq s$ and $p$ is the maximum value over all pairs in $D[r,\{v\}]$.

\textbf{Proof of correctness}: 
Recall that we are looking for a solution that is a set of vertices $\VV^\pr$ that contains the fixed vertex $v$ that belongs to all bags of the nice tree decomposition. In each state we maintain a list of feasible and undominated (weight, profit) pairs that correspond to the solution $\hat{\SS}$ of $\GG_t$ and  satisfy \( \hat{\SS} \cap X_t = \SS \).

We now show that the update rules holds for each $X_t$. To prove this formally, we need to consider the cases of what $t$ can be:

\begin{enumerate}
    \item \textbf{Leaf node.} Recall that in our modified nice tree decomposition we have added a vertex $v$ to all the bags. Suppose a leaf node $t$ contains a single vertex $v$, $D[t,\phi]$ stores the pair $\{(0, 0)\}$ and $D[t,\{v\}]$ = $\{(w_v, p_v)\}$ if $w_v \leq s$. This is true in particular when $j = L$, the base case. From now we can assume that for a node $t$ with $\ell(t) = j < L$ and all subsets $\SS \subseteq X_t$, $D[t',\SS^\pr]$ entries are correct and correspond to a \hor solution in $\GG_t$. when $\ell(t') > j$. 

    \item \textbf{Introduce node.} When $t$ is an introduce node, there is a child $t'$ such that $X_t = X_{t'} \cup \{u\}$. We are introducing a vertex $u$ and the edges associated with it in $\GG_t$. Let us prove for each case.
    \begin{enumerate}
        \item When $u$ is not included in $\SS$, the families of sets $\hat{\SS}$ considered in the definition of $D[t,\SS]$ and of $D[t',\SS]$ are equal. Notice that there may be a vertex $w \in X_{t}$ that belongs to $\SS$ such that $N^+(w) \cap V_{t'} = \emptyset$, i.e. $w$ did not have any out-neighbor in the subgraph $\GG_{t'}$. However, with the introduction of the vertex $u$ and the edges associated to it, if $N^+(w) \cap V_{t} = \{u\}$ i.e. $u$ is the only out-neighbor of $w$ in $\GG_t$, since $u$ does not belong to $\SS$, the \hor constraint is violated and hence the state $D[t,\SS]$ stores $\emptyset$ as it is no longer feasible. Otherwise we copy all pairs of $D[t',\SS^\pr = \SS]$ to $D[t,\SS]$.

        \item Now consider the case when $u$ is part of $\SS$. Let $\hat{\SS}$ be a feasible solution for \hor attained in the definition of $D[t,\SS]$. Then it follows that $\hat{\SS} \setminus \{u\}$ is one of the sets considered in the definition of $D[t',\SS \setminus \{u\}]$. It suffices to just check if $u$ has a neighbor in $X_t$ or not because the nice tree decomposition ensures that the newly introduced vertex $u$ can have neighbors only in $X_t$. There are two subcases:
        \begin{enumerate}
            \item \textbf{Subcase 1}. When $u$ has at least one out-neighbor $x$ selected in $\SS$ or $u$ does not have any out-neighbor in $\GG_{t}$, then for every pair $(w,p)$ in $D[t',\SS^\pr=\SS \setminus \{u\}]$ with $w + w_u \leq s$, we add $(w+w_u, p + p_u)$ and copy to $D[t,\SS]$.
            
            \item \textbf{Subcase 2}. When $u$ has a out-neighbor in $\GG_t$ but none of them are selected in $\SS$, then $D[t,\SS] = \emptyset$.
        \end{enumerate}
    \end{enumerate}
    Since $\ell(t') > \ell(t)$, by induction hypothesis all entries in $D[t', \SS'= \SS]$ and $D[t', \SS'= \SS \setminus \{u\}]$,  $\forall$ $\SS' \subseteq X_{t'}$ are already computed. We update pairs in $D[t,\SS]$ depending on the cases discussed above.
    
    \item \textbf{Forget Node.}  When $t$ is a forget node, there is a child $t'$ such that $X_t = X_{t'} \setminus \{u\} $. Let $\hat{\SS}$ be a set for which the \hor solution is attained in the definition of $D[t,\SS]$. If $u \not\in \hat{\SS}$, then $\hat{\SS}$ is one of the sets considered in the definition of $D[t',\SS]$. And if $u \in \hat{\SS}$, then $\hat{\SS}$ is one of the sets considered in the definition of $D[t',\SS \cup \{u\}]$. Since $\ell(t') > \ell(t)$, by induction hypothesis all entries in $D[t',\SS' = \SS]$ and $D[t',\SS' = \SS \cup \{u\}]$, $\forall$ $\SS' \subseteq X_{t'}$ are already computed and feasible. We copy each undominated $(w,p)$ pair stored in $D[t',\SS' = \SS]$ and $D[t',\SS' = \SS \cup \{u\}]$ to $D[t,\SS]$.

    \item \textbf{Join node.}  When $t$ is a join node, there are two children $t_1$ and $t_2$ of $t$, such that $X_t = X_{t_1} = X_{t_2}$. Let $\hat{\SS}$ be a set for \hor attained in the definition of $D[t,\SS]$. Let $\hat{\SS}_1 = \hat{\SS} \cap V_{t_1}$ and $\hat{\SS}_2 = \hat{\SS} \cap V_{t_2}$. Observe that $\hat{\SS}_1$ is a solution to \hor in $\GG_{t_1}$ and $\hat{\SS}_1 \cap X_{t_1} = \SS_1$, so this is considered in the definition of $D[t_1, \SS]$ and similarly, $\hat{\SS}_2 \cap X_{t_2} = \SS_2$. From the definition of nice tree decomposition we know that there is no edge between the vertices of $V_{t_1} \setminus X_t$ and $V_{t_2} \setminus X_t$. Then we merge solutions from the two subgraphs and remove the over-counting. By the induction hypothesis, the computed entries in $D[t_1,\SS_1]$ and $D[t_2,\SS_2]$ where $\SS_1 \cup \SS_2 = \SS$ are correct and store the feasible and undominated \hor solutions for the subgraph $G_{t_1}$ in $\SS_1$ and similarly, $\SS_2$ for $G_{t_2}$. Now we add ($w_1 + w_2 - w(S_1 \cup S_2),  p_1+p_2 -p(\SS_1 \cup \SS_2)$) to $D[t,\SS]$.

\end{enumerate}

What remains to be shown is that an undominated feasible solution $\VV^\pr$ of \hor in $\GG$ is contained in $D[r,\{v\}]$. Let $w$ be the weight of $\VV^\pr$ and $p$ be the value subject to \hor. Recall that $v \in \VV^\pr$. For each $t$, we consider the subgraph $\GG_t \cap \VV^\pr$. Since the DP state $D[t,\SS]$ is updated correctly for all subgraphs $\GG_t$, the bottom up dynamic programming approach ensures that all feasible and undominated pairs are correctly propagated. If $\VV^\pr$ is a valid solution, then the corresponding (weight, profit) must be stored in some DP state. Since $v \in \VV^\pr$, the optimal states where $v$ is selected will store the pair $(w,p)$. Therefore, $D[r,\{v\}]$ contains the pair $(w,p)$. 

\textbf{Running time}: There are $n$ choices for the fixed vertex $v$. Upon fixing $v$ and adding it to each bag of $(\mathbb{T}, \mathcal{X})$ we consider the total possible number of states. For every node $t$, we have $2^{|X_t|}$ choices of ${\SS}$. For each state, for each $w$, there can be at most one pair with $w$ as the first coordinate; similarly, for each $p$, there can be at most one pair with $p$ as the second coordinate. Thus, the number of undominated pairs in each $D[t,\SS]$ is at most ${\sf min}\{s,d\}$ time. Since the treewidth of the input graph \GG is at most \tw, it is possible to construct a data structure in time $\tw^{\OO(1)} \cdot n$ that allows performing adjacency queries in time $\OO(\tw)$.  For each node $t$, it takes time $\OO\left(2^{\tw} \cdot \tw^{\OO(1)} \cdot {\sf min}\{s^2,d^2\}\right)$ to compute all the values $D[t,\SS]$ and remove all undominated pairs. Since we can assume w.l.o.g that the number of nodes of the given tree decompositions is $\OO(\tw \cdot n)$, and there are $n$ choices for the vertex $v$, the running time of the algorithm is $\OO\left(2^{\tw}\cdot n^{\OO(1)} \cdot {\sf min}\{s^2,d^2\}\right)$.
\end{proof}}

\begin{theorem}\label{hand-di-fpt}
($\star$)  There is an FPT algorithm for directed \hand with running time $\OO^\star\left(2^{\tw}\cdot {\sf min}\{s^2,d^2\}\right)$, where $n$ is the number of vertices in the input graph, $\tw$ is the treewidth of the input graph, $s$ is the input size of the knapsack and $d$ is the input target value. 
\end{theorem} 

\longversion{\begin{proof}
We consider a nice tree decomposition of the underlying graph. Let $(\GG = (V,E),\{w_u\}_{u \in V_G}, \{p_u\}_{u \in V_G}, s,d)$ be an input instance of \hand such that $\tw=tw(\GG)$. Let $\VV^\pr$ be an optimal solution to \hand. For technical purposes, we guess a vertex $v \in \VV^\pr$ --- once the guess is fixed, we are only interested in finding solution subsets $\hat{\VV}$ that contain $v$ and $\VV^\pr$ is one such candidate. We also consider a nice edge tree decomposition $(\mathbb{T} = (V_{\mathbb{T}},E_{\mathbb{T}}),\mathcal{X})$ of $\GG$ that is rooted at a node $r$, and where $v$ has been added to all bags of the decomposition. Therefore, $X_{r} = \{v\}$ and each leaf bag is the singleton set $\{v\}$.

We define a function $\ell: V_{\mathbb{T}} \rightarrow \mathbb{N}$ as follows.
For a vertex $t \in V_\mathbb{T}$, $\ell(t) = \sf{dist}_\mathbb{T}(t,r)$, where $r$ is the root. Note that this implies that $\ell(r) = 0$. Let us assume that the values that $\ell$ take over the nodes of $\mathbb{T}$ are between $0$ and $L$. For a node $t\in V_\mathbb{T}$, we denote the set of vertices in the bags in the subtree rooted at $t$ by $V_t$ and $\GG_t=\GG[V_t]$. Now, we describe a dynamic programming algorithm over $(\mathbb{T},\mathcal{X})$. We have the following states.

    \textbf{States:} We maintain a DP table $D$ where a state has the following components:
    \begin{enumerate}
    \item $t$ represents a node in $V_\mathbb{T}$.
    \item $\SS$ represents a subset of the vertex subset $X_t$
    \end{enumerate}
    \textbf{Interpretation of States}
For each node \( t \in \mathbb{T} \), we maintain a list \( D[t, \SS] \) for every subset \( \SS \subseteq X_t \). Each entry in this list holds a set of feasible and undominated (weight, profit) pairs corresponding to valid \hand solutions $\hat{\SS}$ in the graph \( \GG_t \), where:  

\begin{itemize}  
    \item The selected vertices in \( \GG_t \) satisfy \(\hat{\SS} \cap X_t = \SS \).   
    \item If no such valid solution exists, we set \( D[t, \SS] = \emptyset \).  
\end{itemize}  

For each state $D[t,\SS]$, we initialize $D[t,\SS]$ to the list $\{(0,0)\}$.

    \textbf{Dynamic Programming on $D$}: We update the table $D$ as follows. We initialize the table of states with nodes $t\in V_\mathbb{T}$ such that $\ell(t)=L$. When all such states are initialized, then we move to update states where the node $t$ has $\ell(t) = L-1$, and so on, till we finally update states with $r$ as the node --- note that $\ell(r) =0$. For a particular $j$, $0\leq j< L$ and a state $D[t,\SS]$ such that $\ell(t) = j$, we can assume that $D[t',\SS']$ have been computed for all $t'$ such that $\ell(t')>j$ and all subsets $\SS'$ of $X_{t'}$. Now we consider several cases by which $D[t,\SS]$ is updated based on the nature of $t$ in $\mathbb{T}$:
    
    \begin{enumerate}
    \item \textbf{Leaf node.} Suppose $t$ is a leaf node with $X_{t} = \{v\}$. The only possible cases that can arise are as follows:
    \begin{enumerate}
        \item $D[t,\phi]$ stores the pair $\{(0, 0)\}$
        \item $D[t,\{v\}]$ = $\{(w_v, p_v)\}$, if $w_v \leq s$
    \end{enumerate} 
    
    \item \textbf{Introduce node.} Suppose $t$ is an introduce node. Then it has only one child $t'$ where $X_{t'} \subset X_{t}$ and there is exactly one vertex $u \neq v$ that belongs to $X_{t}$ but not $X_{t'}$. Then for all $\SS \subseteq X_t$:
    \begin{enumerate}
        \item If $u \not\in \SS$, 
        \begin{enumerate}
            \item if \(\exists w: w\in N^-(u), w \in \SS \) and \(N^+(w) \cap \VV_t = \{u\}\), then $D[t,\SS] = \emptyset$, hence assume otherwise.
            \item Otherwise, we copy all pairs of $D[t',\SS^\pr = \SS \setminus \{u\}]$ to $D[t',\SS,]$. 
        \end{enumerate}
        
        \item If $u \in \SS$,
        \begin{enumerate}            
            \item if $\VV_t \cap N^+_{\GG_t}(u) \subseteq \SS$, then for every pair $(w,p)$ in $D[t',\SS^\pr=\SS \setminus \{u\}]$ with $w + w_u \leq s$, we add $(w+w_u, p + p_u)$ and copy to $D[t,\SS]$.
            
            \item Otherwise, $D[t,\SS]$ = $\emptyset$

        \end{enumerate}
    \end{enumerate}
    
    \item \textbf{Forget node.} Suppose $t$ is a forget vertex node. Then it has only one child $t'$, and there is a vertex $u \neq v$ such that $X_t\cup\{u\} = X_{t'}$. Then for all $\SS \subseteq X_t$ we copy all feasible undominated pairs stored in $D[t',\SS\cup \{u\}]$ and $D[t',\SS]$ to $D[t,\SS]$. We remove any dominated pair and maintain only the undominated pairs in $D[t,\SS]$.

    \item \textbf{Join node.} Suppose $t$ is a join node. Then it has two children $t_1,t_2$ such that $X_t = X_{t_1} = X_{t_2}$. Then for all $\SS \subseteq X_t$, let $(w(\SS), p(\SS))$ be the total weight and value of the vertices in $\SS$. Consider a pair $(w_1,p_1)$ in $D[t_1,\SS_1]$ and a pair $(w_2,p_2)$ in $D[t_2,\SS_2]$, where $\SS_1 \cap \SS_2 = \SS$. Then for all pairs $(w_1,p_1) \in D[t_1,\SS_1]$ and $(w_2,p_2) \in D[t_2,\SS_2]$, if $w_1 + w_2 - w(\SS_1 \cap \SS_2) \leq s$, then we add $ \left( w_1 + w_2 - w(\SS_1 \cap \SS_2), p_1+p_2- p(\SS_1 \cap \SS_2) \right)$ and copy to $D[t,\SS]$.
    \end{enumerate}
    
Finally, in the last step of updating $D[t,\SS]$, we go through the list saved in $D[t,\SS]$ and only keep undominated pairs.

The output of the algorithm is a pair $(w,p)$ that is maximum of those stored in $D[r,\{v\}]$ such that $w \leq s$ and $p$ is the maximum value over all pairs in $D[r,\{v\}]$.

\textbf{Proof of correctness}: 
Recall that we are looking for a solution that is a set of vertices $\VV^\pr$ that contains the fixed vertex $v$ that belongs to all bags of the nice tree decomposition. In each state we maintain a list of feasible and undominated (weight, profit) pairs that correspond to the solution $\hat{\SS}$ of $\GG_t$ and  satisfy \( \hat{\SS} \cap X_t = \SS \).

We now show that the update rules holds for each $X_t$. To prove this formally, we need to consider the cases of what $t$ can be:

\begin{enumerate}
    \item \textbf{Leaf node.} Recall that in our modified nice tree decomposition we have added a vertex $v$ to all the bags. Suppose a leaf node $t$ contains a single vertex $v$, $D[t,\phi]$ stores the pair $\{(0, 0)\}$ and $D[t,\{v\}]$ = $\{(w_v, p_v)\}$ if $w_v \leq s$. This is true in particular when $j = L$, the base case. From now we can assume that for a node $t$ with $\ell(t) = j < L$ and all subsets $\SS \subseteq X_t$, $D[t',\SS^\pr]$ entries are correct and correspond to a \hand solution in $\GG_t$. when $\ell(t') > j$. 

    \item \textbf{Introduce node.} When $t$ is an introduce node, there is a child $t'$ such that $X_t = X_{t'} \cup \{u\}$. We are introducing a vertex $u$ and the edges associated with it in $\GG_t$. Let us prove for each case.
    \begin{enumerate}
        \item When $u$ is not included in $\SS$, the families of sets $\hat{\SS}$ considered in the definition of $D[t,\SS]$ and of $D[t',\SS]$ are equal. Notice that there may be a vertex $w \in X_{t}$ that belongs to $\SS$ such that $N^+(w) \cap V_{t'} = \emptyset$, i.e. $w$ did not have any out-neighbor in the subgraph $\GG_{t'}$. However, with the introduction of the vertex $u$ and the edges associated to it, if $N^+(w) \cap V_{t} = \{u\}$ i.e. $u$ is the only out-neighbor of $w$ in $\GG_t$, since $u$ does not belong to $\SS$, the \hand constraint is violated and hence the state $D[t,\SS]$ stores $\emptyset$ as it is no longer feasible. Otherwise we copy all pairs of $D[t',\SS^\pr = \SS]$ to $D[t,\SS]$.

        \item Now consider the case when $u$ is part of $\SS$. Let $\hat{\SS}$ be a feasible solution for \hand attained in the definition of $D[t,\SS]$. Then it follows that $\hat{\SS} \setminus \{u\}$ is one of the sets considered in the definition of $D[t',\SS \setminus \{u\}]$. It suffices to just check if $u$ has a neighbor in $X_t$ or not because the nice tree decomposition ensures that the newly introduced vertex $u$ can have neighbors only in $X_t$. There are two subcases:
        \begin{enumerate}
            \item \textbf{Subcase 1}. When $u$ has all out-neighbors $x$ selected in $\SS$ or $u$ does not have any out-neighbor in $\GG_{t}$, then for every pair $(w,p)$ in $D[t',\SS^\pr=\SS \setminus \{u\}]$ with $w + w_u \leq s$, we add $(w+w_u, p + p_u)$ and copy to $D[t,\SS]$.
            
            \item \textbf{Subcase 2}. When $u$ has a out-neighbor in $X_t \setminus \SS$, then $D[t,\SS] = \emptyset$.
        \end{enumerate}
    \end{enumerate}
    Since $\ell(t') > \ell(t)$, by induction hypothesis all entries in $D[t', \SS'= \SS]$ and $D[t', \SS'= \SS \setminus \{u\}]$,  $\forall$ $\SS' \subseteq X_{t'}$ are already computed. We update pairs in $D[t,\SS]$ depending on the cases discussed above.
    
    \item \textbf{Forget Node.}  When $t$ is a forget node, there is a child $t'$ such that $X_t = X_{t'} \setminus \{u\} $. Let $\hat{\SS}$ be a set for which the \hand solution is attained in the definition of $D[t,\SS]$. If $u \not\in \hat{\SS}$, then $\hat{\SS}$ is one of the sets considered in the definition of $D[t',\SS]$. And if $u \in \hat{\SS}$, then $\hat{\SS}$ is one of the sets considered in the definition of $D[t',\SS \cup \{u\}]$. Since $\ell(t') > \ell(t)$, by induction hypothesis all entries in $D[t',\SS' = \SS]$ and $D[t',\SS' = \SS \cup \{u\}]$, $\forall$ $\SS' \subseteq X_{t'}$ are already computed and feasible. We copy each undominated $(w,p)$ pair stored in $D[t',\SS' = \SS]$ and $D[t',\SS' = \SS \cup \{u\}]$ to $D[t,\SS]$.

    \item \textbf{Join node.}  When $t$ is a join node, there are two children $t_1$ and $t_2$ of $t$, such that $X_t = X_{t_1} = X_{t_2}$. Let $\hat{\SS}$ be a set for \hand attained in the definition of $D[t,\SS]$. Let $\hat{\SS}_1 = \hat{\SS} \cap V_{t_1}$ and $\hat{\SS}_2 = \hat{\SS} \cap V_{t_2}$. Observe that $\hat{\SS}_1$ is a solution to \hand in $\GG_{t_1}$ and $\hat{\SS}_1 \cap X_{t_1} = \SS_1$, so this is considered in the definition of $D[t_1, \SS]$ and similarly, $\hat{\SS}_2 \cap X_{t_2} = \SS_2$. From the definition of nice tree decomposition we know that there is no edge between the vertices of $V_{t_1} \setminus X_t$ and $V_{t_2} \setminus X_t$. Then we merge solutions from the two subgraphs and remove the over-counting. By the induction hypothesis, the computed entries in $D[t_1,\SS_1]$ and $D[t_2,\SS_2]$ where $\SS_1 \cap \SS_2 = \SS$ are correct and store the feasible and undominated \hand solutions for the subgraph $G_{t_1}$ in $\SS_1$ and similarly, $\SS_2$ for $G_{t_2}$. Now we add ($w_1 + w_2 - w(S_1 \cap S_2),  p_1+p_2 -p(\SS_1 \cap \SS_2)$) to $D[t,\SS]$.

\end{enumerate}

What remains to be shown is that an undominated feasible solution $\VV^\pr$ of \hand in $\GG$ is contained in $D[r,\{v\}]$. Let $w$ be the weight of $\VV^\pr$ and $p$ be the value subject to \hand. Recall that $v \in \VV^\pr$. For each $t$, we consider the subgraph $\GG_t \cap \VV^\pr$. Since the DP state $D[t,\SS]$ is updated correctly for all subgraphs $\GG_t$, the bottom up dynamic programming approach ensures that all feasible and undominated pairs are correctly propagated. If $\VV^\pr$ is a valid solution, then the corresponding (weight, profit) must be stored in some DP state. Since $v \in \VV^\pr$, the optimal states where $v$ is selected will store the pair $(w,p)$. Therefore, $D[r,\{v\}]$ contains the pair $(w,p)$. 

\textbf{Running time}: There are $n$ choices for the fixed vertex $v$. Upon fixing $v$ and adding it to each bag of $(\mathbb{T}, \mathcal{X})$ we consider the total possible number of states. For every node $t$, we have $2^{|X_t|}$ choices of ${\SS}$. For each state, for each $w$, there can be at most one pair with $w$ as the first coordinate; similarly, for each $p$, there can be at most one pair with $p$ as the second coordinate. Thus, the number of undominated pairs in each $D[t,\SS]$ is at most ${\sf min}\{s,d\}$ time. Since the treewidth of the input graph \GG is at most \tw, it is possible to construct a data structure in time $\tw^{\OO(1)} \cdot n$ that allows performing adjacency queries in time $\OO(\tw)$.  For each node $t$, it takes time $\OO\left(2^{\tw} \cdot \tw^{\OO(1)} \cdot {\sf min}\{s^2,d^2\}\right)$ to compute all the values $D[t,\SS]$ and remove all undominated pairs. Since we can assume w.l.o.g that the number of nodes of the given tree decompositions is $\OO(\tw \cdot n)$, and there are $n$ choices for the vertex $v$, the running time of the algorithm is $\OO\left(2^{\tw}\cdot n^{\OO(1)} \cdot {\sf min}\{s^2,d^2\}\right)$.
\end{proof}}

%% file: conclusion.tex
\longversion{\section{Conclusion}
Neighborhood knapsack problems have been extensively studied in the literature, particularly the hard variants, under the lens of classical complexity theory and approximation. However, their parameterized complexity has received limited attention. To address this gap, we design fixed-parameter tractable (\FPT) algorithms for the hard variants of the neighborhood knapsack problem parameterized by treewidth, and establish W-hardness results with respect to natural parameters such as weight, profit, and solution size. 
In addition, we introduce and study two novel variants—\sor and \sand—that relax the neighborhood constraints. For these soft variants, we explore both classical and parameterized complexity. We develop \FPT algorithms using dynamic programming over nice tree decompositions, maintaining tables of undominated (weight, profit) pairs. In soft variants, the dynamic program explicitly distinguishes between selected vertices and those contributing to profit. We further present a color-coding-based \FPT algorithm for \sor, parameterized by solution size~$b$, where we enumerate colorful tree structures and compute feasible (weight, profit) pairs via dynamic programming. Using splitters, we de-randomize the algorithm. Since the solution size~$b$ is at most twice the target profit~$d$, the algorithm is also \FPT when parameterized by~$d$. Our results show that even these relaxed variants remain computationally hard on several restricted graph classes. We prove strong \NP-completeness and W-hardness results, offering a comprehensive complexity landscape for both hard and soft neighborhood knapsack problems.}

%% file: appendix.tex
\section{Appendix}
\label{appendix}
\subsection{Conclusion}
Neighborhood knapsack problems have been extensively studied in the literature, particularly the hard variants, under the lens of classical complexity theory and approximation. However, their parameterized complexity has received limited attention. To address this gap, we design fixed-parameter tractable (\FPT) algorithms for the hard variants of the neighborhood knapsack problem parameterized by treewidth, and establish W-hardness results with respect to natural parameters such as weight, profit, and solution size. 
In addition, we introduce and study two novel variants—\sor and \sand—that relax the neighborhood constraints. For these soft variants, we explore both classical and parameterized complexity. We develop \FPT algorithms using dynamic programming over nice tree decompositions, maintaining tables of undominated (weight, profit) pairs. In soft variants, the dynamic program explicitly distinguishes between selected vertices and those contributing to profit. We further present a color-coding-based \FPT algorithm for \sor, parameterized by solution size~$b$, where we enumerate colorful tree structures and compute feasible (weight, profit) pairs via dynamic programming. Using splitters, we de-randomize the algorithm. Since the solution size~$b$ is at most twice the target profit~$d$, the algorithm is also \FPT when parameterized by~$d$. Our results show that even these relaxed variants remain computationally hard on several restricted graph classes. We prove strong \NP-completeness and W-hardness results, offering a comprehensive complexity landscape for both hard and soft neighborhood knapsack problems.

\subsection{Technical Preliminaries}

\textbf{Parameterized Complexity.} In decision problems where the input has size \( n \) and is associated with a parameter \( k \), the objective in parameterized complexity is to design algorithms with running time \( f(k) \cdot n^{O(1)} \), where \( f \) is a computable function that depends only on \( k \) \cite{downey2012parameterized}. Problems that admit such algorithms are said to be \emph{fixed-parameter tractable} (\FPT). An algorithm with running time \( f(k) \cdot n^{O(1)} \) is called an \FPT algorithm, and the corresponding running time is referred to as an \FPT running time \cite{hartmanis2006texts}. Formally, a parameterized problem \( L \subseteq \Sigma^* \times \mathbb{N} \) is said to be \emph{fixed-parameter tractable} if there exists an algorithm \( A \), a computable function \( f : \mathbb{N} \rightarrow \mathbb{N} \), and a constant \( c \in \mathbb{N} \) such that for every input \( (x, k) \in \Sigma^* \times \mathbb{N} \), the algorithm \( A \) decides whether \( (x, k) \in L \) in time at most \( f(k) \cdot |(x, k)|^c \). The class of all such problems is denoted by \FPT \cite{cygan2015parameterized}. 

The treewidth of a graph measures how closely the graph resembles a tree \cite{cygan2015parameterized}. Informally, a tree decomposition of a graph is a tree where each node corresponds to a subset of vertices, called bags, and must satisfy three conditions: (i) every vertex of the graph appears in at least one bag, (ii) both endpoints of every edge are contained in some bag, and (iii) for any vertex, the nodes of the tree containing it must form a connected subtree. For formal definitions of a tree decomposition, a nice tree decomposition, and treewidth, we refer to \cite{cygan2015parameterized}.

\subsection{Missing Proofs}

\begin{proof}[Proof of \Cref{self-loop-un-sor}]
Let $(\GG = (\VV,\EE), \{w_i\}_{i \in \VV}, \{p_i\}_{i \in \VV}, s , d )$ be an instance of the undirected \sor problem where \( G \) contains a self-loop on some vertex \( v \). Construct a new graph \( G' \) by removing the self-loop \( \{v, v\} \), adding a new vertex \( v^\pr \) with \( w_{v^\pr} = p_{v^\pr} = 0 \), and inserting the edge \( \{v, v^\pr\} \) in place of the self-loop. The rest of the instance, including the knapsack budget \( s \) and demand \( d \), remains unchanged.

We prove that a subset \( S \subseteq V \) is a feasible solution in the original instance if and only if the corresponding subset \( S' \subseteq V' \) is feasible in the new instance, where:
\begin{itemize}
    \item If \( v \in S \), then \( S' = S \cup \{v^\pr\} \),
    \item If \( v \notin S \), then \( S' = S \).
\end{itemize}

Suppose \( S \subseteq V \) is a feasible solution in the original instance. If \( v \notin S \), then \( v^\pr \notin S' \), and all other vertices behave identically. So \( S' = S \) is feasible in the modified instance. If \( v \in S \), then in the original graph, \( v \) has a self-loop, so it contributes to the profit. In the modified instance, we include \( v^\pr \) as well (with zero weight and profit), ensuring that \( v \) has a neighbor and hence still contributes profit. The total weight and profit remain unchanged. Therefore, \( S' = S \cup \{v^\pr\} \) is also feasible.

Suppose \( S' \subseteq V' \) is a feasible solution in the modified instance. If \( v \notin S' \), then \( v^\pr \) must also be absent from \( S' \), as it has no other neighbors. Thus, \( S = S' \subseteq V \) is a valid solution in the original instance. If \( v \in S' \), then: If \( v^\pr \in S' \), then \( v \) has a neighbor and contributes profit in the modified instance. In the original instance, the self-loop ensures the same behavior. If \( v^\pr \notin S' \), then \( v \) must have another neighbor in \( S' \), so again the behavior is consistent with the original instance. Thus, in both cases, \( S = S' \setminus \{v^\pr\} \subseteq V \) is feasible in the original instance. Hence, the transformation preserves feasibility and objective values. Therefore, self-loops in the undirected \sor problem can be safely removed via this gadget. 
\end{proof}

\begin{proof}[Proof of \Cref{self-loop-di-sor}]
Let \( \GG = (\VV, \EE) \) be a directed graph where a vertex \( v \in \VV \) has a self-loop \( (v, v) \in \EE \). Let \( \GG' = (\VV, \EE') \) be the modified instance obtained by removing all outgoing edges from \( v \), including the self-loop. We prove that any solution \( \SS \subseteq \VV \) is feasible in \( \GG \) if and only if it is feasible in \( \GG' \).

Suppose \( \SS \) is a feasible solution in \( \GG \). We consider the role of vertex \( v \) in \( \SS \). If \( v \notin \SS \). Then the removal of outgoing edges from \( v \) has no effect on the solution. Since \( v \) is not selected, its outgoing edges do not influence the profit of any vertex in \( \SS \), and the feasibility remains unchanged. If \( v \in \SS \). In \( \GG \), the self-loop \( (v, v) \in \EE \) ensures that \( v \) has an out-neighbor (itself), and hence it contributes to the total profit. In the modified graph \( \GG' \), all outgoing edges from \( v \) are removed, so \( v \) has no out-neighbors. By the definition of \sor, a vertex with no out-neighbors contributes its profit when selected. Hence, \( v \) still contributes to profit in \( \GG' \). Thus, the status of \( v \) and the feasibility of the solution is preserved.

Suppose \( \SS \subseteq \VV \) is a feasible solution in \( \GG' \). We show it is also feasible in \( \GG \). Consider the vertex \( v \in \VV \) that had its outgoing edges (including the self-loop) removed. Let us verify for each case: If \( v \notin \SS \), then nothing changes. If \( v \in \SS \), then in \( \GG' \), \( v \) has no out-neighbors and therefore contributes its profit directly. In the original graph \( \GG \), \( v \) has a self-loop, so \( v \) is its own out-neighbor and hence satisfies the condition to contribute to profit as well. Therefore, the contribution of \( v \) remains valid, and the feasibility is preserved.

Thus, in both directions, the feasibility of the solution remains unchanged after removing all outgoing edges from \( v \), including its self-loop. This completes the proof.
\end{proof}

\begin{proof}[Proof of \Cref{self-loop-sand}]

Let the original instance be \( \GG = (\VV, \EE) \), and let \( \GG' = (\VV, \EE') \) be the modified instance obtained by removing the self-loop \( (v, v) \in \EE \), i.e., \( \EE' = \EE \setminus \{(v, v)\} \). We prove the claim in both directions:

Suppose \( \SS \subseteq \VV \) is a feasible solution in \GG. We show that it remains feasible in \GG'. If \( v \notin \SS \), the profit status of \( v \) is irrelevant. Since the structure of \GG' remains unchanged for all other. If \( v \in \SS \) but \( v \) does not contribute to profit in \GG, then by the \sand condition, there exists some out-neighbor \( w \neq v \) of \( v \) such that \( w \notin \SS \). Since this neighbor \( w \) is unaffected by the removal of the self-loop, the profit condition for \( v \) remains violated in \GG' as well. Hence, the profit contribution of \( v \) remains the same, and \SS is still feasible. If \( v \in \SS \) and \( v \) contributes its profit in \GG, then all out-neighbors of \( v \), including possibly itself, are in \SS. But since all neighbors \( w \neq v \) of \( v \) are also in \SS, removing the self-loop does not violate the \sand condition. Thus, \( v \) remains profitable in \GG', and the solution \SS remains feasible. 

Suppose \( \SS \subseteq \VV \) is a feasible solution in \GG'. We show that it remains feasible in \GG. The only difference between \GG and \GG' is the presence of the self-loop \( (v,v) \), which only affects the profit condition of \( v \). Again, we consider the cases: If \( v \notin \SS \), its status is unaffected by the self-loop. Hence, \SS remains feasible in \GG. If \( v \in \SS \) but \( v \) does not contribute to profit in \GG', then there exists some neighbor \( w \neq v \) such that \( w \notin \SS \), and this remains true in \GG. Thus, \( v \) still does not contribute to profit in \GG. If \( v \in \SS \) and \( v \) is profitable in \GG', then all its out-neighbors \( w \in N^+(v) \setminus \{v\} \) must be in \SS. Adding back the self-loop in \GG only adds a redundant neighbor that is already in \SS, so the profit condition for \( v \) continues to hold in \GG.

In all cases, the feasibility of \SS is preserved under the removal or addition of the self-loop. Therefore, self-loops can be safely ignored when solving \sand.
\end{proof}

\shortversion{

\begin{proof}[Proof of \Cref{sor-tree-npc}]
    \sor for trees $\in$ \NP. To prove hardness, we reduce from \kp problem. 

    We denote an instance of \kp as $(\II, \{s_i\}_{i \in \II}, \{p_i\}_{i \in \II}, c, \alpha)$. Given, an instance of \kp problem, we construct an instance of \sor as follows: let the graph is $\GG = (\VV, \EE)$ with $\VV = \II \cup \{a_{n+1}\}$ where vertex $v_i$ corresponds to item $a_i$ has weight $w_i$ and profit $p_i$, and the vertex $v_{n+1}$ has weight 0 and profit 0. Define the edge set \( \EE = \{ \{v_i, v_{n+1}\} | 1 \leq i \leq n\}\).

Since the central vertex has weight and profit 0, it can always be taken in any feasible solution. Now, the \kp instance $(\II, \{s_i\}_{i \in \II}, \{p_i\}_{i \in \II}, c, \alpha)$ is an \yes instance if and only if the \sor instance $(\GG = (\VV,\EE), \{w_i\}_{i \in \VV}, \{p_i\}_{i \in \VV}, s = c, d = \alpha)$ is an \yes instance. 
\end{proof}

\begin{proof}[Proof of \Cref{sand-npc}]
\sand $\in$ \NP, because given a certificate i.e. a set of vertices  $\VV^\pr \subseteq \VV$, it can be verified in polynomial time if  \(\sum_{v \in \VV^\prime} w_v \leq s 
\quad \text{and} \quad 
\sum_{\substack{v \in \VV^\prime: \\ N(v) \subseteq \VV^\prime}} p_v \geq d.\)

To show hardness, we reduce the \clv problem to \sand. 

Given an instance of \clv as $(\GG = (\VV, \EE),k,l)$, we construct an instance of \sand as follows: we consider the same graph \GG, for each vertex $u \in \VV$, set $w_u = 1$, $p_u = 1$, knapsack capacity $s = l+k$ and profit $d=l$. We now claim that \clv instance $(\GG = (\VV, \EE),k,l)$ is an \yes instance if and only if \sand instance $(\GG = (\VV,\EE), \{w_i\}_{i \in \VV} = 1, \{p_i\}_{i \in \VV} = 1, s = l+k, d = l)$ is an \yes instance. 

Suppose \clv is an \yes instance then there exists a partition $\VV = \XX \cup \SS \cup \YY$ in $\GG$ such that $|\XX| = l$, $|\SS| \leq k$, and there is no edge between $\XX$ and $\YY$. Consider the set $\VV^\pr = \XX \cup \SS$ in $\GG$. Note that
\(\sum_{v \in \VV^\pr} w_v = |\VV^\pr| = |\XX| + |\SS| \leq l + k = s.\)For each $v \in \XX$, since there is no edge from $\XX$ to $\YY$, all neighbors of $v$ lie in $\XX \cup \SS = \VV^\pr$. Thus, $N(v) \subseteq \VV^\pr$, so $v$ contributes profit. Hence,
\(
\sum_{\substack{v \in \VV^\pr: \\ N(v) \subseteq \VV^\pr}} p_v \geq |\XX| = l = d.
\) Therefore, $\VV^\pr$ is a valid solution for \sand.

Suppose \sand is a YES-instance, then \clv is a YES-instance: Suppose $\VV^\pr \subseteq \VV$ is a solution for \sand such that
\(
\sum_{v \in \VV^\pr} w_v \leq s = k + l \quad \text{and} \quad \sum_{\substack{v \in \VV^\pr: \\ N(v) \subseteq \VV^\pr}} p_v \geq d = l.
\) Let $\XX \subseteq \VV^\pr$ denote the set of vertices in $\VV^\pr$ such that all their neighbors are also in $\VV^\pr$, i.e., 
\(
\XX = \{ v \in \VV^\pr \mid N(v) \subseteq \VV^\pr \}.
\)
Then, the total profit contributed is exactly $|\XX| \geq l$. Since $\sum_{v \in \VV^\pr} w_v = |\VV^\pr| \leq k + l$, define $\SS = \VV^\pr \setminus \XX$, and note that
\(
|\SS| = |\VV^\pr| - |\XX| \leq (k + l) - l = k.
\)
Now define $\YY = \VV \setminus \VV^\pr$. Since every $v \in \XX$ has all its neighbors in $\VV^\pr$, there is no edge between $\XX$ and $\YY$. Hence, we have a partition $\VV = \XX \cup \SS \cup \YY$ with $|\XX| = l$, $|\SS| \leq k$, and no edge between $\XX$ and $\YY$. Thus, \clv is a YES-instance.
\end{proof}
\begin{proof}[Proof of \Cref{hand-woh}]
Given an instance of \clique as $(\GG^\pr = (\VV^\pr, \EE^\pr), k)$, we construct an instance of \hand as follows: We construct a bipartite graph $\GG(\VV = \AA \cup \BB, \EE)$ from a given graph $\GG^\pr = (\VV^\pr, \EE^\pr)$ as follows. For each edge $e \in \EE^\pr$, create a vertex $u_e$ in $\AA$, and for each vertex $v \in \VV^\pr$, create a vertex $v$ in $\BB$ i.e, $\AA$ corresponds to the edge set of $\GG^\pr$ and $\BB$ corresponds to the vertex set of $\GG^\pr$. For every edge $e = \{v_i, v_j\} \in \EE^\pr$, we add edges $(u_e, v_i)$ and $(u_e, v_j)$ in $\GG$, where $u_e \in \AA$ and $v_i, v_j \in \BB$, representing the incidence of edge $e$ with vertices $v_i$ and $v_j$ in the original graph $\GG^\pr$. We assign the weights and profits to the vertices of $\GG$ as follows. Each vertex $u \in \AA$ (corresponding to an edge in $\GG^\pr$) is assigned weight \( w_u = 0 \) and profit \( p_u = 1 \). Each vertex \( v \in \BB \) (corresponding to a vertex in $\GG^\pr$) is assigned weight \( w_v = 1 \) and profit \( p_v = 0 \). Set $s = k$ and $d = \binom{k}{2}$. 
We claim that the \clique problem is a \yes instance if and only if the corresponding \hand instance is also a \yes instance.

First, suppose that the \clique problem is a \yes instance. This means there exists a clique of size \(k\) in $\GG^\pr$. We can construct a solution to the \hand instance by selecting the \(k\) nodes from \BB that correspond to the vertices of the \(k\)-\clique. Additionally, we include \(\binom{k}{2}\) nodes from \AA that correspond to the edges in the \(k\)-clique. In this configuration, the \(k\) nodes from set \BB contribute a total weight of \(k\) and no profit, while the edges selected from \AA contribute a profit of \(\binom{k}{2}\). This configuration satisfies the constraints of the knapsack problem: the total weight \(\sum_{\substack{v \in \SS: \\ N(v) \subseteq \SS}} w_u \leq k\) and the total profit \(\sum_{\substack{v \in \SS: \\ N(v) \subseteq \SS}} p_v = \binom{k}{2}\).

Conversely, if the \sand instance is a \yes instance, we can construct a corresponding \(k\)-\clique in $\GG^\pr$. Since we must select at least \(\binom{k}{2}\) edges from \AA to meet the profit requirement, this implies that the chosen edges correspond to a complete subgraph (clique) among the selected vertices in \BB. Therefore, \clique instance $(\GG^\pr = (\VV^\pr, \EE^\pr), k)$ is an \yes instance if and only if  \sand instance $(\GG = (\VV = (\AA \cup \BB),\EE), \{w_i\}_{i \in \VV}, \{p_i\}_{i \in \VV}, s = k, d = \binom{k}{2})$ is an \yes instance.
\end{proof}

\begin{proof}[Proof of \Cref{hor-di-fpt}]
We consider a nice tree decomposition of the underlying graph. Let $(\GG = (V,E),\{w_u\}_{u \in V_G}, \{p_u\}_{u \in V_G}, s,d)$ be an input instance of \hor such that $\tw=tw(\GG)$. Let $\VV^\pr$ be an optimal solution to \hor. For technical purposes, we guess a vertex $v \in \VV^\pr$ --- once the guess is fixed, we are only interested in finding solution subsets $\hat{\VV}$ that contain $v$ and $\VV^\pr$ is one such candidate. We also consider a nice edge tree decomposition $(\mathbb{T} = (V_{\mathbb{T}},E_{\mathbb{T}}),\mathcal{X})$ of $\GG$ that is rooted at a node $r$, and where $v$ has been added to all bags of the decomposition. Therefore, $X_{r} = \{v\}$ and each leaf bag is the singleton set $\{v\}$.

We define a function $\ell: V_{\mathbb{T}} \rightarrow \mathbb{N}$ as follows.
For a vertex $t \in V_\mathbb{T}$, $\ell(t) = \sf{dist}_\mathbb{T}(t,r)$, where $r$ is the root. Note that this implies that $\ell(r) = 0$. Let us assume that the values that $\ell$ take over the nodes of $\mathbb{T}$ are between $0$ and $L$. For a node $t\in V_\mathbb{T}$, we denote the set of vertices in the bags in the subtree rooted at $t$ by $V_t$ and $\GG_t=\GG[V_t]$. Now, we describe a dynamic programming algorithm over $(\mathbb{T},\mathcal{X})$. We have the following states.

    \textbf{States:} We maintain a DP table $D$ where a state has the following components:
    \begin{enumerate}
    \item $t$ represents a node in $V_\mathbb{T}$.
    \item $\SS$ represents a subset of the vertex subset $X_t$
    \end{enumerate}
    \textbf{Interpretation of States}
For each node \( t \in \mathbb{T} \), we maintain a list \( D[t, \SS] \) for every subset \( \SS \subseteq X_t \). Each entry in this list holds a set of feasible and undominated (weight, profit) pairs corresponding to valid \hor solutions $\hat{\SS}$ in the graph \( \GG_t \), where:  

\begin{itemize}  
    \item The selected vertices in \( \GG_t \) satisfy \(\hat{\SS} \cap X_t = \SS \).   
    \item If no such valid solution exists, we set \( D[t, \SS] = \emptyset \).  
\end{itemize}  

For each state $D[t,\SS]$, we initialize $D[t,\SS]$ to the list $\{(0,0)\}$.

    \textbf{Dynamic Programming on $D$}: We update the table $D$ as follows. We initialize the table of states with nodes $t\in V_\mathbb{T}$ such that $\ell(t)=L$. When all such states are initialized, then we move to update states where the node $t$ has $\ell(t) = L-1$, and so on, till we finally update states with $r$ as the node --- note that $\ell(r) =0$. For a particular $j$, $0\leq j< L$ and a state $D[t,\SS]$ such that $\ell(t) = j$, we can assume that $D[t',\SS']$ have been computed for all $t'$ such that $\ell(t')>j$ and all subsets $\SS'$ of $X_{t'}$. Now we consider several cases by which $D[t,\SS]$ is updated based on the nature of $t$ in $\mathbb{T}$:
    
    \begin{enumerate}
    \item \textbf{Leaf node.} Suppose $t$ is a leaf node with $X_{t} = \{v\}$. The only possible cases that can arise are as follows:
    \begin{enumerate}
        \item $D[t,\phi]$ stores the pair $\{(0, 0)\}$
        \item $D[t,\{v\}]$ = $\{(w_v, p_v)\}$, if $w_v \leq s$
    \end{enumerate} 
    
    \item \textbf{Introduce node.} Suppose $t$ is an introduce node. Then it has only one child $t'$ where $X_{t'} \subset X_{t}$ and there is exactly one vertex $u \neq v$ that belongs to $X_{t}$ but not $X_{t'}$. Then for all $\SS \subseteq X_t$:
    \begin{enumerate}
        \item If $u \not\in \SS$, 
        \begin{enumerate}
            \item if \(\exists w: w\in N^-(u), w \in \SS \) and \(N^+(w) \cap \VV_t = \{u\}\), then $D[t,\SS] = \emptyset$, hence assume otherwise.
            \item Otherwise, we copy all pairs of $D[t',\SS^\pr = \SS \setminus \{u\}]$ to $D[t',\SS,]$. 
        \end{enumerate}
        
        \item If $u \in \SS$,
        \begin{enumerate}            
            \item if $u$ has out-neighbor $x$ in $\SS$ or \(N^+(u) \cap \VV_t = \emptyset\), then for every pair $(w,p)$ in $D[t',\SS^\pr=\SS \setminus \{u\}]$ with $w + w_u \leq s$, we add $(w+w_u, p + p_u)$ and copy to $D[t,\SS]$.
            
            \item Otherwise, $D[t,\SS]$ = $\emptyset$

        \end{enumerate}
    \end{enumerate}
    
    \item \textbf{Forget node.} Suppose $t$ is a forget vertex node. Then it has only one child $t'$, and there is a vertex $u \neq v$ such that $X_t\cup\{u\} = X_{t'}$. Then for all $\SS \subseteq X_t$ we copy all feasible undominated pairs stored in $D[t',\SS\cup \{u\}]$ and $D[t',\SS]$ to $D[t,\SS]$. We remove any dominated pair and maintain only the undominated pairs in $D[t,\SS]$.

    \item \textbf{Join node.} Suppose $t$ is a join node. Then it has two children $t_1,t_2$ such that $X_t = X_{t_1} = X_{t_2}$. Then for all $\SS \subseteq X_t$, let $(w(\SS), p(\SS))$ be the total weight and value of the vertices in $\SS$. Consider a pair $(w_1,p_1)$ in $D[t_1,\SS_1]$ and a pair $(w_2,p_2)$ in $D[t_2,\SS_2]$, where $\SS_1 \cup \SS_2 = \SS$. Then for all pairs $(w_1,p_1) \in D[t_1,\SS_1]$ and $(w_2,p_2) \in D[t_2,\SS_2]$, if $w_1 + w_2 - w(\SS_1 \cup \SS_2) \leq s$, then we add $ \left( w_1 + w_2 - w(\SS_1 \cup \SS_2), p_1+p_2- p(\SS_1 \cup \SS_2) \right)$ and copy to $D[t,\SS]$.
    \end{enumerate}
    
Finally, in the last step of updating $D[t,\SS]$, we go through the list saved in $D[t,\SS]$ and only keep undominated pairs.

The output of the algorithm is a pair $(w,p)$ that is maximum of those stored in $D[r,\{v\}]$ such that $w \leq s$ and $p$ is the maximum value over all pairs in $D[r,\{v\}]$.

\textbf{Proof of correctness}: 
Recall that we are looking for a solution that is a set of vertices $\VV^\pr$ that contains the fixed vertex $v$ that belongs to all bags of the nice tree decomposition. In each state we maintain a list of feasible and undominated (weight, profit) pairs that correspond to the solution $\hat{\SS}$ of $\GG_t$ and  satisfy \( \hat{\SS} \cap X_t = \SS \).

We now show that the update rules holds for each $X_t$. To prove this formally, we need to consider the cases of what $t$ can be:

\begin{enumerate}
    \item \textbf{Leaf node.} Recall that in our modified nice tree decomposition we have added a vertex $v$ to all the bags. Suppose a leaf node $t$ contains a single vertex $v$, $D[t,\phi]$ stores the pair $\{(0, 0)\}$ and $D[t,\{v\}]$ = $\{(w_v, p_v)\}$ if $w_v \leq s$. This is true in particular when $j = L$, the base case. From now we can assume that for a node $t$ with $\ell(t) = j < L$ and all subsets $\SS \subseteq X_t$, $D[t',\SS^\pr]$ entries are correct and correspond to a \hor solution in $\GG_t$. when $\ell(t') > j$. 

    \item \textbf{Introduce node.} When $t$ is an introduce node, there is a child $t'$ such that $X_t = X_{t'} \cup \{u\}$. We are introducing a vertex $u$ and the edges associated with it in $\GG_t$. Let us prove for each case.
    \begin{enumerate}
        \item When $u$ is not included in $\SS$, the families of sets $\hat{\SS}$ considered in the definition of $D[t,\SS]$ and of $D[t',\SS]$ are equal. Notice that there may be a vertex $w \in X_{t}$ that belongs to $\SS$ such that $N^+(w) \cap V_{t'} = \emptyset$, i.e. $w$ did not have any out-neighbor in the subgraph $\GG_{t'}$. However, with the introduction of the vertex $u$ and the edges associated to it, if $N^+(w) \cap V_{t} = \{u\}$ i.e. $u$ is the only out-neighbor of $w$ in $\GG_t$, since $u$ does not belong to $\SS$, the \hor constraint is violated and hence the state $D[t,\SS]$ stores $\emptyset$ as it is no longer feasible. Otherwise we copy all pairs of $D[t',\SS^\pr = \SS]$ to $D[t,\SS]$.

        \item Now consider the case when $u$ is part of $\SS$. Let $\hat{\SS}$ be a feasible solution for \hor attained in the definition of $D[t,\SS]$. Then it follows that $\hat{\SS} \setminus \{u\}$ is one of the sets considered in the definition of $D[t',\SS \setminus \{u\}]$. It suffices to just check if $u$ has a neighbor in $X_t$ or not because the nice tree decomposition ensures that the newly introduced vertex $u$ can have neighbors only in $X_t$. There are two subcases:
        \begin{enumerate}
            \item \textbf{Subcase 1}. When $u$ has at least one out-neighbor $x$ selected in $\SS$ or $u$ does not have any out-neighbor in $\GG_{t}$, then for every pair $(w,p)$ in $D[t',\SS^\pr=\SS \setminus \{u\}]$ with $w + w_u \leq s$, we add $(w+w_u, p + p_u)$ and copy to $D[t,\SS]$.
            
            \item \textbf{Subcase 2}. When $u$ has a out-neighbor in $\GG_t$ but none of them are selected in $\SS$, then $D[t,\SS] = \emptyset$.
        \end{enumerate}
    \end{enumerate}
    Since $\ell(t') > \ell(t)$, by induction hypothesis all entries in $D[t', \SS'= \SS]$ and $D[t', \SS'= \SS \setminus \{u\}]$,  $\forall$ $\SS' \subseteq X_{t'}$ are already computed. We update pairs in $D[t,\SS]$ depending on the cases discussed above.
    
    \item \textbf{Forget Node.}  When $t$ is a forget node, there is a child $t'$ such that $X_t = X_{t'} \setminus \{u\} $. Let $\hat{\SS}$ be a set for which the \hor solution is attained in the definition of $D[t,\SS]$. If $u \not\in \hat{\SS}$, then $\hat{\SS}$ is one of the sets considered in the definition of $D[t',\SS]$. And if $u \in \hat{\SS}$, then $\hat{\SS}$ is one of the sets considered in the definition of $D[t',\SS \cup \{u\}]$. Since $\ell(t') > \ell(t)$, by induction hypothesis all entries in $D[t',\SS' = \SS]$ and $D[t',\SS' = \SS \cup \{u\}]$, $\forall$ $\SS' \subseteq X_{t'}$ are already computed and feasible. We copy each undominated $(w,p)$ pair stored in $D[t',\SS' = \SS]$ and $D[t',\SS' = \SS \cup \{u\}]$ to $D[t,\SS]$.

    \item \textbf{Join node.}  When $t$ is a join node, there are two children $t_1$ and $t_2$ of $t$, such that $X_t = X_{t_1} = X_{t_2}$. Let $\hat{\SS}$ be a set for \hor attained in the definition of $D[t,\SS]$. Let $\hat{\SS}_1 = \hat{\SS} \cap V_{t_1}$ and $\hat{\SS}_2 = \hat{\SS} \cap V_{t_2}$. Observe that $\hat{\SS}_1$ is a solution to \hor in $\GG_{t_1}$ and $\hat{\SS}_1 \cap X_{t_1} = \SS_1$, so this is considered in the definition of $D[t_1, \SS]$ and similarly, $\hat{\SS}_2 \cap X_{t_2} = \SS_2$. From the definition of nice tree decomposition we know that there is no edge between the vertices of $V_{t_1} \setminus X_t$ and $V_{t_2} \setminus X_t$. Then we merge solutions from the two subgraphs and remove the over-counting. By the induction hypothesis, the computed entries in $D[t_1,\SS_1]$ and $D[t_2,\SS_2]$ where $\SS_1 \cup \SS_2 = \SS$ are correct and store the feasible and undominated \hor solutions for the subgraph $G_{t_1}$ in $\SS_1$ and similarly, $\SS_2$ for $G_{t_2}$. Now we add ($w_1 + w_2 - w(S_1 \cup S_2),  p_1+p_2 -p(\SS_1 \cup \SS_2)$) to $D[t,\SS]$.

\end{enumerate}

What remains to be shown is that an undominated feasible solution $\VV^\pr$ of \hor in $\GG$ is contained in $D[r,\{v\}]$. Let $w$ be the weight of $\VV^\pr$ and $p$ be the value subject to \hor. Recall that $v \in \VV^\pr$. For each $t$, we consider the subgraph $\GG_t \cap \VV^\pr$. Since the DP state $D[t,\SS]$ is updated correctly for all subgraphs $\GG_t$, the bottom up dynamic programming approach ensures that all feasible and undominated pairs are correctly propagated. If $\VV^\pr$ is a valid solution, then the corresponding (weight, profit) must be stored in some DP state. Since $v \in \VV^\pr$, the optimal states where $v$ is selected will store the pair $(w,p)$. Therefore, $D[r,\{v\}]$ contains the pair $(w,p)$. 

\textbf{Running time}: There are $n$ choices for the fixed vertex $v$. Upon fixing $v$ and adding it to each bag of $(\mathbb{T}, \mathcal{X})$ we consider the total possible number of states. For every node $t$, we have $2^{|X_t|}$ choices of ${\SS}$. For each state, for each $w$, there can be at most one pair with $w$ as the first coordinate; similarly, for each $p$, there can be at most one pair with $p$ as the second coordinate. Thus, the number of undominated pairs in each $D[t,\SS]$ is at most ${\sf min}\{s,d\}$ time. Since the treewidth of the input graph \GG is at most \tw, it is possible to construct a data structure in time $\tw^{\OO(1)} \cdot n$ that allows performing adjacency queries in time $\OO(\tw)$.  For each node $t$, it takes time $\OO\left(2^{\tw} \cdot \tw^{\OO(1)} \cdot {\sf min}\{s^2,d^2\}\right)$ to compute all the values $D[t,\SS]$ and remove all undominated pairs. Since we can assume w.l.o.g that the number of nodes of the given tree decompositions is $\OO(\tw \cdot n)$, and there are $n$ choices for the vertex $v$, the running time of the algorithm is $\OO\left(2^{\tw}\cdot n^{\OO(1)} \cdot {\sf min}\{s^2,d^2\}\right)$.
\end{proof}

\begin{proof}[Proof of \Cref{hand-di-fpt}]
	We consider a nice tree decomposition of the underlying graph. Let $(\GG = (V,E),\{w_u\}_{u \in V_G}, \{p_u\}_{u \in V_G}, s,d)$ be an input instance of \hand such that $\tw=tw(\GG)$. Let $\VV^\pr$ be an optimal solution to \hand. For technical purposes, we guess a vertex $v \in \VV^\pr$ --- once the guess is fixed, we are only interested in finding solution subsets $\hat{\VV}$ that contain $v$ and $\VV^\pr$ is one such candidate. We also consider a nice edge tree decomposition $(\mathbb{T} = (V_{\mathbb{T}},E_{\mathbb{T}}),\mathcal{X})$ of $\GG$ that is rooted at a node $r$, and where $v$ has been added to all bags of the decomposition. Therefore, $X_{r} = \{v\}$ and each leaf bag is the singleton set $\{v\}$.
	
	We define a function $\ell: V_{\mathbb{T}} \rightarrow \mathbb{N}$ as follows.
	For a vertex $t \in V_\mathbb{T}$, $\ell(t) = \sf{dist}_\mathbb{T}(t,r)$, where $r$ is the root. Note that this implies that $\ell(r) = 0$. Let us assume that the values that $\ell$ take over the nodes of $\mathbb{T}$ are between $0$ and $L$. For a node $t\in V_\mathbb{T}$, we denote the set of vertices in the bags in the subtree rooted at $t$ by $V_t$ and $\GG_t=\GG[V_t]$. Now, we describe a dynamic programming algorithm over $(\mathbb{T},\mathcal{X})$. We have the following states.

	\textbf{States:} We maintain a DP table $D$ where a state has the following components:
	\begin{enumerate}
		\item $t$ represents a node in $V_\mathbb{T}$.
		\item $\SS$ represents a subset of the vertex subset $X_t$
	\end{enumerate}
	\textbf{Interpretation of States}
	For each node \( t \in \mathbb{T} \), we maintain a list \( D[t, \SS] \) for every subset \( \SS \subseteq X_t \). Each entry in this list holds a set of feasible and undominated (weight, profit) pairs corresponding to valid \hand $\hat{\SS}$ in the graph \( \GG_t \), where:  
	
	\begin{itemize}  
		\item The selected vertices in \( \GG_t \) satisfy \(\hat{\SS} \cap X_t = \SS \).   
		\item If no such valid solution exists, we set \( D[t, \SS] = \emptyset \).  
	\end{itemize}  
	
	For each state $D[t,\SS]$, we initialize $D[t,\SS]$ to the list $\{(0,0)\}$.
	
	\textbf{Dynamic Programming on $D$}: We update the table $D$ as follows. We initialize the table of states with nodes $t\in V_\mathbb{T}$ such that $\ell(t)=L$. When all such states are initialized, then we move to update states where the node $t$ has $\ell(t) = L-1$, and so on, till we finally update states with $r$ as the node --- note that $\ell(r) =0$. For a particular $j$, $0\leq j< L$ and a state $D[t,\SS]$ such that $\ell(t) = j$, we can assume that $D[t',\SS']$ have been computed for all $t'$ such that $\ell(t')>j$ and all subsets $\SS'$ of $X_{t'}$. Now we consider several cases by which $D[t,\SS]$ is updated based on the nature of $t$ in $\mathbb{T}$:
	
	\begin{enumerate}
		\item \textbf{Leaf node.} Suppose $t$ is a leaf node with $X_{t} = \{v\}$. The only possible cases that can arise are as follows:
		\begin{enumerate}
			\item $D[t,\phi]$ stores the pair $\{(0, 0)\}$
			\item $D[t,\{v\}]$ = $\{(w_v, p_v)\}$, if $w_v \leq s$
		\end{enumerate} 
		
		\item \textbf{Introduce node.} Suppose $t$ is an introduce node. Then it has only one child $t'$ where $X_{t'} \subset X_{t}$ and there is exactly one vertex $u \neq v$ that belongs to $X_{t}$ but not $X_{t'}$. Then for all $\SS \subseteq X_t$:
		\begin{enumerate}
			\item If $u \not\in \SS$, 
			\begin{enumerate}
				\item if \(\exists w: w\in N^-(u), w \in \SS \) and \(N^+(w) \cap \VV_t = \{u\}\), then $D[t,\SS] = \emptyset$, hence assume otherwise.
				\item Otherwise, we copy all pairs of $D[t',\SS^\pr = \SS \setminus \{u\}]$ to $D[t',\SS,]$. 
			\end{enumerate}
			
			\item If $u \in \SS$,
			\begin{enumerate}            
				\item if $\VV_t \cap N^+_{\GG_t}(u) \subseteq \SS$, then for every pair $(w,p)$ in $D[t',\SS^\pr=\SS \setminus \{u\}]$ with $w + w_u \leq s$, we add $(w+w_u, p + p_u)$ and copy to $D[t,\SS]$.
				
				\item Otherwise, $D[t,\SS]$ = $\emptyset$
				
			\end{enumerate}
		\end{enumerate}
		
		\item \textbf{Forget node.} Suppose $t$ is a forget vertex node. Then it has only one child $t'$, and there is a vertex $u \neq v$ such that $X_t\cup\{u\} = X_{t'}$. Then for all $\SS \subseteq X_t$ we copy all feasible undominated pairs stored in $D[t',\SS\cup \{u\}]$ and $D[t',\SS]$ to $D[t,\SS]$. We remove any dominated pair and maintain only the undominated pairs in $D[t,\SS]$.

		\item \textbf{Join node.} Suppose $t$ is a join node. Then it has two children $t_1,t_2$ such that $X_t = X_{t_1} = X_{t_2}$. Then for all $\SS \subseteq X_t$, let $(w(\SS), p(\SS))$ be the total weight and value of the vertices in $\SS$. Consider a pair $(w_1,p_1)$ in $D[t_1,\SS_1]$ and a pair $(w_2,p_2)$ in $D[t_2,\SS_2]$, where $\SS_1 \cap \SS_2 = \SS$. Then for all pairs $(w_1,p_1) \in D[t_1,\SS_1]$ and $(w_2,p_2) \in D[t_2,\SS_2]$, if $w_1 + w_2 - w(\SS_1 \cap \SS_2) \leq s$, then we add $ \left( w_1 + w_2 - w(\SS_1 \cap \SS_2), p_1+p_2- p(\SS_1 \cap \SS_2) \right)$ and copy to $D[t,\SS]$.
	\end{enumerate}
	
	Finally, in the last step of updating $D[t,\SS]$, we go through the list saved in $D[t,\SS]$ and only keep undominated pairs.

	The output of the algorithm is a pair $(w,p)$ that is maximum of those stored in $D[r,\{v\}]$ such that $w \leq s$ and $p$ is the maximum value over all pairs in $D[r,\{v\}]$.

	\textbf{Proof of correctness}: 
	Recall that we are looking for a solution that is a set of vertices $\VV^\pr$ that contains the fixed vertex $v$ that belongs to all bags of the nice tree decomposition. In each state we maintain a list of feasible and undominated (weight, profit) pairs that correspond to the solution $\hat{\SS}$ of $\GG_t$ and  satisfy \( \hat{\SS} \cap X_t = \SS \).
	
	We now show that the update rules holds for each $X_t$. To prove this formally, we need to consider the cases of what $t$ can be:
	
	\begin{enumerate}
		\item \textbf{Leaf node.} Recall that in our modified nice tree decomposition we have added a vertex $v$ to all the bags. Suppose a leaf node $t$ contains a single vertex $v$, $D[t,\phi]$ stores the pair $\{(0, 0)\}$ and $D[t,\{v\}]$ = $\{(w_v, p_v)\}$ if $w_v \leq s$. This is true in particular when $j = L$, the base case. From now we can assume that for a node $t$ with $\ell(t) = j < L$ and all subsets $\SS \subseteq X_t$, $D[t',\SS^\pr]$ entries are correct and correspond to a \hand solution in $\GG_t$. when $\ell(t') > j$. 
		
		\item \textbf{Introduce node.} When $t$ is an introduce node, there is a child $t'$ such that $X_t = X_{t'} \cup \{u\}$. We are introducing a vertex $u$ and the edges associated with it in $\GG_t$. Let us prove for each case.
		\begin{enumerate}
			\item When $u$ is not included in $\SS$, the families of sets $\hat{\SS}$ considered in the definition of $D[t,\SS]$ and of $D[t',\SS]$ are equal. Notice that there may be a vertex $w \in X_{t}$ that belongs to $\SS$ such that $N^+(w) \cap V_{t'} = \emptyset$, i.e. $w$ did not have any out-neighbor in the subgraph $\GG_{t'}$. However, with the introduction of the vertex $u$ and the edges associated to it, if $N^+(w) \cap V_{t} = \{u\}$ i.e. $u$ is the only out-neighbor of $w$ in $\GG_t$, since $u$ does not belong to $\SS$, the \hand constraint is violated and hence the state $D[t,\SS]$ stores $\emptyset$ as it is no longer feasible. Otherwise we copy all pairs of $D[t',\SS^\pr = \SS]$ to $D[t,\SS]$.
			
			\item Now consider the case when $u$ is part of $\SS$. Let $\hat{\SS}$ be a feasible solution for \hand attained in the definition of $D[t,\SS]$. Then it follows that $\hat{\SS} \setminus \{u\}$ is one of the sets considered in the definition of $D[t',\SS \setminus \{u\}]$. It suffices to just check if $u$ has a neighbor in $X_t$ or not because the nice tree decomposition ensures that the newly introduced vertex $u$ can have neighbors only in $X_t$. There are two subcases:
			\begin{enumerate}
				\item \textbf{Subcase 1}. When $u$ has all out-neighbors $x$ selected in $\SS$ or $u$ does not have any out-neighbor in $\GG_{t}$, then for every pair $(w,p)$ in $D[t',\SS^\pr=\SS \setminus \{u\}]$ with $w + w_u \leq s$, we add $(w+w_u, p + p_u)$ and copy to $D[t,\SS]$.
				
				\item \textbf{Subcase 2}. When $u$ has a out-neighbor in $X_t \setminus \SS$, then $D[t,\SS] = \emptyset$.
			\end{enumerate}
		\end{enumerate}
		Since $\ell(t') > \ell(t)$, by induction hypothesis all entries in $D[t', \SS'= \SS]$ and $D[t', \SS'= \SS \setminus \{u\}]$,  $\forall$ $\SS' \subseteq X_{t'}$ are already computed. We update pairs in $D[t,\SS]$ depending on the cases discussed above.
		
		\item \textbf{Forget Node.}  When $t$ is a forget node, there is a child $t'$ such that $X_t = X_{t'} \setminus \{u\} $. Let $\hat{\SS}$ be a set for which the \hand solution is attained in the definition of $D[t,\SS]$. If $u \not\in \hat{\SS}$, then $\hat{\SS}$ is one of the sets considered in the definition of $D[t',\SS]$. And if $u \in \hat{\SS}$, then $\hat{\SS}$ is one of the sets considered in the definition of $D[t',\SS \cup \{u\}]$. Since $\ell(t') > \ell(t)$, by induction hypothesis all entries in $D[t',\SS' = \SS]$ and $D[t',\SS' = \SS \cup \{u\}]$, $\forall$ $\SS' \subseteq X_{t'}$ are already computed and feasible. We copy each undominated $(w,p)$ pair stored in $D[t',\SS' = \SS]$ and $D[t',\SS' = \SS \cup \{u\}]$ to $D[t,\SS]$.
		
		\item \textbf{Join node.}  When $t$ is a join node, there are two children $t_1$ and $t_2$ of $t$, such that $X_t = X_{t_1} = X_{t_2}$. Let $\hat{\SS}$ be a set for \hand attained in the definition of $D[t,\SS]$. Let $\hat{\SS}_1 = \hat{\SS} \cap V_{t_1}$ and $\hat{\SS}_2 = \hat{\SS} \cap V_{t_2}$. Observe that $\hat{\SS}_1$ is a solution to \hand in $\GG_{t_1}$ and $\hat{\SS}_1 \cap X_{t_1} = \SS_1$, so this is considered in the definition of $D[t_1, \SS]$ and similarly, $\hat{\SS}_2 \cap X_{t_2} = \SS_2$. From the definition of nice tree decomposition we know that there is no edge between the vertices of $V_{t_1} \setminus X_t$ and $V_{t_2} \setminus X_t$. Then we merge solutions from the two subgraphs and remove the over-counting. By the induction hypothesis, the computed entries in $D[t_1,\SS_1]$ and $D[t_2,\SS_2]$ where $\SS_1 \cap \SS_2 = \SS$ are correct and store the feasible and undominated \hand solutions for the subgraph $G_{t_1}$ in $\SS_1$ and similarly, $\SS_2$ for $G_{t_2}$. Now we add ($w_1 + w_2 - w(S_1 \cap S_2),  p_1+p_2 -p(\SS_1 \cap \SS_2)$) to $D[t,\SS]$.
		
	\end{enumerate}
	
	What remains to be shown is that an undominated feasible solution $\VV^\pr$ of \hand in $\GG$ is contained in $D[r,\{v\}]$. Let $w$ be the weight of $\VV^\pr$ and $p$ be the value subject to \hand. Recall that $v \in \VV^\pr$. For each $t$, we consider the subgraph $\GG_t \cap \VV^\pr$. Since the DP state $D[t,\SS]$ is updated correctly for all subgraphs $\GG_t$, the bottom up dynamic programming approach ensures that all feasible and undominated pairs are correctly propagated. If $\VV^\pr$ is a valid solution, then the corresponding (weight, profit) must be stored in some DP state. Since $v \in \VV^\pr$, the optimal states where $v$ is selected will store the pair $(w,p)$. Therefore, $D[r,\{v\}]$ contains the pair $(w,p)$. 
	
	\textbf{Running time}: There are $n$ choices for the fixed vertex $v$. Upon fixing $v$ and adding it to each bag of $(\mathbb{T}, \mathcal{X})$ we consider the total possible number of states. For every node $t$, we have $2^{|X_t|}$ choices of ${\SS}$. For each state, for each $w$, there can be at most one pair with $w$ as the first coordinate; similarly, for each $p$, there can be at most one pair with $p$ as the second coordinate. Thus, the number of undominated pairs in each $D[t,\SS]$ is at most ${\sf min}\{s,d\}$ time. Since the treewidth of the input graph \GG is at most \tw, it is possible to construct a data structure in time $\tw^{\OO(1)} \cdot n$ that allows performing adjacency queries in time $\OO(\tw)$.  For each node $t$, it takes time $\OO\left(2^{\tw} \cdot \tw^{\OO(1)} \cdot {\sf min}\{s^2,d^2\}\right)$ to compute all the values $D[t,\SS]$ and remove all undominated pairs. Since we can assume w.l.o.g that the number of nodes of the given tree decompositions is $\OO(\tw \cdot n)$, and there are $n$ choices for the vertex $v$, the running time of the algorithm is $\OO\left(2^{\tw}\cdot n^{\OO(1)} \cdot {\sf min}\{s^2,d^2\}\right)$.
\end{proof}

\begin{proof}[Proof of \Cref{sor-di-fpt}]
	We consider a nice tree decomposition of the underlying graph. Let $(\GG = (V,E),\{w_u\}_{u \in V_G}, \{p_u\}_{u \in V_G}, s,d)$ be an input instance of \sor such that $\tw=tw(\GG)$. Let $\VV^\pr$ be an optimal solution to \sor. For technical purposes, we guess a vertex $v \in \VV^\pr$ --- once the guess is fixed, we are only interested in finding solution subsets $\hat{\VV}$ that contain $v$ and $\VV^\pr$ is one such candidate. We also consider a nice edge tree decomposition $(\mathbb{T} = (V_{\mathbb{T}},E_{\mathbb{T}}),\mathcal{X})$ of $\GG$ that is rooted at a node $r$, and where $v$ has been added to all bags of the decomposition. Therefore, $X_{r} = \{v\}$ and each leaf bag is the singleton set $\{v\}$.
	
	We define a function $\ell: V_{\mathbb{T}} \rightarrow \mathbb{N}$ as follows.
	For a vertex $t \in V_\mathbb{T}$, $\ell(t) = \sf{dist}_\mathbb{T}(t,r)$, where $r$ is the root. Note that this implies that $\ell(r) = 0$. Let us assume that the values that $\ell$ take over the nodes of $\mathbb{T}$ are between $0$ and $L$. For a node $t\in V_\mathbb{T}$, we denote the set of vertices in the bags in the subtree rooted at $t$ by $V_t$ and $\GG_t=\GG[V_t]$. Now, we describe a dynamic programming algorithm over $(\mathbb{T},\mathcal{X})$. We have the following states.

	\textbf{States:} We maintain a DP table $D$ where a state has the following components:
	\begin{enumerate}
		\item $t$ represents a node in $V_\mathbb{T}$.
		\item $\SS$ represents a subset of the vertex subset $X_t$
		\item $\PP$ represents a subset of the vertex subset $\SS$ 
	\end{enumerate}
	\textbf{Interpretation of States}
	For each node \( t \in \mathbb{T} \), we maintain a list \( D[t, \SS, \PP] \) for every subset \( \SS \subseteq X_t \) and \( \PP \subseteq \SS \). Each entry in this list holds a set of undominated (weight, profit) pairs corresponding to valid \sor solutions $\hat{\SS}$ in the graph \( \GG_t \), where:  
	
	\begin{itemize}  
		\item The selected vertices in \( \GG_t \) satisfy \(\hat{\SS} \cap X_t = \SS \).  
		\item The set of profit-contributing vertices is \( \PP \), meaning all vertices in \( \PP \) are part of \( \SS \) and each has at least one selected out-neighbor in \( \GG_t \) if it has a out-neighbor in \( \GG_t \).  
		\item If no such valid solution exists, we set \( D[t, \SS, \PP] = \emptyset \).  
	\end{itemize}  
	
	For each state $D[t,\SS,\PP]$, we initialize $D[t,\SS,\PP]$ to the list $\{(0,0)\}$.
	
	\textbf{Dynamic Programming on $D$}: We update the table $D$ as follows. We initialize the table of states with nodes $t\in V_\mathbb{T}$ such that $\ell(t)=L$. When all such states are initialized, then we move to update states where the node $t$ has $\ell(t) = L-1$, and so on, till we finally update states with $r$ as the node --- note that $\ell(r) =0$. For a particular $j$, $0\leq j< L$ and a state $D[t,\SS,\PP]$ such that $\ell(t) = j$, we can assume that $D[t',\SS',\PP']$ have been computed for all $t'$ such that $\ell(t')>j$ and all subsets $\SS'$ of $X_{t'}$. Now we consider several cases by which $D[t,\SS,\PP]$ is updated based on the nature of $t$ in $\mathbb{T}$:
	
	\begin{enumerate}
		\item \textbf{Leaf node.} Suppose $t$ is a leaf node with $X_{t} = \{v\}$. Then the list stored in $D[t,\SS,\PP]$ depends on the sets $\SS$ and $\PP$.  If $\PP$ is not a subset of $\SS$, we store $D[t,\SS,\PP] = \emptyset$; hence assume otherwise. The only possible cases that can arise are as follows:
		\begin{enumerate}
			\item $D[t,\phi,\phi]$ stores the pair $\{(0, 0)\}$
			\item $D[t,\{v\},\{v\}]$ = $\{(w_v, p_v)\}$, if $w_v \leq s$
			\item $D[t,\{v\},\phi]$ stores the pair $\{(w_v, 0)\}$, if $w_v \leq s$
		\end{enumerate} 
		
		\item \textbf{Introduce node.} Suppose $t$ is an introduce node. Then it has only one child $t'$ where $X_{t'} \subset X_{t}$ and there is exactly one vertex $u \neq v$ that belongs to $X_{t}$ but not $X_{t'}$. Then for all $\SS \subseteq X_t$ and all $\PP \subseteq X_t$ , if $\PP$ is not a subset of $\SS$, we store $D[t,\SS,\PP] = \emptyset$; hence assume otherwise.
		\begin{enumerate}
			\item If $u \not\in \SS$ and therefore $u \not\in \PP$, 
			\begin{enumerate}
				\item if \(\exists w: w\in N^-(u), w \in \PP \) and \(N^+(w) \cap \VV_t = \{u\}\), then $D[t,\SS,\PP] = \emptyset$.
				\item Otherwise, we copy all pairs of $D[t',\SS^\pr = \SS,\PP^\pr=\PP]$ to $D[t',\SS,\PP]$. 
			\end{enumerate}
			
			\item If $u \in \SS$,
			\begin{enumerate}            
				\item if $u$ has out-neighbor $x$ in $\SS$ or \(N^+(u) \cap \VV_t = \emptyset\), then $u \in \PP$, and $\PP = \PP^\pr \cup \{u\}$,
				and for every pair $(w,p)$ in $D[t',\SS^\pr=\SS \setminus \{u\},\PP^\pr=\PP \setminus \{u\}]$ with $w + w_u \leq s$, we add $(w+w_u, p + \sum_{z \in \SS: (z,u) \in \GG, z \not\in \PP^\pr} p_z + p_u)$ and copy to $D[t,\SS,\PP]$.
				
				\item if $u$ does not have a out-neighbor $x$ in $\SS$ and \(N^+(u) \cap \VV_t \neq \emptyset\), then $u \not\in \PP$, and $\PP = \PP^\pr$,
				then for every pair $(w,p)$ in $D[t',\SS^\pr=\SS \setminus \{u\},\PP^\pr=\PP \setminus \{u\}]$ with $w + w_u \leq s$, we add $(w+w_u, p + \sum_{z \in \SS: (z,u) \in \GG, z \not\in \PP^\pr} p_z )$ and copy to $D[t,\SS,\PP]$.
				
			\end{enumerate}
		\end{enumerate}
		
		\item \textbf{Forget node.} Suppose $t$ is a forget vertex node. Then it has only one child $t'$, and there is a vertex $u \neq v$ such that $X_t\cup\{u\} = X_{t'}$. Then for all $\SS \subseteq X_t$ and all $\PP \subseteq X_t$ , if $\PP$ is not a subset of $\SS$, we store $D[t,\SS,\PP] = \emptyset$; hence assume otherwise. For all $\SS \subseteq X_t$, we copy all feasible undominated pairs stored in $D[t',\SS\cup \{u\},\PP\cup \{u\}]$, $D[t',\SS\cup \{u\},\PP]$, and $D[t',\SS,\PP]$ to $D[t,\SS,\PP]$. If any pair stored in $D[t',\SS\cup \{u\},\PP]$, and $D[t',\SS,\PP]$ is dominated by any pair of $D[t',\SS\cup \{u\},\PP\cup \{u\}]$, we copy only the undominated pairs to $D[t,\SS,\PP]$.

		\item \textbf{Join node.} Suppose $t$ is a join node. Then it has two children $t_1,t_2$ such that $X_t = X_{t_1} = X_{t_2}$. Then for all $\SS \subseteq X_t$, let $(w(\SS), p(\SS))$ be the total weight and value of the vertices in $\SS$. Consider a pair $(w_1,p_1)$ in $D[t_1,\SS_1,\PP_1]$ and a pair $(w_2,p_2)$ in $D[t_2,\SS_2,\PP_2]$, where $\SS_1 \cup \SS_2 = \SS$ and $\PP_1 \cup \PP_2 = \PP$. Then for all pairs $(w_1,p_1) \in D[t_1,\SS_1,\PP_1]$ and $(w_2,p_2) \in D[t_2,\SS_2,\PP_2]$, if $w_1 + w_2 - w(\SS_1 \cup \SS_2) \leq s$, then we add $ \left( w_1 + w_2 - w(\SS_1 \cup \SS_2), p_1+p_2- p(\PP_1 \cap \PP_2) \right)$ and copy to $D[t,\SS,\PP]$.
	\end{enumerate}
	
	Finally, in the last step of updating $D[t,\SS,\PP]$, we go through the list saved in $D[t,\SS,\PP]$ and only keep undominated pairs.

	The output of the algorithm is a pair $(w,p)$ that is maximum of those stored in $D[r,\{v\},\{v\}]$ and $D[r,\{v\},\emptyset]$ such that $w \leq s$ and $p$ is the maximum value over all pairs in $D[r,\{v\},\{v\}]$ and $D[r,\{v\},\emptyset]$.

	\textbf{Proof of correctness}: 
	Recall that we are looking for a solution that is a set of vertices $\VV^\pr$ that contains the fixed vertex $v$ that belongs to all bags of the nice tree decomposition. In each state we maintain a list of feasible and undominated (weight, profit) pairs that correspond to the solution $\hat{\SS}$ of $\GG_t$ and  satisfy \( \hat{\SS} \cap X_t = \SS \) and the set of profit contributing vertices is $\PP$.
	
	We now show that the update rules holds for each $X_t$. To prove this formally, we need to consider the cases of what $t$ can be:
	
	\begin{enumerate}
		\item \textbf{Leaf node.} Recall that in our modified nice tree decomposition we have added a vertex $v$ to all the bags. Suppose a leaf node $t$ contains a single vertex $v$, $D[t,\phi,\phi]$ stores the pair $\{(0, 0)\}$, $D[t,\{v\},\{v\}]$ = $\{(w_v, p_v)\}$ if $w_v \leq s$, $D[t,\{v\},\phi]$ stores the pair $\{(w_v, 0)\}$ if $w_v \leq s$, otherwise $\emptyset$. This is true in particular when $j = L$, the base case. From now we can assume that for a node $t$ with $\ell(t) = j < L$ and all subsets $\SS, \PP \subseteq X_t$, $D[t',\SS^\pr,\PP^\pr]$ entries are correct and correspond to a \sor solution in $\GG_t$. when $\ell(t') > j$. 
		
		\item \textbf{Introduce node.} When $t$ is an introduce node, there is a child $t'$ such that $X_t = X_{t'} \cup \{u\}$. We are introducing a vertex $u$ and the edges associated with it in $\GG_t$. When $\PP$ is not a subset of $\SS$, the state $D[t,\SS,\PP] = \emptyset$ because it violates the definition of $\PP$; hence assume otherwise. Let us prove for each case.
		\begin{enumerate}
			\item When $u$ is not included in $\SS$, it is also not in $\PP$ and the families of sets $\hat{\SS}$ considered in the definition of $D[t,\SS,\PP]$ and of $D[t',\SS,\PP]$ are equal. Notice that there may be a vertex $w \in X_{t}$ that belongs to $\PP$ such that $N^+(w) \cap V_{t'} = \emptyset$, i.e. $w$ did not have any out-neighbor in the graph $\GG_{t'}$. However, with the introduction of the vertex $u$ and the edges associated to it, if $N^+(w) \cap V_{t} = \{u\}$ i.e. $u$ is the only out-neighbor of $w$ in $\GG_t$, since $u$ does not belong to $\SS$, the \sor constraint is violated and hence the state $D[t,\SS,\PP]$ stores $\emptyset$ as it is no longer feasible. Otherwise we copy all pairs of $D[t',\SS^\pr = \SS,\PP^\pr=\PP]$ to $D[t,\SS,\PP]$.
			
			\item Now consider the case when $u$ is part of $\SS$. Let $\hat{\SS}$ be a feasible solution for \sor attained in the definition of $D[t,\SS,\PP]$. Then it follows that $\hat{\SS} \setminus \{u\}$ is one of the sets considered in the definition of $D[t',\SS \setminus \{v\},\PP \setminus \{v\}]$. Notice that all $z$ in $\SS' \setminus \PP'$ such that $z \in N^-(u)$ get activated or become profitable due to inclusion of $u$. Thus the edge $(z,u)$ in $\GG_t$ is a witness that $z$ is profitable. Now we must compute entries depending on whether $u$ contributes to profit or not. It suffices to just check if $u$ has a neighbor in $X_t$ or not because the nice tree decomposition ensures that the newly introduced vertex $u$ can have neighbors only in $X_t$. There are two subcases:
			\begin{enumerate}
				\item \textbf{Subcase 1: $u\in \PP$}. When $u$ has at least one out-neighbor $x$ selected in $\SS$ or $u$ does not have any out-neighbor in $\GG_{t}$, then $u \in \PP$ and $\PP = \PP^\pr \cup \{u\}$, and for every pair $(w,p)$ in $D[t',\SS^\pr=\SS \setminus \{u\},\PP^\pr=\PP \setminus \{u\}]$ with $w + w_u \leq s$, we add $(w+w_u, p + \sum_{z \in \SS: (z,u) \in \GG, z \not\in \PP^\pr} p_z + p_u)$ and copy to $D[t,\SS,\PP]$.
				
				\item \textbf{Subcase 2: $u \notin \PP$}. When $u$ has a out-neighbor in $\GG_t$ but none of them are selected in $\SS$, then $u$ does not become profitable and thus does not belong to $\PP$. For every pair $(w,p)$ in $D[t',\SS^\pr=\SS \setminus \{u\},\PP^\pr=\PP \setminus \{u\}]$ with $w + w_u \leq s$, we add $(w+w_u, p + \sum_{z \in \SS: (z,u) \in \GG, z \not\in \PP^\pr} p_z )$ and copy to $D[t,\SS,\PP]$.
			\end{enumerate}
		\end{enumerate}
		Since $\ell(t') > \ell(t)$, by induction hypothesis all entries in $D[t', \SS'= \SS,\PP'=\PP]$, $D[t', \SS'= \SS \setminus \{u\},\PP'=\PP]$, $D[t', \SS'= \SS,\PP'=\PP \setminus \{u\}]$ and $D[t', \SS'= \SS \setminus \{u\},\PP'=\PP \setminus \{u\}]$ $\forall$ $\SS', \PP' \subseteq X_{t'}$ are already computed. We update pairs in $D[t,\SS,\PP]$ depending on the cases discussed above.
		
		\item \textbf{Forget Node.}  When $t$ is a forget node, there is a child $t'$ such that $X_t = X_{t'} \setminus \{u\} $. Let $\hat{\SS}$ be a set for which the \sor solution is attained in the definition of $D[t,\SS,\PP]$. If $u \not\in \hat{\SS}$, then $\hat{\SS}$ is one of the sets considered in the definition of $D[t',\SS,\PP]$. And if $u \in \hat{\SS}$, then $\hat{\SS}$ is one of the sets considered in the definition of $D[t',\SS \cup \{u\},\PP]$ and $D[t',\SS \cup \{u\},\PP \cup \{u\}]$. Since $\ell(t') > \ell(t)$, by induction hypothesis all entries in $D[t',\SS' = \SS,\PP' = \PP]$,  $D[t',\SS' = \SS \cup \{u\},\PP' = \PP]$, and $D[t',\SS' = \SS\cup \{u\} ,\PP' = \PP\cup \{u\}]$ $\forall$ $\SS', \PP' \subseteq X_{t'}$ are already computed and feasible. We copy each undominated $(w,p)$ pair stored in $D[t',\SS' = \SS,\PP' = \PP]$,  $D[t',\SS' = \SS \cup \{u\},\PP' = \PP]$, and $D[t',\SS' = \SS\cup \{u\} ,\PP' = \PP\cup \{u\}]$ to $D[t,\SS,\PP]$.
		
		\item \textbf{Join node.}  When $t$ is a join node, there are two children $t_1$ and $t_2$ of $t$, such that $X_t = X_{t_1} = X_{t_2}$. Let $\hat{\SS}$ be a set for \sor attained in the definition of $D[t,\SS,\PP]$. Let $\hat{\SS}_1 = \hat{\SS} \cap V_{t_1}$ and $\hat{\SS}_2 = \hat{\SS} \cap V_{t_2}$. Observe that $\hat{\SS}_1$ is a solution to \sor in $\GG_{t_1}$ and $\hat{\SS}_1 \cap X_{t_1} = \SS_1$, so this is considered in the definition of $D[t_1, \SS, \PP]$ and similarly, $\hat{\SS}_2 \cap X_{t_2} = \SS_2$. From the definition of nice tree decomposition we know that there is no edge between the vertices of $V_{t_1} \setminus X_t$ and $V_{t_2} \setminus X_t$. Then we merge solutions from the two subgraphs and remove the over-counting. By the induction hypothesis, the computed entries in $D[t_1,\SS_1, \PP_1]$ and $D[t_2,\SS_2, \PP_2]$ where $\SS_1 \cup \SS_2 = \SS$ and $\PP_1 \cup \PP_2 = \PP$ are correct and store the feasible and undominated \sor solutions for the subgraph $G_{t_1}$ in $\SS_1$ and similarly, $\SS_2$ for $G_{t_2}$. Now we add ($w_1 + w_2 - w(\SS_1 \cup \SS_2),  p_1+p_2 -p(\PP_1 \cap \PP_2)$) to $D[t,\SS, \PP]$.
		
	\end{enumerate}
	
	What remains to be shown is that an undominated feasible solution $\VV^\pr$ of \sor in $\GG$ is contained in $D[r,\{v\},\{v\}] \cup D[t,\{v\}, \emptyset]$. Let $w$ be the weight of $\VV^\pr$ and $p$ be the value subject to \sor. Recall that $v \in \VV^\pr$. For each $t$, we consider the subgraph $\GG_t \cap \VV^\pr$. Since the DP state $D[t,\SS,\PP]$ is updated correctly for all subgraphs $\GG_t$, the bottom up dynamic programming approach ensures that all feasible and undominated pairs are correctly propagated. If $\VV^\pr$ is a valid solution, then the corresponding (weight, profit) must be stored in some DP state. Since $v \in \VV^\pr$, the optimal states where $v$ is selected will store the pair $(w,p)$. Therefore, $D[r,\{v\},\{v\}] \cup D[t,\{v\}, \emptyset]$ contains the pair $(w,p)$. 
	
	\textbf{Running time}: There are $n$ choices for the fixed vertex $v$. Upon fixing $v$ and adding it to each bag of $(\mathbb{T}, \mathcal{X})$ we consider the total possible number of states. For every node $t$, we have $2^{|X_t|}$ choices of ${\SS}$ and $2^{|X_t|}$ choices of ${\PP}$ for each choice of $\SS$. For each state, for each $w$, there can be at most one pair with $w$ as the first coordinate; similarly, for each $p$, there can be at most one pair with $p$ as the second coordinate. Thus, the number of undominated pairs in each $D[t,\SS,\PP]$ is at most ${\sf min}\{s,d\}$ time. Since the treewidth of the input graph \GG is at most \tw, it is possible to construct a data structure in time $\tw^{\OO(1)} \cdot n$ that allows performing adjacency queries in time $\OO(\tw)$.  For each node $t$, it takes time $\OO\left(4^{\tw} \cdot \tw^{\OO(1)} \cdot {\sf min}\{s^2,d^2\}\right)$ to compute all the values $D[t,\SS,\PP]$ and remove all undominated pairs. Since we can assume w.l.o.g that the number of nodes of the given tree decompositions is $\OO(\tw \cdot n)$, and there are $n$ choices for the vertex $v$, the running time of the algorithm is $\OO\left(4^{\tw}\cdot n^{\OO(1)} \cdot {\sf min}\{s^2,d^2\}\right)$.
\end{proof}

\begin{proof}[Proof of \Cref{sand-di-fpt}]
	We consider a nice tree decomposition of the underlying graph. Let $(\GG = (V,E),\{w_u\}_{u \in V_G}, \{p_u\}_{u \in V_G}, s,d)$ be an input instance of \sor such that $\tw=tw(\GG)$. Let $\VV^\pr$ be an optimal solution to \sand. For technical purposes, we guess a vertex $v \in \VV^\pr$ --- once the guess is fixed, we are only interested in finding solution subsets $\hat{\VV}$ that contain $v$ and $\VV^\pr$ is one such candidate. We also consider a nice edge tree decomposition $(\mathbb{T} = (V_{\mathbb{T}},E_{\mathbb{T}}),\mathcal{X})$ of $\GG$ that is rooted at a node $r$, and where $v$ has been added to all bags of the decomposition. Therefore, $X_{r} = \{v\}$ and each leaf bag is the singleton set $\{v\}$.
	
	We define a function $\ell: V_{\mathbb{T}} \rightarrow \mathbb{N}$ as follows.
	For a vertex $t \in V_\mathbb{T}$, $\ell(t) = \sf{dist}_\mathbb{T}(t,r)$, where $r$ is the root. Note that this implies that $\ell(r) = 0$. Let us assume that the values that $\ell$ take over the nodes of $\mathbb{T}$ are between $0$ and $L$. For a node $t\in V_\mathbb{T}$, we denote the set of vertices in the bags in the subtree rooted at $t$ by $V_t$ and $\GG_t=\GG[V_t]$. Now, we describe a dynamic programming algorithm over $(\mathbb{T},\mathcal{X})$. We have the following states.

	\textbf{States:} We maintain a DP table $D$ where a state has the following components:
	\begin{enumerate}
		\item $t$ represents a node in $V_\mathbb{T}$.
		\item $\SS$ represents a subset of the vertex subset $X_t$
		\item $\PP$ represents a subset of the vertex subset $\SS$ 
	\end{enumerate}
	\textbf{Interpretation of States}
	For each node \( t \in \mathbb{T} \), we maintain a list \( D[t, \SS, \PP] \) for every subset \( \SS \subseteq X_t \) and \( \PP \subseteq \SS \). Each entry in this list holds a set of feasible undominated (weight, profit) pairs corresponding to valid \sand solutions $\hat{\SS}$ in the graph \( \GG_t \), where:  
	
	\begin{itemize}  
		\item The selected vertices in \( \GG_t \) satisfy \(\hat{\SS} \cap X_t = \SS \). 
		\item The set of profit-contributing vertices is \( \PP \), meaning all vertices in \( \PP \) are part of \( \SS \) and each has selected all out-neighbors in \( \GG_t \).  
		\item If no such valid solution exists, we set \( D[t, \SS, \PP] = \emptyset \).  
	\end{itemize}  
	
	For each state $D[t,\SS,\PP]$, we initialize $D[t,\SS,\PP]$ to the list $\{(0,0)\}$.
	
	\textbf{Dynamic Programming on $D$}: We update the table $D$ as follows. We initialize the table of states with nodes $t\in V_\mathbb{T}$ such that $\ell(t)=L$. When all such states are initialized, then we move to update states where the node $t$ has $\ell(t) = L-1$, and so on, till we finally update states with $r$ as the node --- note that $\ell(r) =0$. For a particular $j$, $0\leq j< L$ and a state $D[t,\SS,\PP]$ such that $\ell(t) = j$, we can assume that $D[t',\SS',\PP']$ have been computed for all $t'$ such that $\ell(t')>j$ and all subsets $\SS'$ of $X_{t'}$. Now we consider several cases by which $D[t,\SS,\PP]$ is updated based on the nature of $t$ in $\mathbb{T}$:
	
	\begin{enumerate}
		\item \textbf{Leaf node.} Suppose $t$ is a leaf node with $X_{t} = \{v\}$. Then the list stored in $D[t,\SS,\PP]$ depends on the sets $\SS$ and $\PP$.  If $\PP$ is not a subset of $\SS$, we store $D[t,\SS,\PP] = \emptyset$; hence assume otherwise. The only possible cases that can arise are as follows:
		\begin{enumerate}
			\item $D[t,\phi,\phi]$ stores the pair $\{(0, 0)\}$
			\item $D[t,\{v\},\{v\}]$ = $\{(w_v, p_v)\}$, if $w_v \leq s$ 
			\item $D[t,\{v\},\phi]$ stores the pair $\{(w_v, 0)\}$, if $w_v \leq s$
		\end{enumerate} 
		
		\item  \textbf{Introduce node.} Suppose $t$ is an introduce node. Then it has only one child $t'$ where $X_{t'} \subset X_{t}$ and there is exactly one vertex $u \neq v$ that belongs to $X_{t}$ but not $X_{t'}$. Then for all $\SS \subseteq X_t$ and all $\PP \subseteq X_t$ , if $\PP$ is not a subset of $\SS$, we store $D[t,\SS,\PP] = \emptyset$; hence assume otherwise. 
		\begin{enumerate}
			\item If $u \not\in \SS$ and therefore $u \not\in \PP$, then 
			\begin{enumerate}
				\item if $\exists w : w \in \PP$ and $w \in N^-(u)$, we store for $D[t,\SS,\PP]$ = $\emptyset$.
				\item Otherwise, we copy all pairs of $D[t',\SS^\pr = \SS,\PP^\pr=\PP]$ to $D[t',\SS,\PP]$.
			\end{enumerate}
			
			\item If $u \in \SS$,
			\begin{enumerate}
				\item if $u$ does not have a out-neighbor $x$ in $X_t \setminus \SS$ i.e. \(N^+(u) \cap \VV_t = \emptyset\), then $u \in \PP$, and $\PP = \PP^\pr \cup \{u\}$, then for every pair $(w,p)$ in $D[t',\SS^\pr=\SS \setminus \{u\},\PP^\pr=\PP \setminus \{u\}]$ with $w + w_u \leq s$, we add $(w+w_u, p + p_u)$ and copy to $D[t,\SS,\PP]$.
				
				\item if $u$ has a out-neighbor $x$ in $X_t \setminus \SS$, then $u \not\in \PP$, and $\PP = \PP^\pr$, then for every pair $(w,p)$ in $D[t',\SS^\pr=\SS \setminus \{u\},\PP^\pr=\PP \setminus \{u\}]$ with $w + w_u \leq s$, we add $(w+w_u, p + 0)$ and copy to $D[t,\SS,\PP]$.
			\end{enumerate}
			\item Otherwise, we store $D[t,\SS,\PP] = \emptyset$
		\end{enumerate}
		
		\item  \textbf{Forget node.} Suppose $t$ is a forget vertex node. Then it has only one child $t'$, and there is a vertex $u \neq v$ such that $X_t\cup\{u\} = X_{t'}$. Then for all $\SS \subseteq X_t$ and all $\PP \subseteq X_t$ , if $\PP$ is not a subset of $\SS$, we store $D[t,\SS,\PP] = \emptyset$; hence assume otherwise. Then for all $\SS \subseteq X_t$:
		\begin{enumerate}
			\item we copy all feasible undominated pairs stored in $D[t',\SS\cup \{u\},\PP\cup \{u\}]$ to $D[t,\SS,\PP]$ if $u \in \SS^\pr$ and $u \in \PP^\pr$
			\item or we copy all feasible undominated pairs stored in $D[t',\SS\cup \{u\},\PP]$ to $D[t,\SS,\PP]$ if $u \in \SS^\pr$ and $u \not\in \PP^\pr$,
			\item or we copy all feasible undominated pairs stored in $D[t',\SS,\PP]$ to $D[t,\SS,\PP]$ otherwise. 
		\end{enumerate}

		\item \textbf{Join node.} Suppose $t$ is a join node. Then it has two children $t_1,t_2$ such that $X_t = X_{t_1} = X_{t_2}$. Then for all $\SS \subseteq X_t$, let $(w(\SS), p(\SS))$ be the total weight and value of the vertices in $\SS$ Consider a pair $(w_1,p_1)$ in $D[t_1,\SS_1,\PP_1]$ and a pair $(w_2,p_2)$ in $D[t_2,\SS_2,\PP_2]$ where $\SS_1 \cup \SS_2 = \SS$ and $\PP_1 \cap \PP_2 = \PP$. Suppose $w_1 + w_2 - w(\SS_1 \cup \SS_2) \leq s$, then we add $ \left( w_1 + w_2 - w(\SS_1 \cup \SS_2), p_1+p_2- p(\PP_1 \cap \PP_2) \right)$ to $D[t,\SS,\PP]$.
	\end{enumerate}
	
	Finally, in the last step of updating $D[t,\SS,\PP]$, we go through the list saved in $D[t,\SS,\PP]$ and only keep undominated pairs. 
	
	The output of the algorithm is a pair $(w,p)$ that is maximum of those stored in $D[r,\{v\},\{v\}]$ and $D[r,\{v\},\emptyset]$ such that $w \leq s$ and $p$ is the maximum value over all pairs in $D[r,\{v\},\{v\}]$ and $D[r,\{v\},\emptyset]$.
	
	\textbf{Proof of correctness}: 
	Recall that we are looking for a solution that is a set of vertices $\VV^\pr$ that contains the fixed vertex $v$ that belongs to all bags of the nice tree decomposition. In each state we maintain a list of feasible and undominated (weight, profit) pairs that correspond to the solution $\hat{\SS}$ of $\GG_t$ and  satisfy \( \hat{\SS} \cap X_t = \SS \) and the set of profit contributing vertices is $\PP$.
	
	We now show that the update rules holds for each $X_t$. To prove this formally, we need to consider the cases of what $t$ can be:
	
	\begin{enumerate}
		\item \textbf{Leaf node.} Recall that in our modified nice tree decomposition we have added a vertex $v$ to all the bags. Suppose a leaf node $t$ contains a single vertex $v$, $D[t,\phi,\phi]$ stores the pair $\{(0, 0)\}$, $D[t,\{v\},\{v\}]$ = $\{(w_v, p_v)\}$ if $w_v \leq s$, $D[t,\{v\},\phi]$ stores the pair $\{(w_v, 0)\}$ if $w_v \leq s$, otherwise $\emptyset$. This is true in particular when $j = L$, the base case. From now we can assume that for a node $t$ with $\ell(t) = j < L$ and all subsets $\SS, \PP \subseteq X_t$, $D[t',\SS^\pr,\PP^\pr]$ entries are correct and correspond to a \sor solution in $\GG_t$. when $\ell(t') > j$.
		
		\item \textbf{Introduce node.} When $t$ is an introduce node, there is a child $t'$ such that $X_t = X_{t'} \cup \{u\}$. We are introducing a vertex $u$ and the edges associated with it in $\GG_t$. When $\PP$ is not a subset of $\SS$, the state $D[t,\SS,\PP] = \emptyset$ because it violates the definition of $\PP$; hence assume otherwise. Let us prove for each case.
		\begin{enumerate}
			\item When $u$ is not included in $\SS$, it is also not in $\PP$ and the families of sets $\hat{\SS}$ considered in the definition of $D[t,\SS,\PP]$ and of $D[t',\SS,\PP]$ are equal. Notice that there may be a vertex $w \in X_{t}$ that belongs to $\PP$ such that $N^+(w) \cap V_{t'} = \emptyset$, i.e. $w$ did not have any out-neighbor in the graph $\GG_{t'}$. However, with the introduction of the vertex $u$ and the edges associated to it, if $N^+(w) \cap V_{t} = \{u\}$ i.e. $u$ is the only out-neighbor of $w$ in $\GG_t$, since $u$ does not belong to $\SS$, the \sand constraint is violated and hence the state $D[t,\SS,\PP]$ stores $\emptyset$ as it is no longer feasible. Otherwise we copy all pairs of $D[t',\SS^\pr = \SS,\PP^\pr=\PP]$ to $D[t,\SS,\PP]$.
			
			\item Now consider the case when $u$ is part of $\SS$. Let $\hat{\SS}$ be a feasible solution for \sand attained in the definition of $D[t,\SS,\PP]$. Then it follows that $\hat{\SS} \setminus \{u\}$ is one of the sets considered in the definition of $D[t',\SS \setminus \{v\},\PP \setminus \{v\}]$. Now we must compute entries depending on whether $u$ contributes to profit or not. It suffices to just check if $u$ has a neighbor in $X_t$ or not because the nice tree decomposition ensures that the newly introduced vertex $u$ can have neighbors only in $X_t$. There are two subcases:
			\begin{enumerate}
				\item \textbf{Subcase 1: $u\in \PP$}. When $u$ has all out-neighbors $x$ selected in $\SS$ or $u$ does not have any out-neighbor in $\GG_{t}$, then $u \in \PP$ and $\PP = \PP^\pr \cup \{u\}$, and then for every pair $(w,p)$ in $D[t',\SS^\pr=\SS \setminus \{u\},\PP^\pr=\PP \setminus \{u\}]$ with $w + w_u \leq s$, we add $(w+w_u, p + p_u)$ and copy to $D[t,\SS,\PP]$.
				
				\item \textbf{Subcase 2: $u \notin \PP$}. When $u$ has at least one out-neighbor in $\GG_t$ that is not selected are selected in $\SS$, then $u$ does not become profitable and thus does not belong to $\PP$. For then for every pair $(w,p)$ in $D[t',\SS^\pr=\SS \setminus \{u\},\PP^\pr=\PP \setminus \{u\}]$ with $w + w_u \leq s$, we add $(w+w_u, p + 0)$ and copy to $D[t,\SS,\PP]$.
			\end{enumerate}
		\end{enumerate}
		Since $\ell(t') > \ell(t)$, by induction hypothesis all entries in $D[t', \SS'= \SS,\PP'=\PP]$, $D[t', \SS'= \SS,\PP'=\PP \setminus \{u\}]$ and $D[t', \SS'= \SS \setminus \{u\},\PP'=\PP \setminus \{u\}]$ $\forall$ $\SS', \PP' \subseteq X_{t'}$ are already computed. We update pairs in $D[t,\SS,\PP]$ depending on the cases discussed above.
		
		\item \textbf{Forget Node.}  When $t$ is a forget node, there is a child $t'$ such that $X_t = X_{t'} \setminus \{u\} $. Let $\hat{\SS}$ be a set for which the \sand solution is attained in the definition of $D[t,\SS,\PP]$. If $u \not\in \hat{\SS}$, then $\hat{\SS}$ is one of the sets considered in the definition of $D[t',\SS,\PP]$. And if $u \in \hat{\SS}$, then $\hat{\SS}$ is one of the sets considered in the definition of $D[t',\SS \cup \{u\},\PP]$ and $D[t',\SS \cup \{u\},\PP \cup \{u\}]$. Since $\ell(t') > \ell(t)$, by induction hypothesis all entries in $D[t',\SS' = \SS,\PP' = \PP]$,  $D[t',\SS' = \SS \cup \{u\},\PP' = \PP]$, and $D[t',\SS' = \SS\cup \{u\} ,\PP' = \PP\cup \{u\}]$ $\forall$ $\SS', \PP' \subseteq X_{t'}$ are already computed and feasible. We copy each undominated $(w,p)$ pair stored in $D[t',\SS' = \SS,\PP' = \PP]$,  $D[t',\SS' = \SS \cup \{u\},\PP' = \PP]$, and $D[t',\SS' = \SS\cup \{u\} ,\PP' = \PP\cup \{u\}]$ to $D[t,\SS,\PP]$ depending on whether $u$ belonged to $\SS$ or $\SS \cup \{u\}$.
		
		\item \textbf{Join node.}  When $t$ is a join node, there are two children $t_1$ and $t_2$ of $t$, such that $X_t = X_{t_1} = X_{t_2}$. Let $\hat{\SS}$ be a set for \sor attained in the definition of $D[t,\SS,\PP]$. Let $\hat{\SS}_1 = \hat{\SS} \cap V_{t_1}$ and $\hat{\SS}_2 = \hat{\SS} \cap V_{t_2}$. Observe that $\hat{\SS}_1$ is a solution to \sand in $\GG_{t_1}$ and $\hat{\SS}_1 \cap X_{t_1} = \SS_1$, so this is considered in the definition of $D[t_1, \SS, \PP]$ and similarly, $\hat{\SS}_2 \cap X_{t_2} = \SS_2$. From the definition of nice tree decomposition we know that there is no edge between the vertices of $V_{t_1} \setminus X_t$ and $V_{t_2} \setminus X_t$. Then we merge solutions from the two subgraphs and remove the over-counting. By the induction hypothesis, the computed entries in $D[t_1,\SS_1, \PP_1]$ and $D[t_2,\SS_2, \PP_2]$ where $\SS_1 \cup \SS_2 = \SS$ and $\PP_1 \cap \PP_2 = \PP$ are correct and store the feasible and undominated \sand solutions for the subgraph $G_{t_1}$ in $\SS_1$ and similarly, $\SS_2$ for $G_{t_2}$. Now we add ($w_1 + w_2 - w(S_1 \cup S_2),  p_1+p_2 -p(\PP_1 \cap \PP_2)$) to $D[t,\SS, \PP]$.
	\end{enumerate}
	
	What remains to be shown is that an undominated feasible solution $\VV^\pr$ of \sand in $\GG$ is contained in $D[r,\{v\},\{v\}] \cup D[t,\{v\}, \emptyset]$. Let $w$ be the weight of $\VV^\pr$ and $p$ be the value subject to \sand. Recall that $v \in \VV^\pr$. For each $t$, we consider the subgraph $\GG_t \cap \VV^\pr$. Since the DP state $D[t,\SS,\PP]$ is updated correctly for all subgraphs $\GG_t$, the bottom up dynamic programming approach ensures that all feasible and undominated pairs are correctly propagated. If $\VV^\pr$ is a valid solution, then the corresponding (weight, profit) must be stored in some DP state. Since $v \in \VV^\pr$, the optimal states where $v$ is selected will store the pair $(w,p)$. Therefore, $D[r,\{v\},\{v\}] \cup D[t,\{v\}, \emptyset]$ contains the pair $(w,p)$. 
	
	\textbf{Running time}: There are $n$ choices for the fixed vertex $v$. Upon fixing $v$ and adding it to each bag of $(\mathbb{T}, \mathcal{X})$ we consider the total possible number of states. For every node $t$, we have $2^{|X_t|}$ choices of ${\SS}$ and $2^{|X_t|}$ choices of ${\PP}$ for each choice of $\SS$. For each state, for each $w$, there can be at most one pair with $w$ as the first coordinate; similarly, for each $p$, there can be at most one pair with $p$ as the second coordinate. Thus, the number of undominated pairs in each $D[t,\SS,\PP]$ is at most ${\sf min}\{s,d\}$ time. Since the treewidth of the input graph \GG is at most \tw, it is possible to construct a data structure in time $\tw^{\OO(1)} \cdot n$ that allows performing adjacency queries in time $\OO(\tw)$.  For each node $t$, it takes time $\OO\left(4^{\tw} \cdot \tw^{\OO(1)} \cdot {\sf min}\{s^2,d^2\}\right)$ to compute all the values $D[t,\SS,\PP]$ and remove all undominated pairs. Since we can assume w.l.o.g that the number of nodes of the given tree decompositions is $\OO(\tw \cdot n)$, and there are $n$ choices for the vertex $v$, the running time of the algorithm is $\OO\left(4^{\tw}\cdot n^{\OO(1)} \cdot {\sf min}\{s^2,d^2\}\right)$.
\end{proof}

}